\documentclass[sisc]{siamart1116}

\usepackage{lipsum}
\usepackage{amsfonts}
\usepackage{mathtools}
\usepackage{graphicx}
\usepackage{epstopdf}
\usepackage{algorithmic}
\ifpdf
  \DeclareGraphicsExtensions{.eps,.pdf,.png,.jpg}
\else
  \DeclareGraphicsExtensions{.eps}
\fi

\newcommand{\boundary}{\partial}

\numberwithin{theorem}{section}

\newsiamremark{remark}{Remark}

\newcommand{\TheTitle}{An Approach to Quad Meshing Based on Harmonic Cross-Valued Maps and the Ginzburg-Landau Theory}
\newcommand{\TheAuthors}{R. Viertel and B. Osting}

\newcommand{\ShortTitle}{Quad Meshing, Cross Fields, and the Ginzburg-Landau Theory}

\headers{\ShortTitle}{\TheAuthors}

\title{{\TheTitle}\thanks{\today
\funding{B. Osting is partially supported by NSF DMS 16-19755 and 17-52202. R. Viertel is supported by Sandia National Laboratories. Sandia National Laboratories is a multi-mission laboratory managed and operated by National Technology \& Engineering Solutions of Sandia, LLC., a wholly owned subsidiary of Honeywell International, Inc., for the U.S. Department of Energy's National Nuclear Security Administration under contract DE-NA0003525. SAND2017-8242 J.}}}

\author{
  Ryan Viertel\thanks{Department of Mathematics, University of Utah, Salt Lake City, UT (\email{viertel@math.utah.edu},
    \email{osting@math.utah.edu}).}
  \and
  Braxton Osting\footnotemark[2]
}

\usepackage{amsopn}

\DeclareMathOperator{\sgn}{sgn}
\DeclareMathOperator*{\interior}{int}
\DeclareMathOperator*{\Index}{I}

\ifpdf
\hypersetup{
  pdftitle={\TheTitle},
  pdfauthor={\TheAuthors}
}
\fi

\usepackage{algorithm}
\usepackage{algorithmic}

 \newtheorem{assumption}[theorem]{Assumption}

\begin{document}

\maketitle

\begin{abstract}
A generalization of vector fields, referred to as $N$-direction fields or cross fields when $N = 4$, has been recently introduced and studied for geometry processing, with applications in quadrilateral (quad) meshing, texture mapping, and parameterization. We make the observation that cross field design for two-dimensional quad meshing is related to the well-known Ginzburg-Landau problem from mathematical physics. This yields a variety of theoretical tools for efficiently computing boundary-aligned quad meshes, with provable guarantees on the resulting mesh, such as the number of mesh defects and bounds on the defect locations. The procedure for generating the quad mesh is to (i) find a complex-valued  ``representation'' field that minimizes the Ginzburg-Landau energy subject to a boundary constraint, (ii) convert the representation field into a boundary-aligned, smooth cross field, (iii) use separatrices of the cross field to partition the domain into four sided regions, and (iv)  mesh each of these four-sided regions using standard techniques. Leveraging the Ginzburg-Landau theory, we prove that this procedure can be used to produce a cross field whose separatrices partition the domain into four sided regions. To minimize the Ginzburg-Landau energy for the representation field, we use an extension of the Merriman-Bence-Osher (MBO) threshold dynamics method, originally conceived as an algorithm to simulate mean curvature flow. Finally, we demonstrate the method on a variety of test domains.
\end{abstract}



\section{Introduction}
A generalization of vector fields, referred to as \emph{$N$-direction fields} allow for $N$ directions to be encoded at each point of a domain. When $N = 4$, the term \emph{cross field} is often used. Such fields are better suited to encode multi-directional information than simply using $N$ overlaid vector fields because they allow singularities of ``fractional index'' in the neighborhood of which an $N$-direction field turns $2\pi/N$ radians (see \cref{fig:fractional_index}). $N$-direction fields have found recent applications in quad re-meshing for computer graphics and finite element simulations \cite{bommes_mixed-integer_2009,jakob_instant_2015,kalberer_quadcover_2007,kowalski_pde_2013,panozzo_frame_2014}, parameterization \cite{myles_robust_2014,ray_periodic_2006}, non-photorealistic rendering \cite{hertzmann_illustrating_2000}, and texture mapping \cite{knoppel_stripe_2015}.

\begin{figure}
  \begin{center}
    \includegraphics[width=.3\linewidth]{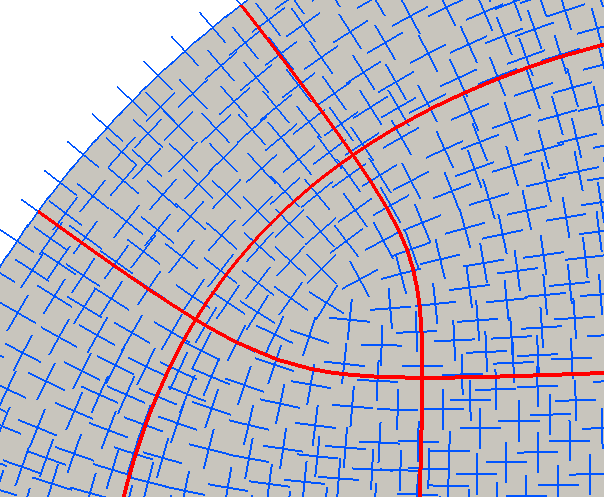} \qquad
    \includegraphics[width=.3\linewidth, trim=250 0 415 0, clip]{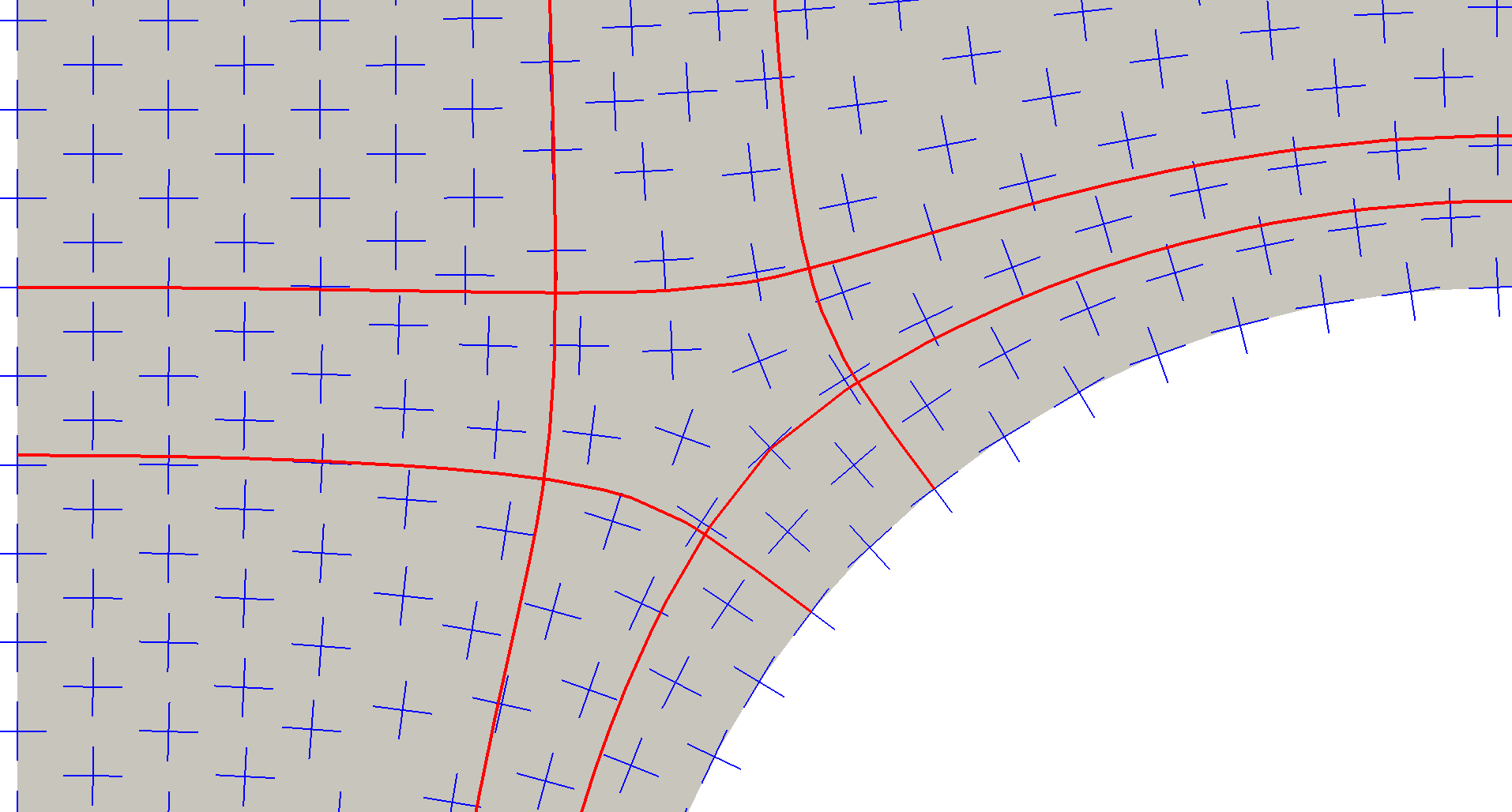}
  \end{center}
  \caption{{\bf Singularities in a cross field with indices $+1/4$ and $-1/4$.}  {\bf (left)} A singularity with index $+1/4$ is contained in the region enclosed by the three red lines, which are streamlines of the cross field. Since the index is +1/4, the cross field makes a quarter turn counter-clockwise when circulating around this singularity.
 {\bf (right)} Similarly, a singularity with index -1/4 is enclosed by 5 streamlines. The cross field makes a quarter turn clockwise when circulating around this singularity.}
\label{fig:fractional_index}
\end{figure}

\begin{figure}
  \begin{center}
    \includegraphics[width=.45\linewidth,trim=150 0 150 0, clip]{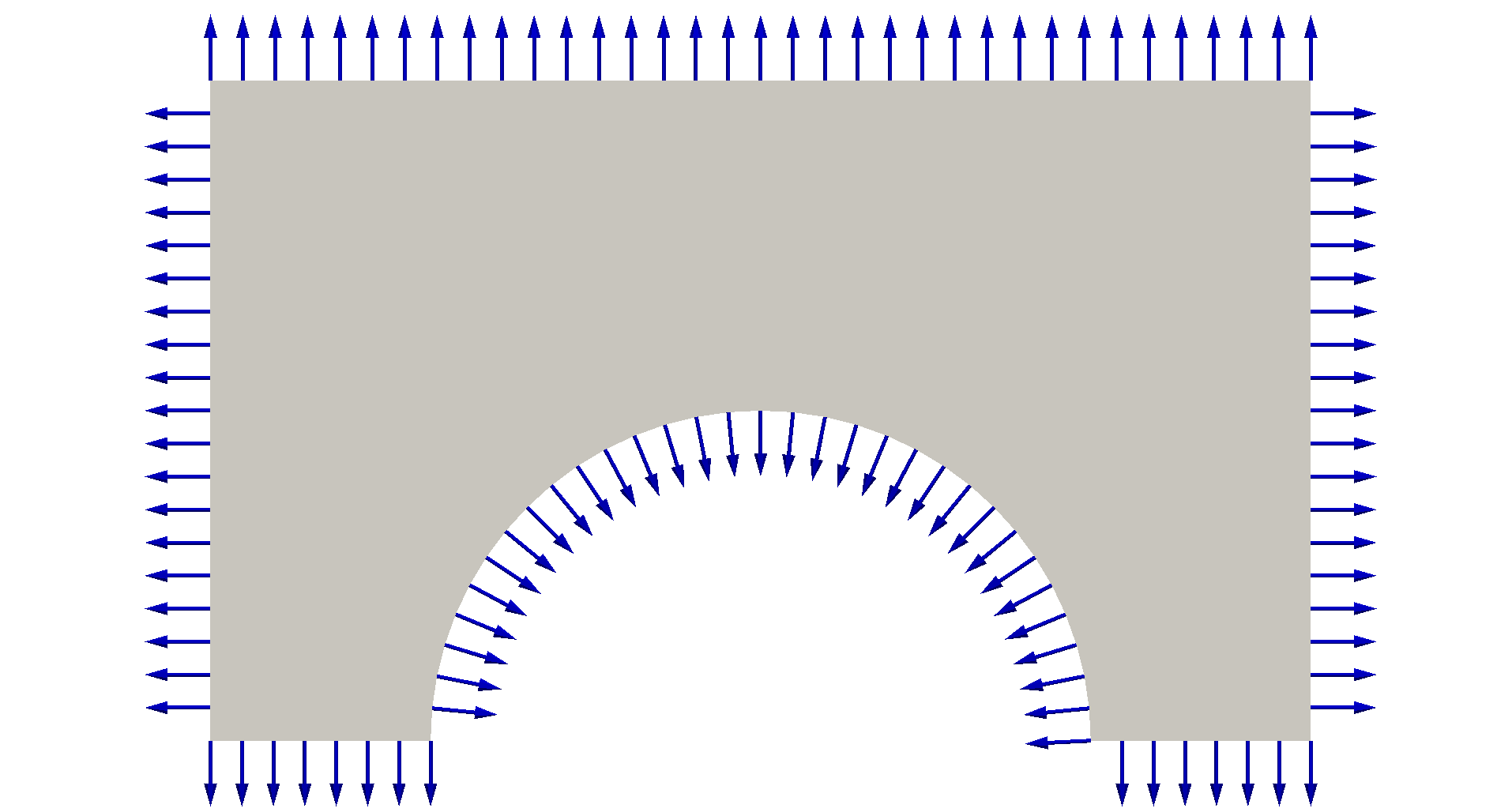} \qquad
    \includegraphics[width=.45\linewidth,trim=150 0 150 0, clip]{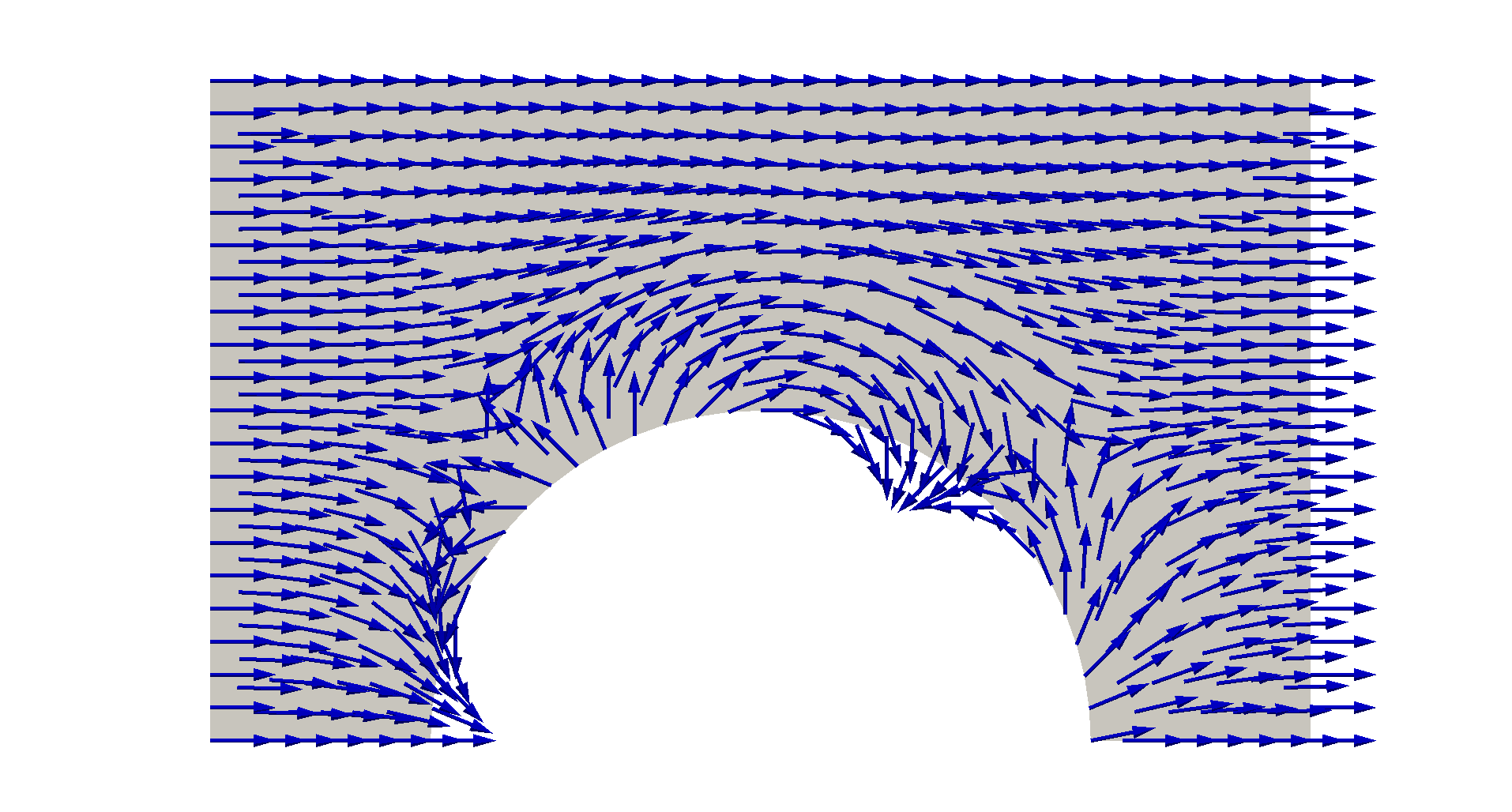} \\
    \includegraphics[width=.45\linewidth,trim=150 0 150 0, clip]{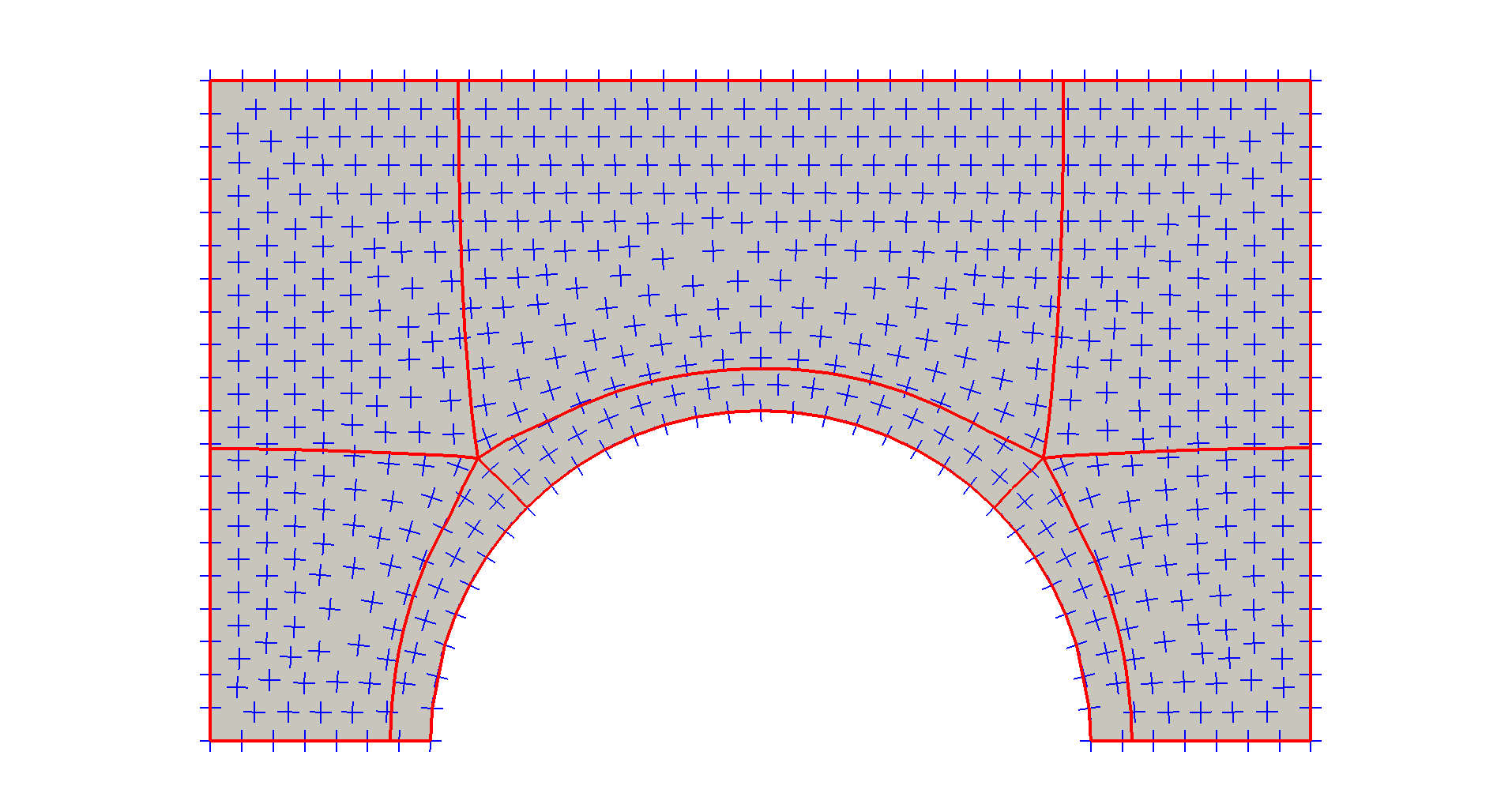} \qquad
    \includegraphics[width=.45\linewidth,trim=150 0 150 0, clip]{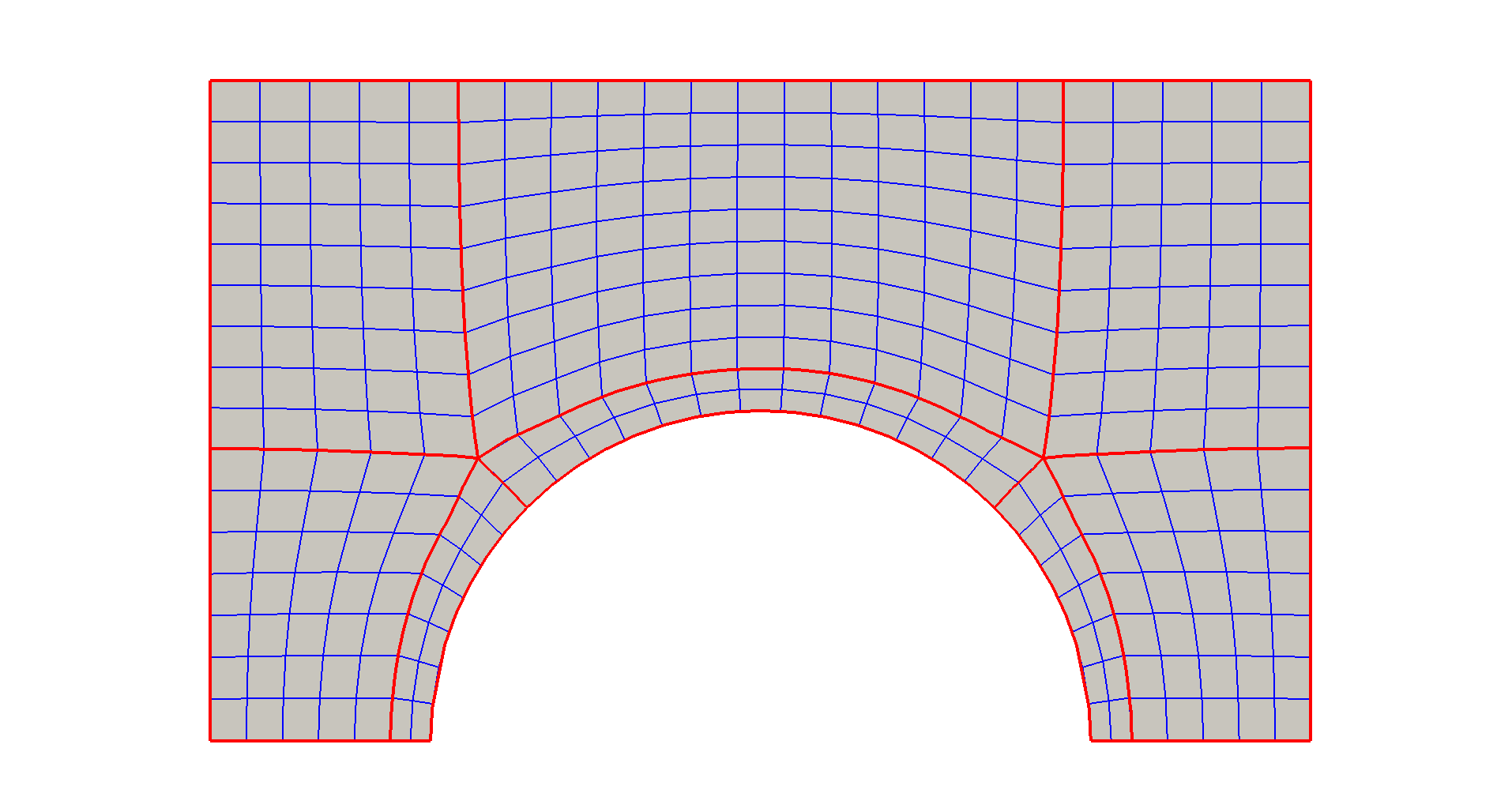}
  \end{center}
  \caption{{\bf Overview of the cross field based meshing method.} {\bf (top left)} The domain is shown with outward pointing normals. {\bf (top right)} A 4-aligned boundary condition is assigned (see \cref{def:boundary-aligned}) and a representation vector field is found by approximately minimizing the Ginzburg-Landau energy. {\bf (bottom left)} The representation field is mapped to a smooth cross field and separatrices of the cross field are traced to partition the domain into a quad layout. {\bf (bottom right)} A regular mesh is mapped into each region.}\label{fig:overview}
\end{figure}

Various cross field based quad meshing techniques have been proposed (see \cref{sec:cross_meshing}). One basic procedure is illustrated in \cref{fig:overview}. The top left panel shows a domain, $D$, with outward normal boundary vectors indicated. In the top right panel, a complex-valued ``representation'' field that minimizes the Dirichlet energy subject to a boundary condition and a unit norm constraint is found. In the bottom left panel, the representation field is converted into a boundary-aligned, smooth cross field. Separatrices of the cross field are computed which partition the domain into four sided regions, referred to as a quad layout. Finally, the bottom right panel shows a regular mesh mapped into each of the four-sided regions. Variations on this procedure have been employed in \cite{kowalski_pde_2013,myles_robust_2014}. In this paper we investigate the mathematics of cross field generation and cross field guided quad meshing via streamline tracing.

\subsection{Contributions}
We make the observation that cross field design is related to the Ginzburg-Landau theory from mathematical physics. In particular, many of the computational methods currently used for cross field design attempt to minimize an energy that is, or is very similar to, a discrete Ginzburg-Landau energy; see, for example, \cite{beaufort_computing_2017,bommes_mixed-integer_2009,jakob_instant_2015,jiang_frame_2015,kowalski_pde_2013,palacios_rotational_2007,ray_n-symmetry_2008,ray_geometry-aware_2009}.

We make this correspondence precise and use results from the Ginzburg-Landau theory to prove, in \cref{thm:quad-partitioning,thm:partitioning_termination}, that the separatrices of a harmonic cross field with indices $\leq 1/4$ partition a domain into four-sided regions, possibly with a $T$-junction if a limit cycle is present. The corners of these four-sided regions are located at the singularities of the harmonic cross field. The proof of this result depends on an asymptotic analysis of the cross field near singularities, which uses results from the Ginzburg-Landau theory; see \cref{sec:behavior}. Because we consider domains with corners, we also make precise a notion of boundary singularities (see \cref{def:boundary_singularity}) and study their properties. We show, in \cref{lem:sector_index}, the structure of a boundary singularity with these definitions is consistent with the structure of a cross field near an (interior) singularity. These properties of cross field singularities have been observed before and are assumed ubiquitously in the cross field literature. Our contribution is the proof that singularities with such properties are obtained in cross fields which approximate local minimizers of the Ginzburg-Landau energy (\cref{sec:topology}). In addition, the renormalization of the Ginzburg-Landau energy derived by Bethuel et al. \cite{bethuel_ginzburg-landau_1994} (\cref{sec:renormalized}) guarantees that the singularities of such a cross field will be isloated and have index $\pm1/4$. Based on these results, we develop an algorithm (\cref{alg:partitioning}) that uses a harmonic cross field to partition a domain into four-sided regions, possibly with T-junctions. While similar streamline tracing methods have been employed before \cite{kowalski_pde_2013,myles_robust_2014}, the steps of our algorithm carefully leverage the theorems in \cref{sec:topology} to guarantee a partition of four-sided regions with a bound on the number of T-junctions.

To compute a harmonic cross field, we approximately minimize the Ginzburg-Landau energy and find a suitable representation field by using a generalization of the MBO algorithm (\cref{alg:MBO}) \cite{merriman_motion_1994}. This results in a harmonic cross field with isolated singularities of degree $\pm 1/4$. This can then be used as input for \cref{alg:partitioning} to find a partition with four-sided regions. The partition can then be used to generate a high-quality quad mesh on the domain with standard techniques.

Finally, we use our cross field design algorithm and partitioning theorem to design quad meshes for several example geometries; see \cref{fig:ExampleMeshes} for examples. Throughout, we also include figures, generated using the algorithms described in this paper, to illustrate the main ideas.

\section{Previous Work}\label{sec:review}
Here we present a selection of works that have contributed to the understanding of fundamental issues in the cross field design and meshing problems. For brevity, we only review the work that is most relevant to this paper, however excellent literature reviews are available on cross field and directional field design \cite{vaxman_directional_2016} and recent approaches to quad meshing \cite{bommes_quad_2012}.

\subsection{Cross Field Design}
Designing globally smooth direction fields has proven to be a challenging task. Such approaches are typically formulated as an energy minimization problem, where the energy functional is a discrete approximation of the Dirichlet energy of the cross field. Papers which use the $N$-RoSy representation, introduced in \cite{palacios_rotational_2007}, minimize the Dirichlet energy of a vector valued field \cite{beaufort_computing_2017,jiang_frame_2015,knoppel_globally_2013,kowalski_pde_2013,ray_geometry-aware_2009}. Directional information is typically removed from the problem by requiring a pointwise unit-norm constraint. Since the na\"ive formulation is ill-posed, researchers have proceeded by adding a penalty term to the energy in place of directly enforcing the pointwise unit norm constraint. This penalty term is typically nonlinear and thus leads to an inefficient solution. Kowalski et al. \cite{kowalski_pde_2013} add the linear constraint term $u_{n-1} \cdot u_n = 1$ to the objective function via Lagrange multipliers. At each iteration the solution is normalized so that the solution to the next iteration is penalized wherever it is not of unit norm. Jiang et al. \cite{jiang_frame_2015} add the non-linear term $( |u|^2 - 1)^2$ to their objective and use a nonlinear solver to minimize the cross field energy with respect to a custom Riemannian metric. Beaufort et al. \cite{beaufort_computing_2017} take direction from the Ginzburg-Landau theory and also add the nonlinear term $( |u|^2 - 1)^2$ to the objective. They then minimize the energy via a Newton method.

Rather than applying a pointwise unit norm constraint to remove directional information, Kn\"oppel et al. \cite{knoppel_globally_2013} use a modified energy and enforce an $L^2$ norm constraint of the field. The minimizer is then given by the principle eigenfunction of the Laplacian. Unfortunately, this method is not directly applicable when aligning a cross field to the boundary as the Dirichlet boundary condition is inhomogeneous. Bommes et al. \cite{bommes_mixed-integer_2009} use an angle based representation for crosses and use integer variables called \emph{period jumps}, introduced in \cite{li_representing_2006}, to encode rotations between crosses on adjacent nodes. This leads to a mixed-integer optimization problem. Jakob et al. \cite{jakob_instant_2015} use a multigrid approach and perform per-node local smoothing iterations on the cross field, where each cross vector is renormalized after smoothing. In \cref{sec:comparison}, we compare our cross field design method with a nonlinear method that directly minimizes the Dirichlet energy with a nonlinear penalty term and with instant meshes \cite{jakob_instant_2015}.

\begin{figure}
  \begin{center}
    \includegraphics[width=.32\linewidth, trim=0 -30 -40 0, clip]{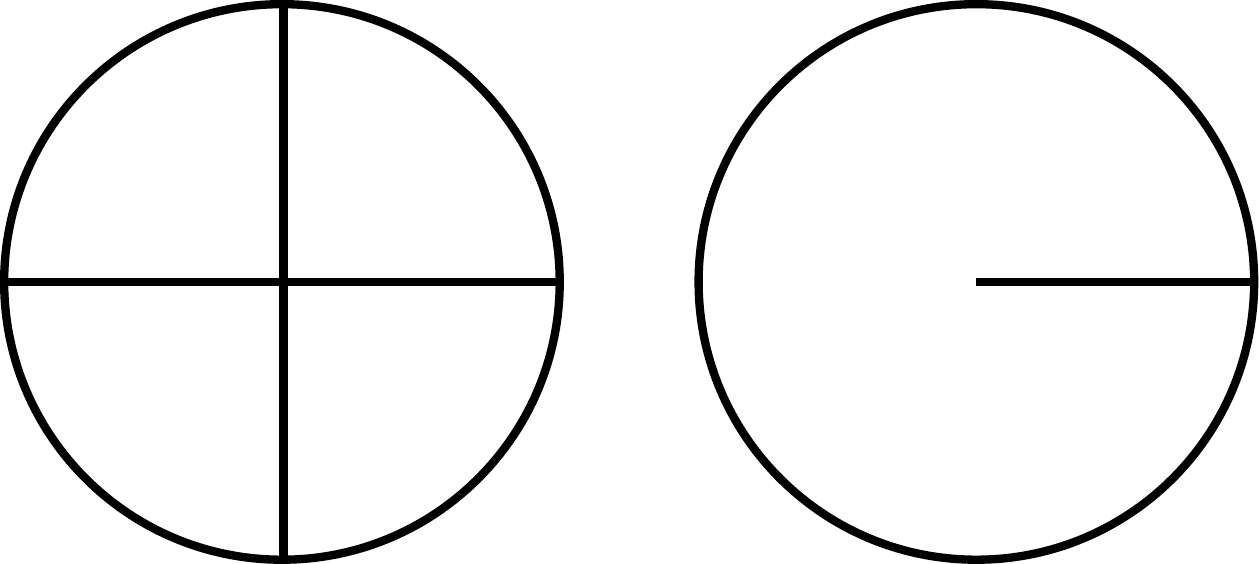}  \qquad \qquad
    \includegraphics[width=.32\linewidth, trim=0 -30 -40 0, clip]{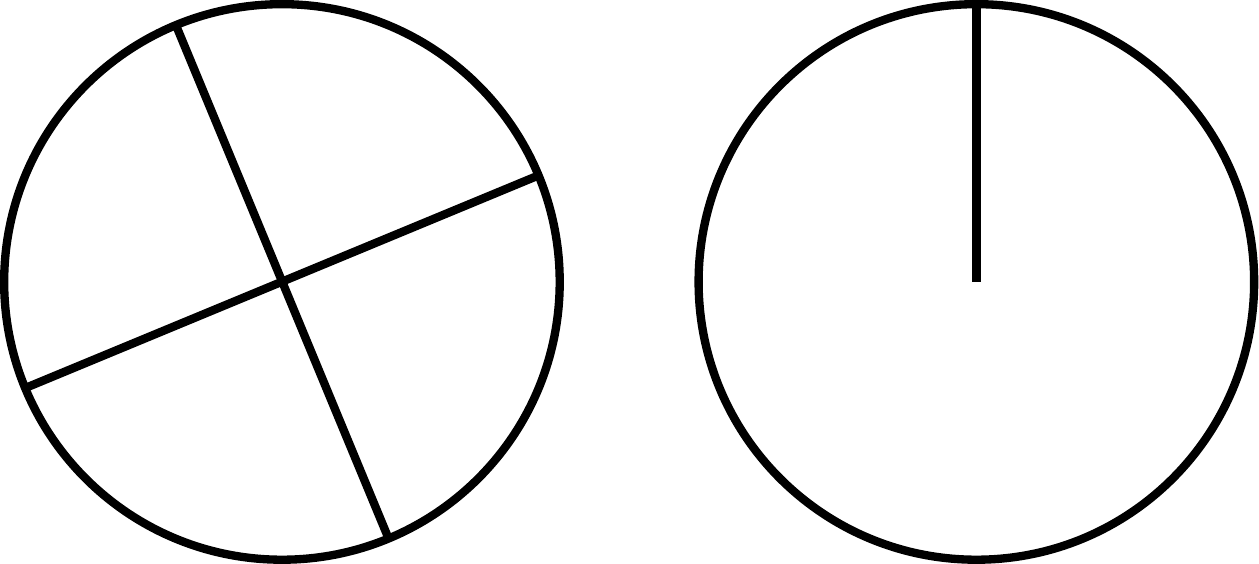} \\
    \includegraphics[width=.32\linewidth, trim=0 -30 -40 0, clip]{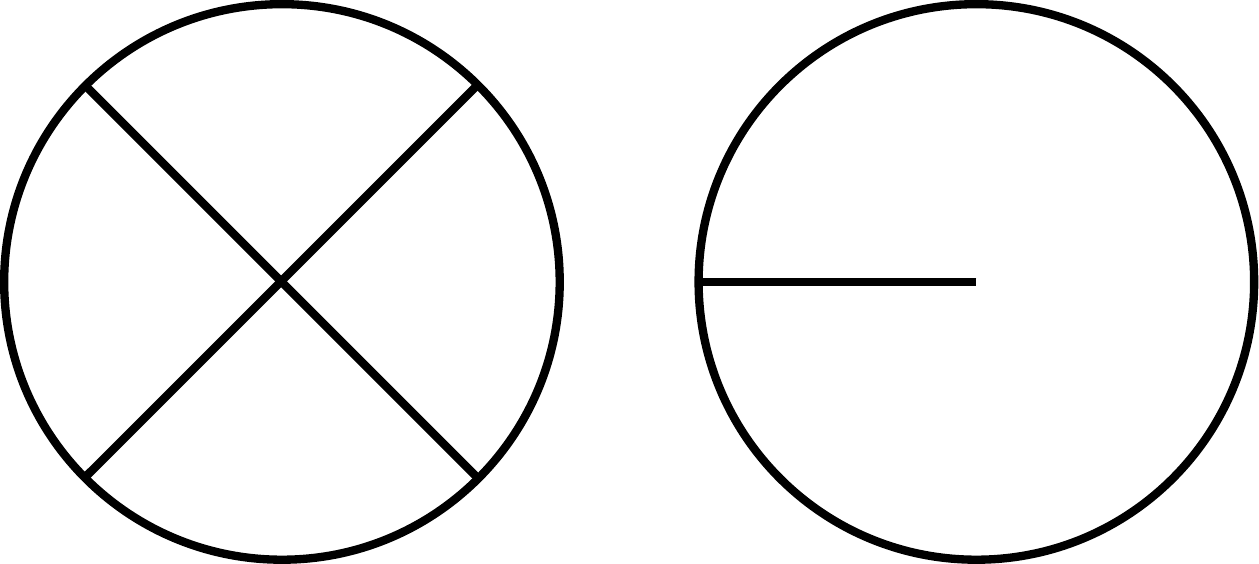}  \qquad \qquad
    \includegraphics[width=.32\linewidth, trim=0 -30 -40 0, clip]{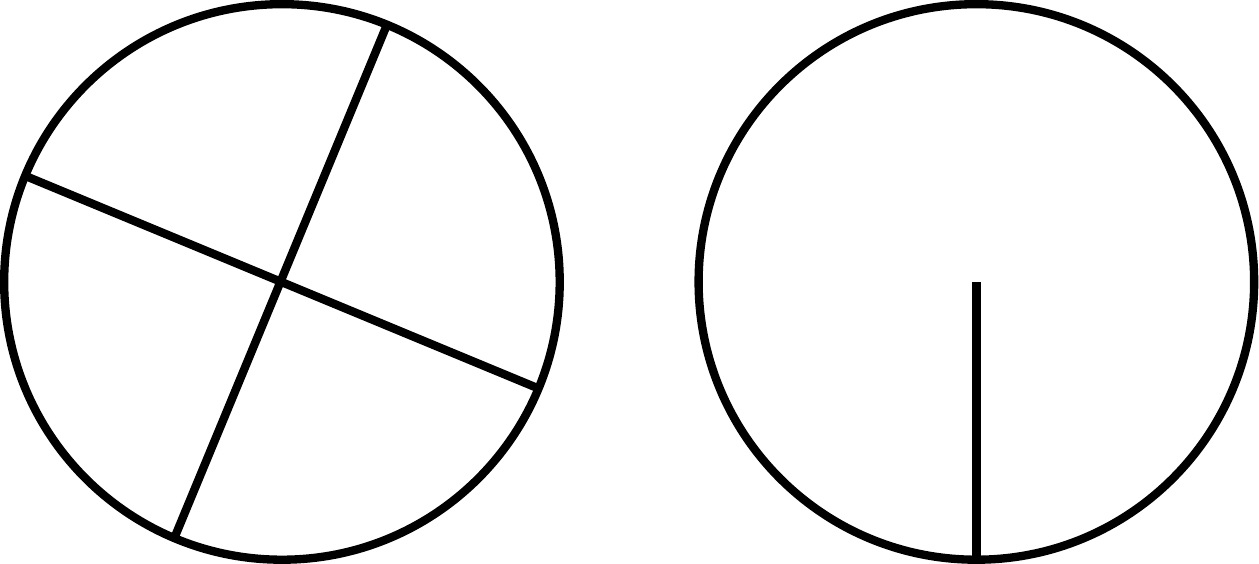}
  \end{center}
\caption{{\bf The representation map for cross fields. } The left figure for each pair represents an element, $c$, of $C_4$ (see \cref{def:n-direction}), visualized by a line from zero to the representative numbers on the unit circle. The right figure shows the representation as a line from zero to $[c]^4$. This representation is equivalent to the $N$-RoSy representation, only expressed in complex numbers.}\label{fig:nrosy}
\end{figure}

\subsection{Cross Field Guided Meshing} \label{sec:cross_meshing}

Most cross field driven quad meshing techniques fall into two categories: streamline tracing techniques and parameterization techniques. While parameterization techniques have found success in some commercial applications and received more attention in early research \cite{bommes_mixed-integer_2009,bommes_integer-grid_2013,kalberer_quadcover_2007}, we choose to focus on streamline tracing methods because of their direct connection with the topology of the cross field. We refer  the reader interested in parameterization based meshing techniques to \cite{bommes_quad_2012}.

Streamline tracing methods decompose the geometry into four-sided regions by tracing separatrices of the cross field. Alliez et al. \cite{alliez_anisotropic_2003} produced quad dominant meshes by tracing streamlines in principle curvature directions and filling in flat areas with triangle elements. Kowalski et al. \cite{kowalski_pde_2013} implemented a streamline tracing algorithm in flat 2D by tracing streamlines starting at singularities and continuing until the streamline reaches either the boundary of the domain or another singularity. This approach works well on most domains but can fail on geometries with limit cycles. Na\"ive streamline tracing algorithms have the disadvantage that parallel streamlines can intersect due to numerical inaccuracies. Ray and Sokolov \cite{ray_robust_2014} and Myles and Zorin \cite{myles_robust_2014} independently developed robust streamline tracing methods that prevent such errors. Further, Myles and Zorin \cite{myles_robust_2014} developed a robust algorithm to partition a 2-manifold into four-sided regions and demonstrated its robustness on a database of 100 objects. They prove that for a cross field whose singularities satisfy certain assumptions, the regions traced by their algorithm are guaranteed to have four sides. In \cref{sec:topology} we show that the singularities of cross fields that are local minimizers of the Ginzburg-Landau energy will satisfy such assumptions. Campen et al. \cite{campen_quantized_2015} use a hybrid approach where streamlines are traced on the surface to form a T-mesh and then integer values are assigned to the edges of the T-mesh to compute a valid conformal parameterization.

\section{Basic Definitions}\label{sec:definitions}
Here we recall some background material and establish nomenclature for the paper. We briefly review concepts in meshing, but provide more detailed and precise definitions for concepts related to cross fields. For an introduction to quad meshing, see \cite{bommes_quad_2012,mario_botsch_polygon_2010,owen_survey_1998}.

\subsection{Quad Meshing}\label{sec:quad-meshing-def}
A 2D \emph{quad mesh} is a is a finite collection of \emph{nodes}, \emph{edges}, and \emph{faces} where each face is is bounded by four edges. Here edges are non-intersecting straight lines between nodes, and faces do not overlap. Nodes can be seen as points in $\mathbb{R}^2$. The \emph{valence} of a node is the number of edges connected to it. An internal node of a quad mesh is said to be \emph{irregular} if its valence is not four, and a boundary node if its valence is not three. A \emph{quad layout} is similar to a quad mesh, but in place of straight edges piecewise smooth curves are allowed. A mesh or quad layout is \emph{conforming} if any two faces share at most a single vertex or an entire edge. A \emph{T-junction} is a feature in a non-conforming mesh where at least two faces share only part of an edge; see \cref{fig:LimitCycle} (right). A quad mesh has an underlying structure called a \emph{base complex} or \emph{skeleton} which is the coarsest quad layout obtained from a subset of the mesh that has the same boundary and set of irregular nodes. For example, the skeleton of the quad meshes in \cref{fig:overview} (bottom right), \cref{fig:LocalAnalysis} (bottom right), and \cref{fig:ExampleMeshes} (right) are highlighted in red.

\subsection{Cross Fields}\label{sec:cross-field-def}
\begin{definition}\label{def:n-direction}
Let $\mathbb{T} = \{ z \in \mathbb{C}\colon |z| = 1 \}$ be the circle group with group operation given by complex multiplication and let $\rho(N)$ be the set of the $N$th roots of unity. An \emph{$N$-direction} is an element of $C_N = \mathbb{T}/\rho(N)$. A 4-direction is also called a \emph{cross}.
\end{definition}

\begin{definition}
There is a canonical group isomorphism $R \colon C_N \rightarrow \mathbb{T}$ called the \emph{representation map} given by \[R([c]) = c^N,\] where $c$ is any representative member of the equivalence class $[c] \in C_N$. The \emph{inverse representation map} $R^{-1} \colon \mathbb{T} \rightarrow C_N$ is given by
\[
R^{-1}(u) = \left[\sqrt[N]{u}\right],
\]
i.e., by choosing the equivalence class of the $N$th roots of $u$.
\end{definition}

\begin{assumption} \label{as:dom}
  Throughout this paper we will assume that $D$ is a bounded, simply connected domain in $\mathbb{R}^2$ with piecewise-smooth boundary.
\end{assumption}

\begin{definition}
  An \emph{$N$-direction field} on $D$ is a map $f \colon D \rightarrow C_N \cup \{ 0\}$ where only finitely many points are mapped to zero. The map $R\circ f \colon D \rightarrow \mathbb{T} \cup \{ 0\}$ is called the \emph{representation field for $f$}.
\end{definition}

An $N$-direction, $[c] \in C_N$, is often visualized as an unordered set  of $N$ unit vectors, $\{v_0,v_1,\ldots,v_{N-1}\}$, each one pointing from the origin towards the $N$ representative elements of the class $[c]$. Clearly, an $N$-direction has a rotational symmetry of $2\pi/N$. An example of a cross field is illustrated in \cref{fig:overview}(bottom left).

\begin{definition}
Let $N(f)$ be the zero set of an $N$-direction field, $f$ on $D$. Then $f$ is \emph{smooth} if $R \circ f \colon D \setminus N(f) \rightarrow \mathbb{T}$ is a smooth map. Similarly, $f$ is \emph{harmonic} if $R \circ f$ is harmonic on $D \setminus N(f)$, i.e., satisfies $\Delta (R \circ f) = 0$.
\end{definition}

\begin{definition}\label{def:boundary-aligned}
  Let $\nu = (\nu_1,\nu_2)$ be the outward pointing unit normal vector on $\partial D$, and let $c_\nu = (\nu_1 + i \nu_2)$. If a map $g \colon  D \rightarrow \mathbb{T}$ is such that $g = c_\nu^N$ for every smooth point $p \in\partial D$, then  $g$ is said to be \emph{$N$-aligned} to the boundary of $D$. If $f\colon D \to C_N$ is an $N$-direction field on $D$ such that $R \circ f|_{\partial D} = c_\nu^N$, then $f$ is said to be \emph{boundary-aligned} to $D$.
\end{definition}

\begin{definition}
The \emph{Brouwer degree} of a map $g \colon \partial D \rightarrow \mathbb{T}$, written $d = deg(g,\partial D)$, is the winding number of the curve $g(\partial D)$ around the origin in the complex plane.
\end{definition}

\begin{definition} \label{def:IntInd}
Let $\gamma \colon [0,1] \rightarrow D$ be a simple closed curve circulating a single zero of the $N$-direction field $f$ at an interior point $p$ of $D$. Then the value \[\Index(p) := \frac{\arg{R(f(\gamma(1)))} - \arg{R(f(\gamma(0)))}}{2\pi N}\] is the \emph{index} of $p$. The zero $p$ is called a \emph{singularity} of the $N$-direction field if its index is not zero. The index of a singularity of an $N$-direction field is  $1/N$ times the index of the corresponding singularity of the representation field.
\end{definition}

\begin{definition}\label{def:boundary_singularity}
  Let $\partial D$ be piecewise smooth with corners $\{c_1, \ldots, c_k\}$. For a corner, $c_i$, let $\gamma \colon [0,1] \rightarrow \overline{D}$ be a simple closed curve such that $\gamma(0) = c_i = \gamma(1)$, and $y'(0)$ and $y'(1)$ are tangent to $\partial D$ at $c$, and containing no other singularity. Let
\[\Delta \arg(c_i) = \lim_{s \downarrow 0} \ \arg{R(f(\gamma(1-s)))} - \arg{R(f(\gamma(s)))}, \]
and let $\interior{c_i}$ be the interior angle at $c_i$. The  \emph{index of corner $c_i$}  is defined
\begin{equation}\label{eq:boundary_index}
  \Index(c_i) \coloneqq \frac{\pi - \interior(c_i) - \frac{1}{N}\Delta \arg(c_i)}{2\pi}.
\end{equation}
The corner $c_i$ is said to be a \emph{boundary singularity} if its index is non-zero.
\end{definition}

\begin{figure}
\begin{center}
  \includegraphics[width=\linewidth]{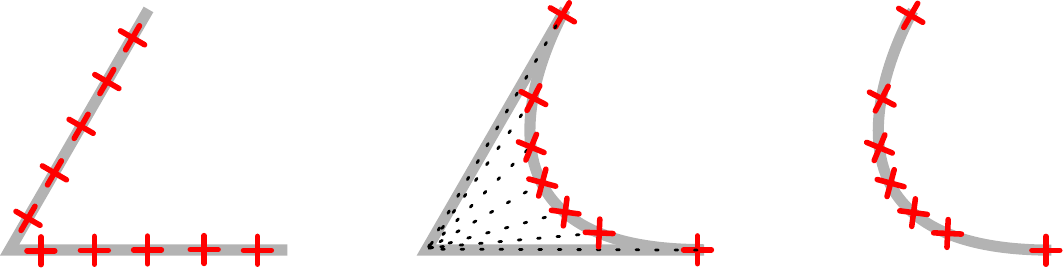}
\end{center}
\caption{{\bf (left)} The cross field and boundary singularity at a sharp corner with index $1/4$. The cross turns $90^{\circ}$ counterclockwise with respect to the direction of the curve, all at a single point. {\bf (center)} Transformation between a sharp corner and a smooth curve, the crosses turn smoothly with the angle of inclination. {\bf (right)} The cross field along a smooth curve turning $90^{\circ}$ counterclockwise with respect to the direction of the curve.}
\label{fig:CornerSmoothing}
\end{figure}
The index of a boundary singularity can be interpreted as the number of $1/N$ turns in the counterclockwise direction that the $N$-direction makes in relation to the boundary. It is akin to the concept of a \emph{turning number} from \cite{ray_n-symmetry_2008}, except that it happens at a single point (see \cref{fig:CornerSmoothing}).

\subsubsection{Streamlines}\label{sec:streamline-def}

A characteristic trait of a smooth vector field is that away from zeros it locally foliates the space, meaning that streamlines of the vector field partition the space into disjoint curves. Streamlines can be similarly defined for an $N$-direction field except that these streamlines can intersect themselves and each other at angles of $2\pi/N$, as  precisely describe below.

For the representation field $R\circ f \colon D \rightarrow \mathbb{T}$ of the $N$-direction field $f$ on $D$, the map $\Lambda_N \colon D \rightarrow \mathbb{T}$ defined by $\Lambda_N = \sqrt[N]{R(f)}$ is a multi-valued map on $D$. We can make a branch cut from each of the singularities of $R(f)$ to the boundary, and define a Riemann surface, $\mathcal{R}$ for this choice of branch cuts.
\begin{definition}
The \emph{covering vector field} on $D$ for the given $N$-direction field $f$ and choice of branch cuts of $\Lambda_N$ is the continuous vector field
$\widehat{\Lambda}_N \colon \mathcal{R} \rightarrow \mathbb{T}$ defined by assigning the vector pointing from the origin to $\Lambda_N(p)$ at each point $p$ of $\mathcal{R}$.
\end{definition}

The  observation that a cross field has a corresponding continuous vector field on a 4-covering of the domain of definition is accredited to K\"alberer et al. \cite{kalberer_quadcover_2007}.

\begin{figure}\label{fig:n-covering}
\begin{center}
\includegraphics[width=.45\linewidth]{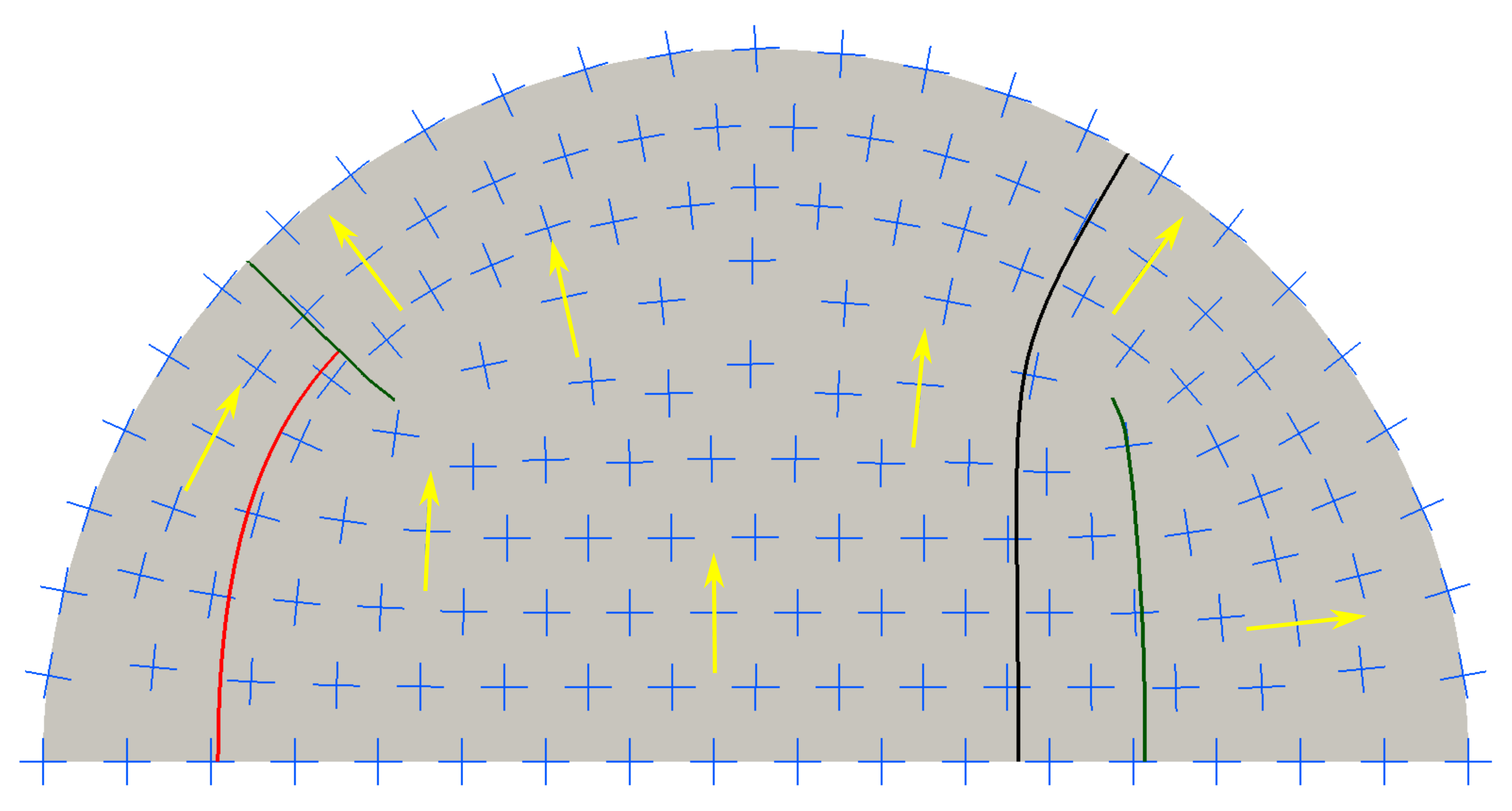} \qquad
\includegraphics[width=.45\linewidth]{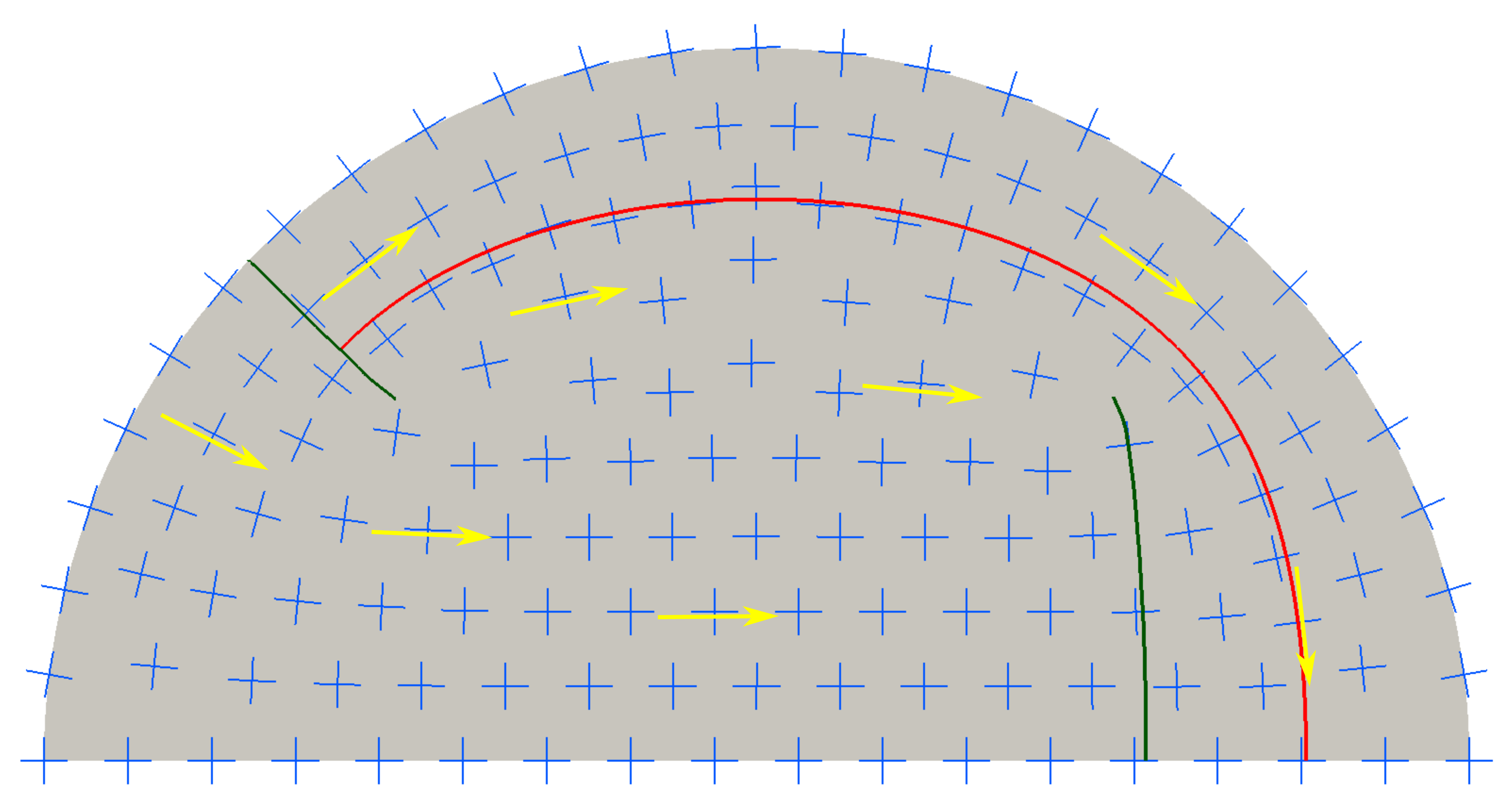} \\
\includegraphics[width=.45\linewidth]{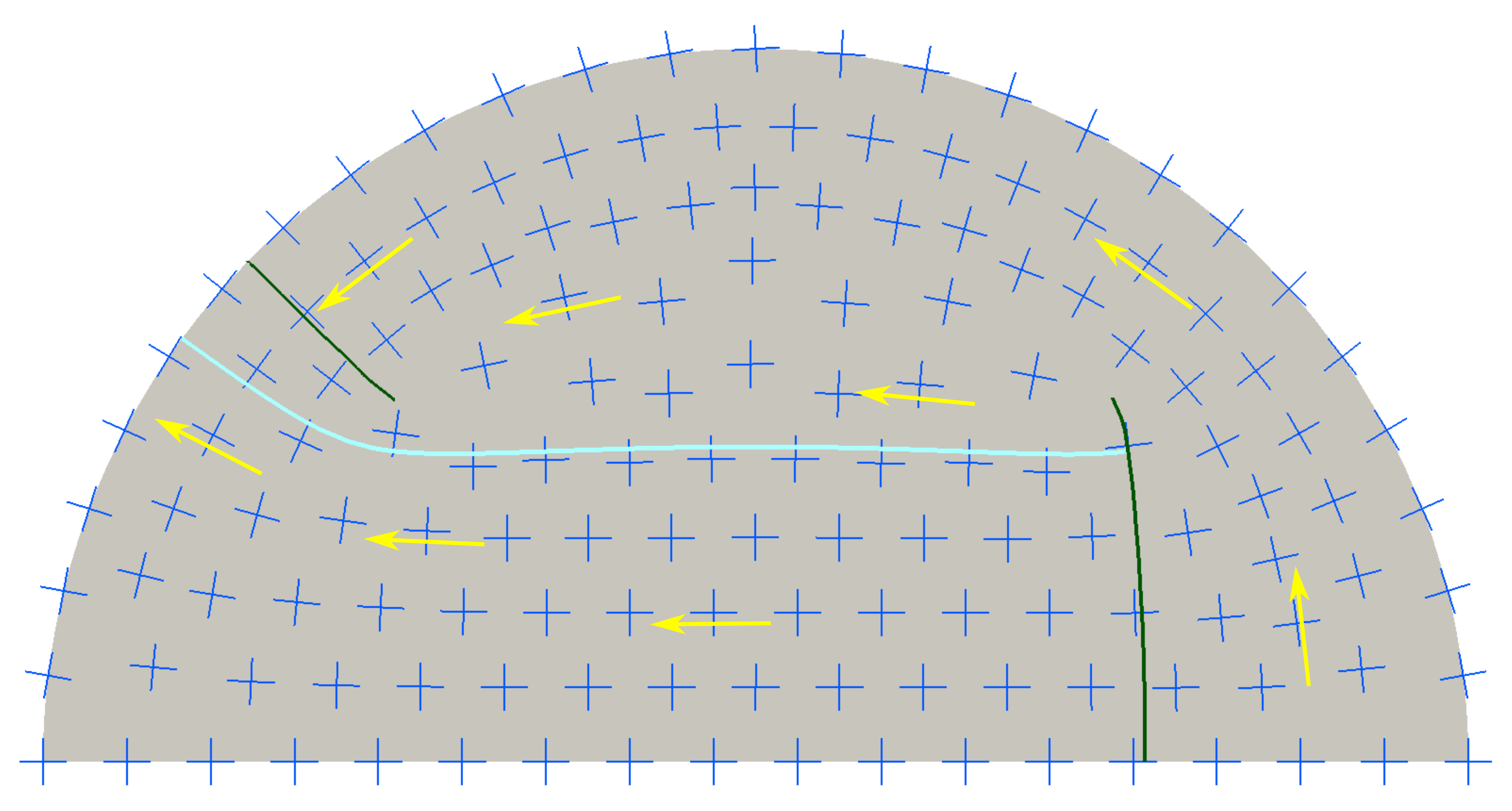} \qquad
\includegraphics[width=.45\linewidth]{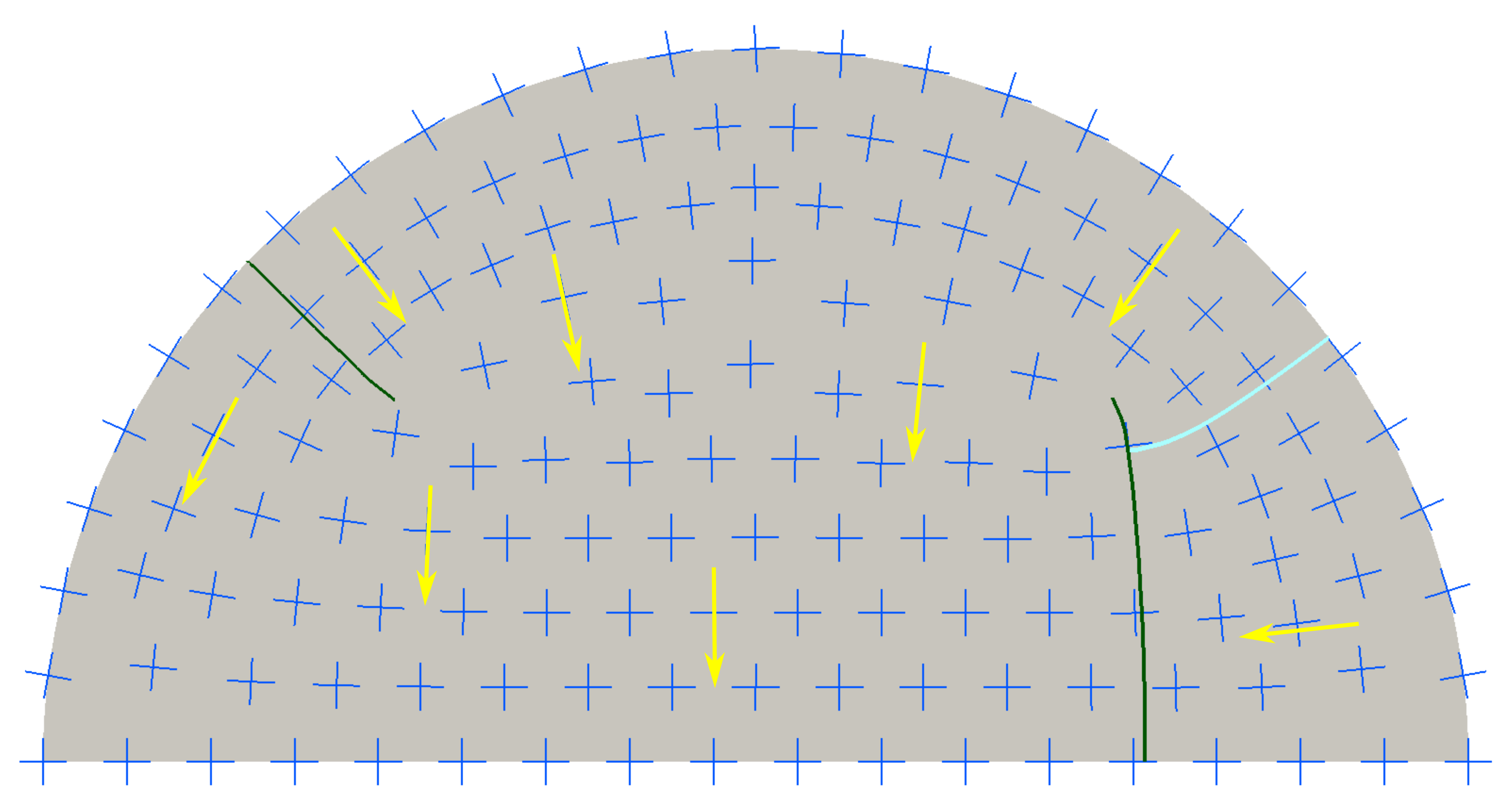}
\end{center}
\caption{The covering vector field and the four sheets of the Riemann surface for the half disk. The red and cyan streamlines both cross a green branch cut, causing the streamline to change sheets. The red, black, and cyan streamlines of the covering vector field project onto the base domain as streamlines of the cross field.}
\end{figure}

\begin{definition}\label{def:streamline}
  Let $\gamma_N \colon [a,b] \rightarrow \mathcal{R}$ be a streamline of $\widehat{\Lambda}_N$, satisfying $\frac{d\gamma_N}{dt} = \widehat{\Lambda}_N(\gamma_N(t))$ for $t \in [a,b]$. Let $\pi \colon \mathcal{R} \rightarrow D$ be the canonical projection from $\mathcal{R}$ to $D$. Then the function $\gamma \colon [a,b] \rightarrow D$ given by
  \[\gamma = \pi \circ \gamma_N\]
  is a \emph{streamline} of the $N$-direction field $f$. A \emph{separatrix} of an $N$-direction field is a streamline that begins or ends at a singularity.
\end{definition}

\Cref{fig:n-covering} shows an illustration of the covering  field for a cross field on a half disk. The two green lines are branch cuts that join a singularity to a point on the boundary.  Several example streamlines are drawn. For example, the red streamline in the top left panel is continued in the top right panel.

No two streamlines of $\widehat{\Lambda}_N$ can intersect each other on the Riemann surface $\mathcal{R}$. Thus, streamlines of the $N$-direction field intersect at a point $p \in D$ only if $\pi^{-1}(\gamma)|_p$ is on two different sheets of $\mathcal{R}$. We conclude that streamlines of an $N$-direction field can only intersect themselves and each other at integer multiples of the angle $2\pi/N$; see, {\it e.g.}, \cref{fig:ExampleMeshes} (bottom center).

\subsection{Treating domains with corners} \label{sec:domains-with-corners}

The theorems of the Ginzburg-Landau theory in \cref{sec:ginzburg-landau} require a smooth boundary, however many domains of interest for the meshing problem have corners. Throughout this paper, we will assume that $D$ is a bounded, simply connected domain in $\mathbb{R}^2$ with piecewise smooth boundary (\cref{as:dom}). In order to make use of the Ginzburg-Landau theory, we define a smooth \emph{auxiliary domain}, $\widetilde{D} \subseteq D$, as follows. We smooth each corner, $c$, of $D$ with a B\'ezier curve, $\partial \widetilde{D}_c$, with three control points, one on the corner, and two on the boundary at a distance $\varepsilon \ll 1$ from the corner. The boundary condition along $\partial \widetilde{D}_c$ is then assigned by linearly interpolating the cross with the angle of inclination above the corner; see \cref{fig:CornerSmoothing}. The smooth domain $\widetilde{D}$ will be considered in \cref{sec:ginzburg-landau}.

 If $\widetilde{f}$ is a cross field on $\widetilde{D}$, then a cross field $f$ on $D$ can be defined as follows: set $f|_{\widetilde{D}} = \widetilde{f}$, and for each corner $c$ of $D$, assign the crosses on $D \setminus \widetilde{D}$ by propagating the crosses from $\partial \widetilde{D}_c$ constant along each line segment from $c$ to $\partial \widetilde{D}_c$. The cross field $f$ is called the \emph{extension} of $\widetilde{f}$ to $D$. Note that if $\widetilde{f}$ is continuous then $f$ is also continuous. It is important to note that defining the auxiliary domain $\widetilde{D}$ consequently extending $\widetilde{f}$ to $D$ is only important for theoretical considerations, to be able to directly apply the theorems from \cite{bethuel_ginzburg-landau_1994}. In practice there is no need to perform such operations to generate a cross field on a discrete representation of the domain.

\section{Correspondence Between Cross Field Design and the Ginzburg-Landau Theory} \label{sec:ginzburg-landau}

The Ginzburg-Landau equation and associated energy functional are classically used to describe the physics of superconductors and superfluids. In this section, we describe the connection between a commonly used formulation of the cross field design problem and the Ginzburg-Landau theory. In each subsection we review different aspects of the Ginzburg-Landau theory in \cite{bethuel_ginzburg-landau_1994} and apply the results to the cross field design problem.

As mentioned in \cref{sec:review}, the goal of the cross field design problem is often to find a harmonic cross field. This is sometimes formulated as an energy minimization problem with a feasibility constraint. Approaches using the $N$-RoSy representation for a cross field, $f$, use the Dirichlet energy for the representation map, given by
\begin{equation} \label{eq:energy}
  E[R(f)],
  \qquad \qquad \textrm{where} \quad
   E[u] := \frac{1}{2}\int_{\widetilde{D}}{|\nabla{u}|^2dA}.
\end{equation}
The problem then is to find the minimizer among all complex fields that represent a cross field:
\begin{align}
\label{eq:ginzburg-landau-problem}
  \inf_{u \in H_g^1(\widetilde{D}; \mathbb T)}  \  {E(u)}
\end{align}
where
$$
H_g^1(\widetilde{D};\mathbb{T}) := \{ u \in H^1(\widetilde{D}; \mathbb{C}) \colon \  u(x) = g(x) \  \forall x \in \partial \widetilde{D} \ \ \textrm{and} \ \ |u(x)| = 1   \ \textrm{a.e.} \  x \in \widetilde{D} \}.
$$
Here, $g$ is boundary aligned (see \cref{def:boundary-aligned}), and the feasibility constraint $|u(x)| = 1$ keeps the solution on the unit circle so that a cross can be defined at a point $x$ by $R^{-1}(u(x))$. Note that the notion that a cross should only carry \emph{directional} information motivates the definition of a cross given in \cref{def:n-direction}, and in turn restricts the solution space to $\mathbb{T}$.

The admissible set in \cref{eq:ginzburg-landau-problem}, $H_g^1(\widetilde{D},\mathbb{T})$, is empty whenever the Brouwer degree, $d$, is non-zero. \cite{bethuel_ginzburg-landau_1994}. Indeed, if $d \neq 0$, the Poincar\'e-Hopf theorem necessitates a singularity will occur somewhere in $\widetilde{D}$, but the Dirichlet energy  \cref{eq:energy} in any neighborhood of the singularity is infinite \cite{bethuel_ginzburg-landau_1994,knoppel_globally_2013}. Problem \cref{eq:ginzburg-landau-problem} can be relaxed by enlarging the admissible set to $H_g^1(\widetilde{D},\mathbb{C})$, so that the solution can approach zero in the neighborhood of a singularity. A penalty term can then be added to the Dirichlet energy giving  the minimization  problem
\begin{equation} \label{eq:GL_functional}
  \inf_{u \in H_g^1(\widetilde{D},\mathbb{C})}{E_\varepsilon(u)}
  \qquad \qquad \textrm{where} \quad
  E_\varepsilon(u) = \frac{1}{2}\int_{\widetilde{D}}{|\nabla{u}|^2dA} + \frac{1}{4\varepsilon^2}\int_{\widetilde{D}}{(|u|^2 - 1)^2dA}.
\end{equation}
This is the approach taken in the study of the Ginzburg-Landau theory, and has also been taken in the cross field literature \cite{beaufort_computing_2017,jiang_frame_2015}; see also \cite{jakob_instant_2015,kowalski_pde_2013,ray_geometry-aware_2009} which use different strategies to enforce the pointwise unit norm constraint. In their foundational work on Ginzburg-Landau vortices, Betheul et al. show that there is a well-defined sense in which there exists a generalized solution to \cref{eq:ginzburg-landau-problem} as $\varepsilon \to 0$.

\subsection{Existence of a Generalized Solution} \label{sec:existence}

\begin{theorem}[{\cite[Theorem 0.1]{bethuel_ginzburg-landau_1994}}] \label{thm:existence}
Let $\widetilde{D} \subset \mathbb{R}^2$ be a bounded, simply connected domain with smooth boundary
\footnote{\Cref{thm:existence,thm:fixedField} are stated in \cite{bethuel_ginzburg-landau_1994} for star shaped domains, but this assumption was relaxed to simply connected domains in \cite{Struwe1994}.}
and let $g \colon \partial \widetilde{D} \rightarrow \mathbb{T}$ be a smooth function. Let $d = \deg(g,\partial \widetilde{D})$ be the Brouwer degree of $g$ on $\partial \widetilde{D}$. Denote by $u_\varepsilon$ a solution of \eqref{eq:GL_functional} for $\varepsilon >0$.
Given a sequence $\varepsilon_n \rightarrow 0$ there exists a subsequence $\varepsilon_{n_i}$ and exactly $d$ points $\{a_1,a_2,...,a_d\} \subset \widetilde{D}$ and a smooth harmonic map $u_\star \colon \widetilde{D} \setminus \{a_1,...,a_d\} \rightarrow \mathbb{T}$ with $u_\star = g \text{ on } \partial \widetilde{D}$ such that
\begin{equation*}
  u_{\varepsilon_{n_i}} \rightarrow u_\star \qquad \text{in} \ \  C_{loc}^k(\widetilde{D} \setminus \underset{i}{\cup}(a_i)) \ \ \forall k \ \ \text{and in} \ \  C^{1,\alpha}(\bar{\widetilde{D}} \setminus \underset{i}{\cup}(a_i)) \ \ \forall \alpha < 1.
\end{equation*}
In addition, if $d \neq 0$, each singularity of $u_\star$ has index $\sgn(d)$ and, more precisely, there are complex constants, $\alpha_i,$ with $|\alpha_i| = 1$ such that
\begin{equation*}
  \left\lvert u_\star(z) - \alpha_i\frac{z-a_i}{|z-a_i|} \right\rvert \leq C|z-a_i|^2 \quad \text{ as } \quad  z \rightarrow a_i \text{, }  \ \ \forall i.
\end{equation*}
\end{theorem}

In other words, \cref{thm:existence} guarantees a sequence of minimizers of the relaxed problem \cref{eq:GL_functional} that converges to a function, $u_\star$, that is harmonic on $\widetilde{D} \setminus \{a_1,...,a_d\}$. Thus we can think of $u_\star$ as a generalized solution of the minimization problem \cref{eq:ginzburg-landau-problem}. Note that while $u_\star$ is a (generalized) global minimizer, it is not necessarily unique. For example, it can be shown on the disk that the the energy \cref{eq:GL_functional} is invariant under rotation. Unfortunately, $u^\star$ cannot in general be found analytically. Further, due to the nonlinear second term of the energy functional in \cref{eq:GL_functional}, the energy becomes non-convex for small epsilon \cite{bethuel_ginzburg-landau_1994} and there is no guarantee that numerical methods will converge to the global minimum. (Interestingly, the global minimum may not always be the preferred solution for generating a quad mesh; see \cref{sec:discussion}). There is however a related concept of which we can make practical use.

\subsection{Canonical Harmonic Maps} \label{sec:canonical-harmonic-maps}
\begin{theorem}[{\cite[Corollary I.1]{bethuel_ginzburg-landau_1994}}] \label{thm:fixedField}
Let $\widetilde{D} \subset \mathbb{R}^2$ and $g \colon \partial \widetilde{D} \rightarrow \mathbb{T}$ be as in \cref{thm:existence} with $d = \deg(g,\partial \widetilde{D})$. Given any configuration $a = \{a_1,\ldots,a_n\}$ of distinct points $a_j \in \widetilde{D}$ with indices $I = \{d_1,\ldots,d_n\}$ satisfying $d = \sum_{i=1}^n d_i$, there is a unique function $u_0$ satisfying
\begin{enumerate}
\item[(i)] $u_0$ is a smooth harmonic map from $\widetilde{D} \setminus \cup_i a_i$ to $\mathbb{T}$,
\item[(ii)] $u_0 = g \text{ on } \partial \widetilde{D} $, and
\item[(iii)] for some complex numbers $\alpha_j$ with $|\alpha_j| = 1$,
\begin{equation}     \label{eq:estimate}
\left\lvert u_0(z) - \alpha_j\frac{{(z-a_j)}^{d_j}}{|z-a_j|^{d_j}} \right\rvert \leq C|z-a_j|
\ \ \text{as} \ \  z \rightarrow a_j, \ \   \forall j.
\end{equation}
\end{enumerate}
Furthermore, $u_0$ is given by
  \begin{equation} \label{eq:u_0}
    u_0 = e^{i\varphi(z)} \ \frac{{(z-a_1)}^{d_1}}{|z-a_1|^{d_1}} \
    \frac{{(z-a_2)}^{d_2}}{|z-a_2|^{d_2}}  \ \cdots \ \frac{{(z-a_n)}^{d_n}}{|z-a_n|^{d_n}},
  \end{equation}
where $\varphi$ is the solution of the Dirichlet problem
\begin{subequations} \label{eq:phidef}
  \begin{align}
        \Delta \varphi &= 0   &&\text{ in } \widetilde{D}\\
        \varphi &= \varphi_0   &&\text{ on } \partial \widetilde{D},
  \end{align}
  \end{subequations}
and $\varphi_0$ is defined on $\partial \widetilde{D}$ by
  \begin{equation} \label{eq:phibc}
    e^{i\varphi_0(z)} = g(z) \ \frac{|z-a_1|^{d_1}}{{(z-a_1)}^{d_1}} \ \frac{|z-a_2|^{d_2}}{{(z-a_2)}^{d_2}} \  \cdots \  \frac{|z-a_n|^{d_n}}{{(z-a_n)}^{d_n}}.
  \end{equation}
\end{theorem}

\begin{definition} \label{d:canHarMap}
 The smooth harmonic map $u_0 \colon \widetilde{D} \setminus \cup_i a_i  \to  \mathbb{T}$  in \cref{thm:fixedField} is called the \emph{canonical harmonic map} associated with the boundary condition $g$ and singularity configuration with locations $a = \{a_1,\ldots,a_n\}$ and indices $I = \{d_1,\ldots,d_n\}$. The $N$-direction field associated with the canonical harmonic map, defined by $R^{-1}(u_0)$ is called the \emph{canonical harmonic $N$-direction field} associated with $(g,a,I)$.
\end{definition}

\cref{thm:fixedField} states that a unique canonical harmonic map exists for a given boundary condition and singularity configuration. The immediately obvious application to cross fields is that such representation maps can be generated easily via \cref{eq:u_0}. \Cref{fig:mushroom} displays streamlines of cross fields having the same geometry and boundary condition but different singularity configurations. Each of these cross fields was generated using the explicit formulation from \cref{thm:fixedField}.

\begin{figure}
\begin{center}
\includegraphics[width=.32\linewidth]{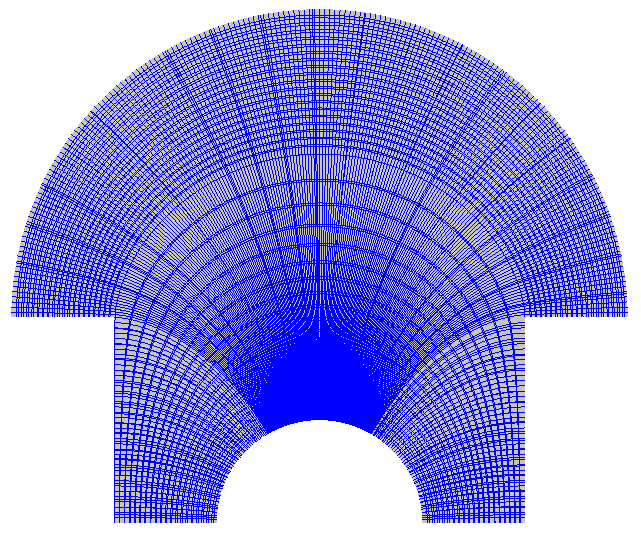}
\includegraphics[width=.32\linewidth]{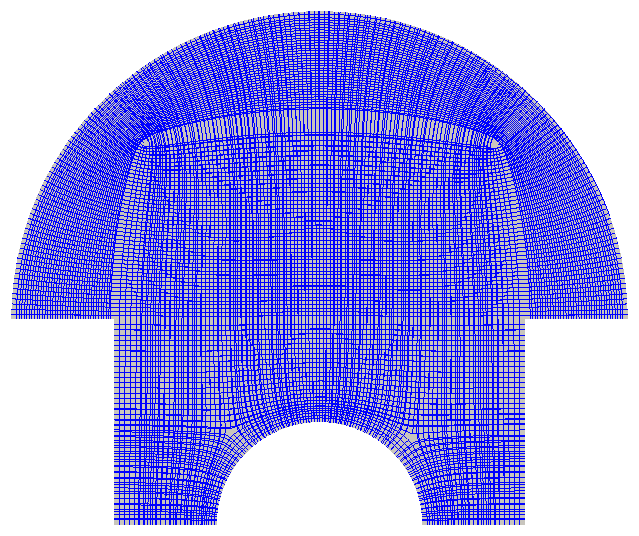}
\includegraphics[width=.32\linewidth]{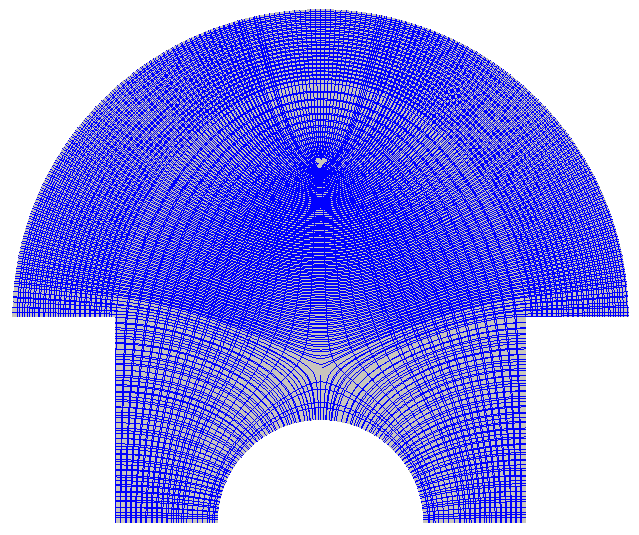}
\end{center}
\caption{A plot of the streamlines for multiple smooth cross fields on a ``mushroom'' domain, each with different singularity configuration. The Brouwer degree of this domain is zero.
{\bf (left)} This is the boundary-aligned canonical harmonic cross field with no singularities, which by \cite[Theorem 0.3]{bethuel_ginzburg-landau_1994} is the global minimizer of the Ginzburg-Landau energy.
{\bf (center)} This cross field has four singularities. The top two have degree $+1/4$ and the bottom two have degree $-1/4$.
{\bf (right)} This cross field has two singularities. The top one has degree $+1/2$ and the bottom one has degree $-1/2$. }
  \label{fig:mushroom}
\end{figure}

\cref{thm:fixedField} has an interesting theoretical consequence that draws a connection between angle based approaches to cross field design (\cite{bommes_mixed-integer_2009,li_representing_2006,ray_n-symmetry_2008}) and $N$-RoSy approaches, i.e., problem \cref{eq:ginzburg-landau-problem}. In the case where $d = 0$, the global minimizer of \cref{eq:ginzburg-landau-problem}, $u_\star$, admits no singularities. A result by Bethuel et al. \cite[Theorem 0.3]{bethuel_ginzburg-landau_1994} states that $u_\star$ is the canonical harmonic map for its associated singularity configuration. Thus, since no singularities occur, we can solve for $u_\star$ via \cref{eq:u_0}, simply by solving \cref{eq:phidef}. Since there are no singularities, the quantity $\varphi$ becomes exactly the angle of inclination of the $N$-Rosy representation vector, which in turn is simply a scalar multiple of the angle of inclination of the cross. Thus when $d = 0$, $u_\star$ is found by minimizing the Dirichlet energy of the cross angle. Put slightly stronger, away from singularities, $u_\star$ minimizes both the Dirichlet energy over fields $u \colon D \rightarrow \mathbb{T}$, and the Dirichlet energy of the \emph{argument} $\varphi = \arg(u)$ over all $u \colon D \rightarrow \mathbb{T}$.

Another useful fact from \cref{thm:fixedField} is that for a given boundary condition and singularity configuration, the canonical harmonic map has the smallest Dirichlet energy \cref{eq:energy}. Thus local minimizers of \cref{eq:GL_functional} in the limit as $\varepsilon \rightarrow 0$ must themselves be canonical harmonic maps. In \cref{sec:behavior} we use the estimate \cref{eq:estimate} to establish some properties about the singularities of canonical harmonic cross fields. These results extend previous results on the asymptotic behavior of cross fields near singularities. A common assumption in the cross field meshing literature is that such properties hold for cross fields given as input to meshing algorithms. Since local minimizers are canonical harmonic maps, the results of \cref{sec:behavior} provide theoretical justification for making such assumptions when the input cross field is derived by a method that approximately minimizes the energy \cref{eq:GL_functional} in the limit as $\varepsilon \rightarrow 0$.

\subsection{Renormalized Energy, Singularity Location, and Index}\label{sec:renormalized}
From \Cref{thm:existence} we know that a \emph{global minimizer}, $u_\star$, of \cref{eq:ginzburg-landau-problem} will have isolated singularities occurring on the interior of $\widetilde{D}$, each with index $\sgn(d)$. But as previously noted, for practical applications such as meshing, we cannot rely on finding a global minimizer as the problem is non-convex. An asymptotic analysis reveals that even the singularities of \emph{local minimizers} of \cref{eq:GL_functional} as $\varepsilon \rightarrow 0$ must be simple, isolated, and occur on the interior of $\widetilde{D}$. The details of this argument are presented in \cite{bethuel_ginzburg-landau_1994}, but we include an outline here to keep the discussion self-contained. A similar discussion appears in \cite{beaufort_computing_2017}.

In the limit as $\varepsilon \rightarrow 0$, the dominating term of the energy \cref{eq:GL_functional} is multiplied by the sum of the squares of singularity indices. Thus a local minimizer of \cref{eq:GL_functional} can only have simple singularities, since the dominating term carries more energy for a singularity of index $n$ where $|n| > 1$ than it does for $n$ simple singularities. Beside the dominating term, the remaining \emph{renormalized energy} is given by two terms that depend only on singularity placement. The first is a logarithmic term which repels singularities of the same sign, and attracts singularities of opposite sign. The second term becomes infinite as singularities approach the boundary. All together, we conclude that the singularities of \emph{local minimizers} of \cref{eq:GL_functional} in the limit $\varepsilon \rightarrow 0$ must be simple, isolated and occur on the interior. Further, they are typically well distributed because of the logarithmic term of the renormalized energy. It is interesting to note that local minimizers can admit singularities of opposite signs, even though $u_\star$ does not. These results apply directly to $N$-direction fields because the singularities of an $N$-direction field occur in the same locations as those of their representation fields, only with index multiplied by $1/N$.

\section{Cross Field Topology and Quad Mesh Structure}\label{sec:topology}
In this section, we  make rigorous the connection between the topology of a harmonic cross field and the structure that can be extracted from it for use in building a quad mesh on a domain. Central to this idea is the relation between the index of a cross field singularity, and the number of separatrices meeting at that singularity, which in turn determines the valence of a node in the quad mesh.

In \cref{sec:behavior} we generalize previous results on $N$-direction field singularities by studying the asymptotic properties of a singularity of a boundary-aligned canonical harmonic $N$-direction field, not necessarily a minimizer of the Ginzburg-Landau energy.  \Cref{lem:int_singularity_partition,lem:bnd_singularity_partition} use an asymptotic estimate provided by the Ginzburg-Landau theory to make explicit the relationship between index and the number and distribution of separatrices meeting at a singularity. \Cref{lem:sector_index} shows that the local topology of the cross field in each sector between separatrices is identical to that of a constant field on a $90^{\circ}$ corner. These lemmas culminate in a proof that the separatrices of a cross field decompose the geometry into four-sided regions. Our results are illustrated in \cref{fig:LocalAnalysis}. For the sake of generality, the results in the following section are stated in terms of $N$-direction fields.

\subsection{Behavior of $N$-Direction Fields Near Singularities}\label{sec:behavior}

\begin{figure}
\begin{center}
\includegraphics[width=.32\linewidth, trim=50 -30 300 0, clip]{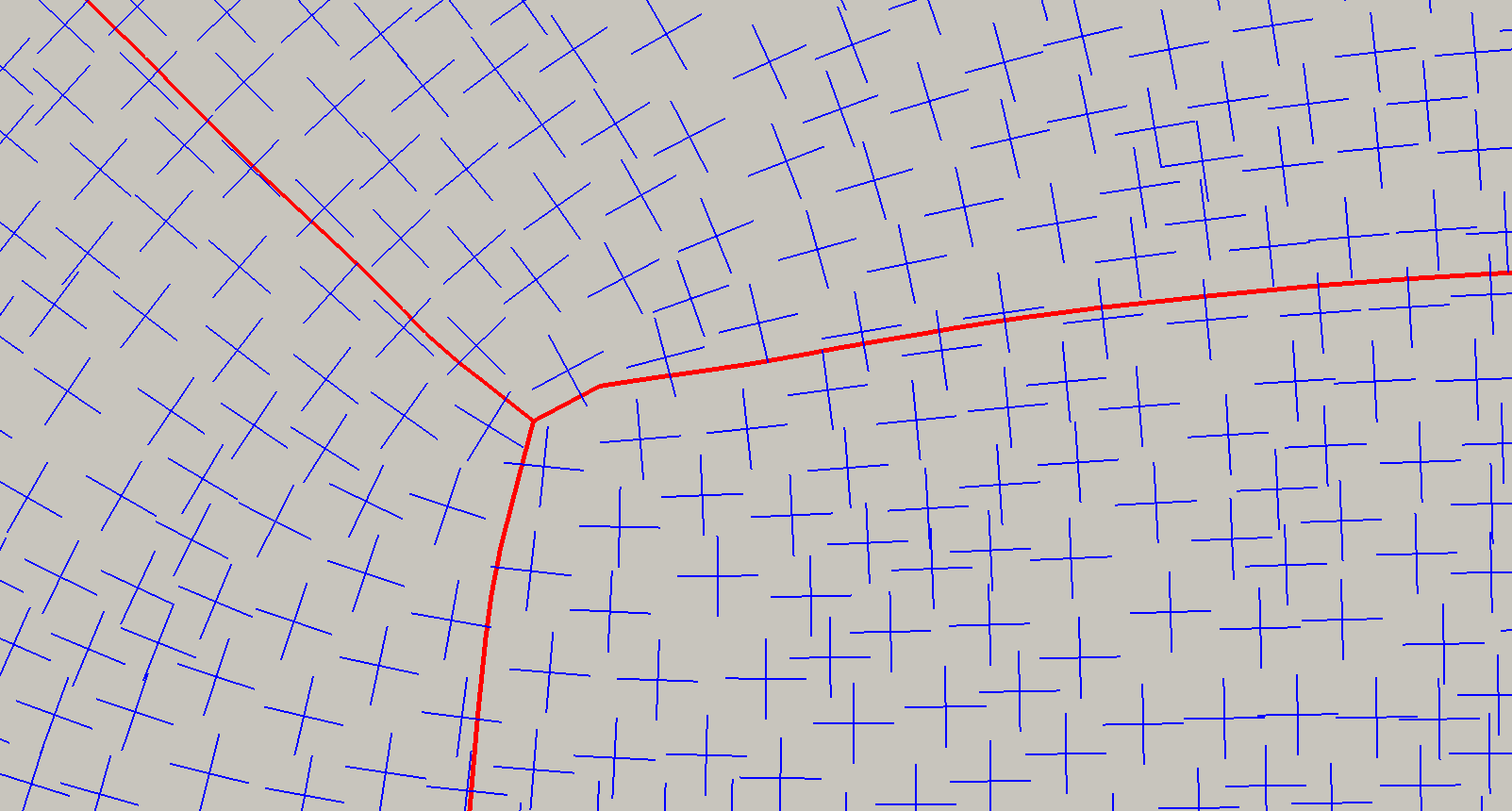} \quad
\includegraphics[width=.32\linewidth, trim=90 -30 260 0, clip]{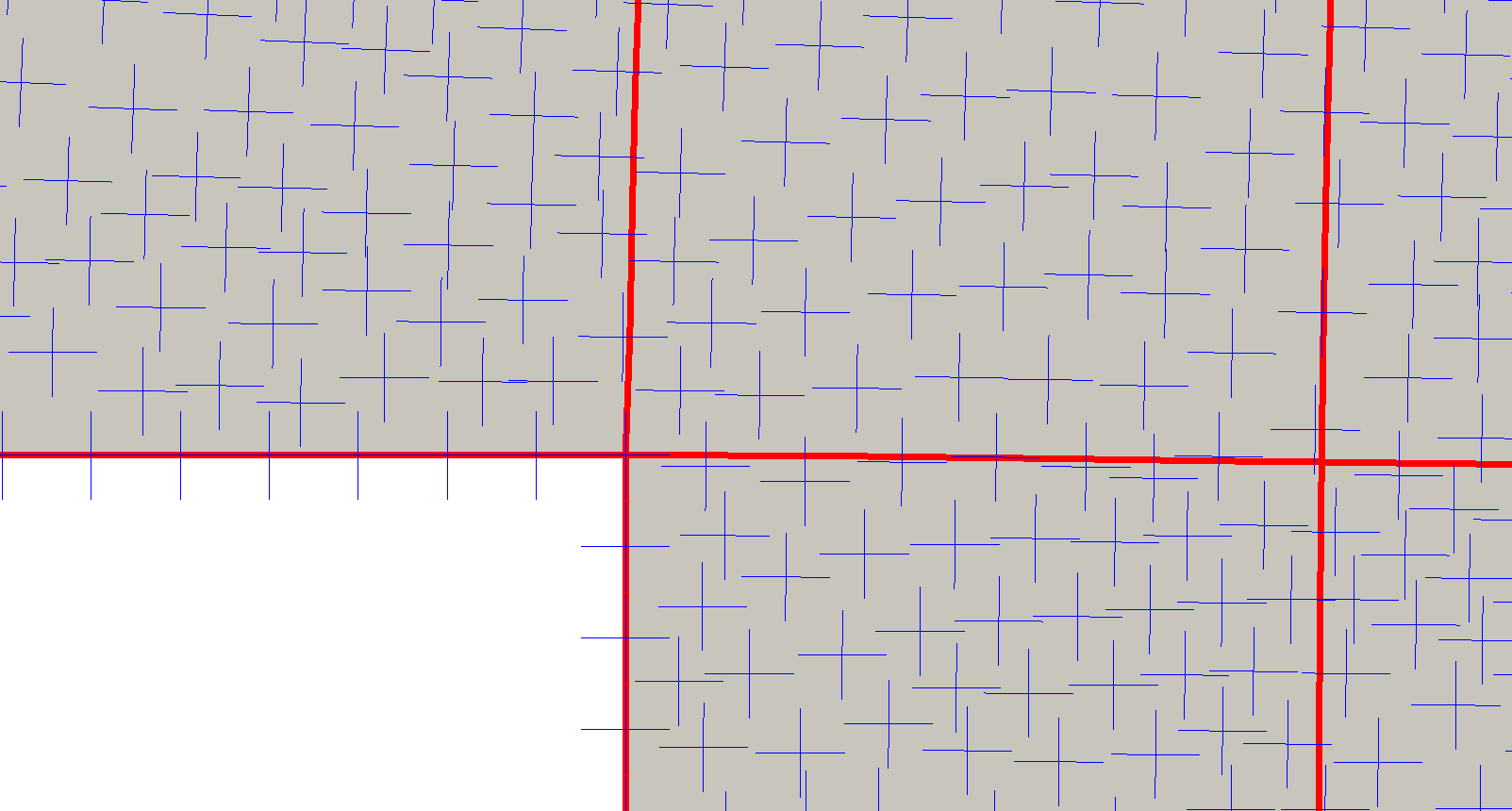} \\
\includegraphics[width=.32\linewidth, trim=0 0 -70 -30, clip]{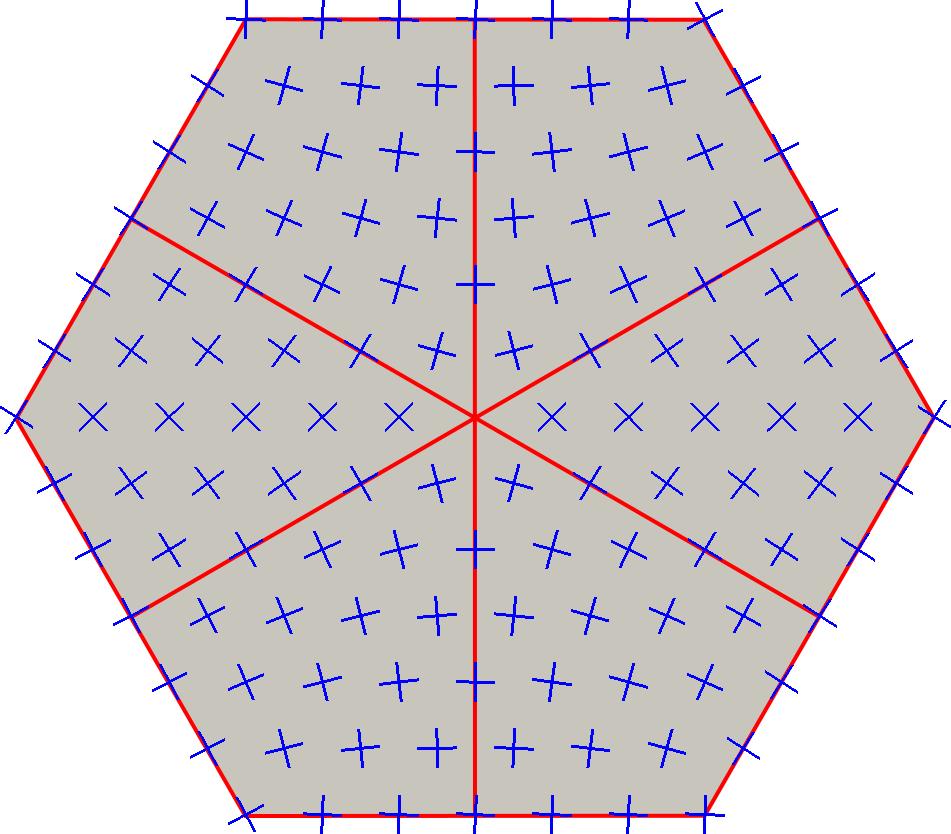} \quad
\includegraphics[width=.32\linewidth, trim=-70 0 0 -30, clip]{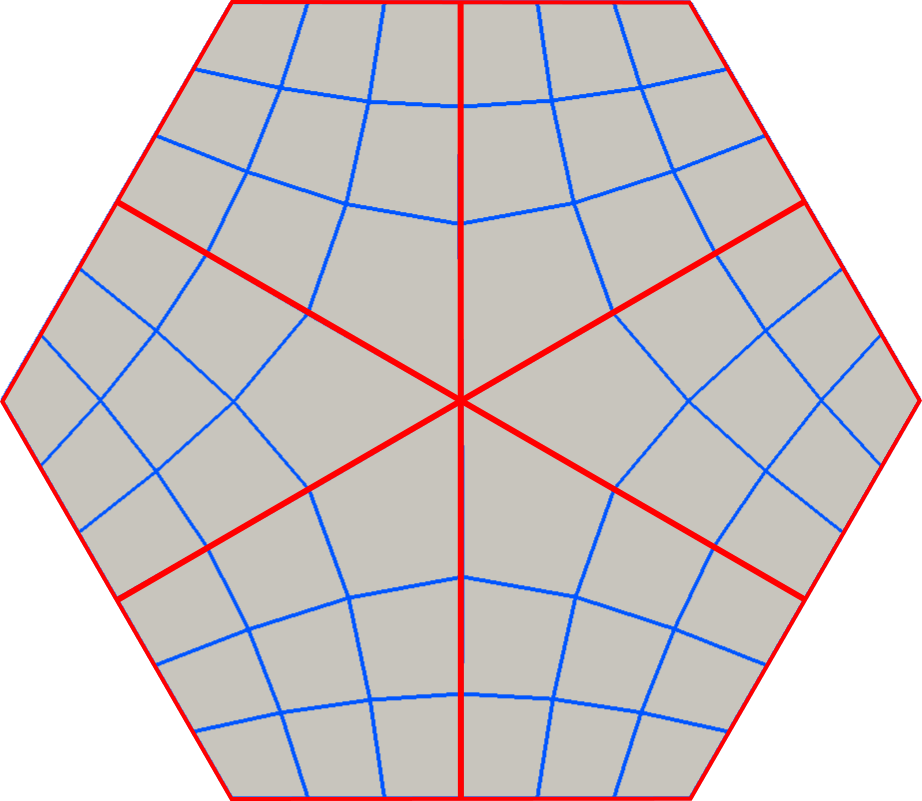}
\end{center}
\caption{{\bf (top)} Local behavior of cross fields around an interior singularity ({\bf left}) and boundary singularity ({\bf right}). The thick red lines show the separatrices exiting singularities as described by \cref{lem:int_singularity_partition,lem:bnd_singularity_partition}. Locally, the cross field on each sector is identical to a constant cross field on a $90^{\circ}$ corner, stated another way, if the separatrices are considered as boundaries of the sectors, then the index of each corner is $+1/4$; see \cref{def:boundary_singularity}. {\bf (bottom left)} A cross field and separatrix partition for a regular hexagon. The singularity in the center has index $-1/2$. {\bf (bottom right)} The corresponding quad mesh with skeleton highlighted in red; see \cref{sec:quad-meshing-def}.}
\label{fig:LocalAnalysis}
\end{figure}

Our first result is that singularities locally partition an $N$-direction field into evenly angled sectors. The number of sectors depends only on the index of the singularity.

\begin{lemma}\label{lem:int_singularity_partition}
  Let $f$ be a boundary-aligned canonical harmonic $N$-direction field on $D$. Let $a$ be an interior singularity of $f$ of index $d/N$ with $d < N$. There are exactly $N - d$ separatrices meeting at $a$. These separatrices partition a neighborhood of $a$ into $N - d$ evenly angled sectors.
\end{lemma}

While the result in \cref{lem:int_singularity_partition} is often leveraged \cite{kowalski_domain_2013,myles_robust_2014,campen_quantized_2015}, relatively little work has been done to disclose its exact nature. To the best of our knowledge this was first studied in \cite{palacios_rotational_2007}. By linearizing around a singularity, they show that the the angles at which separatrices exit a singularity must occur at the zeros of an $N+1$ degree polynomial, bounding the number of separatrices above by $N+1$. Kowalski et al. \cite{kowalski_pde_2013} show that for a discrete cross field interpolated linearly over triangle mesh elements, singularities of index $\pm 1/4$ will have 3 and 5 separatrices respectively, but do not address the question for the underlying continuous field. Beaufort et al. \cite{beaufort_computing_2017} identify the relationship of $N-d$ separatrices for a singularity of index $d/N$, and show that this relationship holds for cross fields that are \emph{a priori} aligned to a quad mesh. Here we prove the relationship holds for singularities of a canonical harmonic $N$-direction field, including local minimizers of the Ginzburg-Landau energy \cref{eq:GL_functional}.

\begin{proof}[Proof of \cref{lem:int_singularity_partition}] Let  $u$ be the representation vector field for $f$. Write $z = a + re^{i\theta}$. The estimate \cref{eq:estimate} gives
  \begin{equation}\label{eq:linearization}
      u(z) = \alpha e^{i d\theta} + o(r) \qquad \textrm{for } \theta \in [0,2\pi) \textrm{ and } |\alpha| = 1.
  \end{equation}
  We seek directions where the vector $\mathbf{v} = z - a$ is parallel to any of the component vectors of the cross $f(z) = u(z)^{\frac{1}{N}}$. Thus, writing $\alpha = e^{i \theta_0/N}$, we want to solve the equation
\begin{equation}\label{eq:interior-separatrices}
e^{i\theta} = \lim_{r \downarrow 0}u(z)^{\frac{1}{N}} = e^{i(\frac{d\theta + \theta_0}{N} + \frac{2\pi k}{N})}
\quad \implies \quad
\theta = 2\pi k/(N-d) + \theta_0/(N-d)
\end{equation}
for $k \in \mathbb{Z}$. Since we are looking for solutions where $\theta \in [0,2\pi)$ we have exactly $N - d$ solutions. This gives $N - d$ separatrices and $N - d$ sectors, each with interior angle $2 \pi / (N - d)$.
\end{proof}

\begin{lemma}\label{lem:bnd_singularity_partition}
  Let $f$ be the extension of a canonical harmonic direction field $\widetilde{f}$. Let $c$ be a boundary singularity of $f$ of index $d/N$ with $d < N/2$. There are exactly $\frac{N}{2} - d + 1$ separatrices meeting at $c$ (including the boundaries themselves). These separatrices partition a neighborhood of $c$ into $\frac{N}{2} - d$ evenly angled sectors.
\end{lemma}

\begin{proof}
  Let $\phi_c = \interior(c)$, the interior angle at $c$. By the definition of a boundary singularity (\cref{def:boundary_singularity}), we calculate that $\Delta \arg(c) = N(\pi - \phi_c) - 2\pi d$. We can parameterize the geometry near the corner with polar coordinates $z = re^{i \phi}$, where $r$ is the distance between from the corner and $\phi \in [0,\phi_c]$ measures the angle from the segment succeeding the corner. In this coordinate system, the representation field, $u$, near the corner can be written
\begin{equation}\label{eq:bnd_linearization}
u(z) = e^{-i\phi\frac{\Delta \arg(c)}{\phi_c}} = e^{-i \phi \frac{N(\pi - \phi_c) - 2\pi d}{\phi_c}}.
\end{equation}
As in \cref{lem:int_singularity_partition}, we are looking for values of $\phi$ where $\mathbf{v} = z - c$ is parallel to a component vector of the cross $f(z) = u(z)^{\frac{1}{N}}$. Thus, we want to solve
\begin{equation}\label{eq:bnd_separatrices}
e^{i\phi} = u(z)^{\frac{1}{N}} = e^{-i\phi\frac{N(\pi - \phi_c) - 2\pi d}{N\phi_c} + \frac{2 \pi k}{N}}
\quad \implies \quad
\phi = \frac{\phi_c}{\frac{N}{2}-d}k
\end{equation}
for $k \in \mathbb{Z}$. Since we are looking for solutions where $\phi \in [0,\phi_c]$, for $d < N/2$ and $N$ even we have solutions for $k = 0, \ldots ,\frac{N}{2}-d$. Thus, $\frac{N}{2} - d + 1$ separatrices partition a neighborhood of the corner $c$ into $\frac{N}{2}-d$ sectors, each with interior angle $\phi_c / (\frac{N}{2} - d)$.
\end{proof}

In \cref{lem:bnd_singularity_partition}, if we allow $d = N/2$ then \cref{eq:bnd_separatrices} simplifies to the identity, and so for every value of $\theta$ we get a separatrix. In the case of a cross field, this means that any boundary singularity with index $+1/2$ has infinitely many separatrices; see \cref{fig:degenerate-singularities} (left). Likewise, \cref{eq:interior-separatrices} reduces to the identity when $d = N$ and $\theta_0 = 0$.

\begin{figure}
  \begin{center}
    \resizebox{\linewidth}{!}{%
    \includegraphics[scale=1, trim=0 0 0 0, clip]{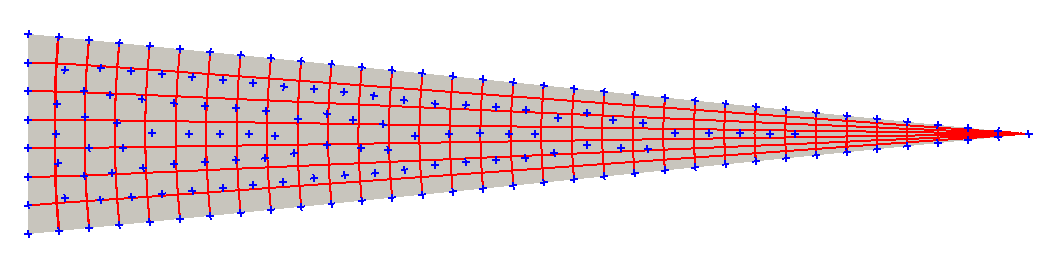}
    \includegraphics[scale=.75, trim=0 0 0 0, clip]{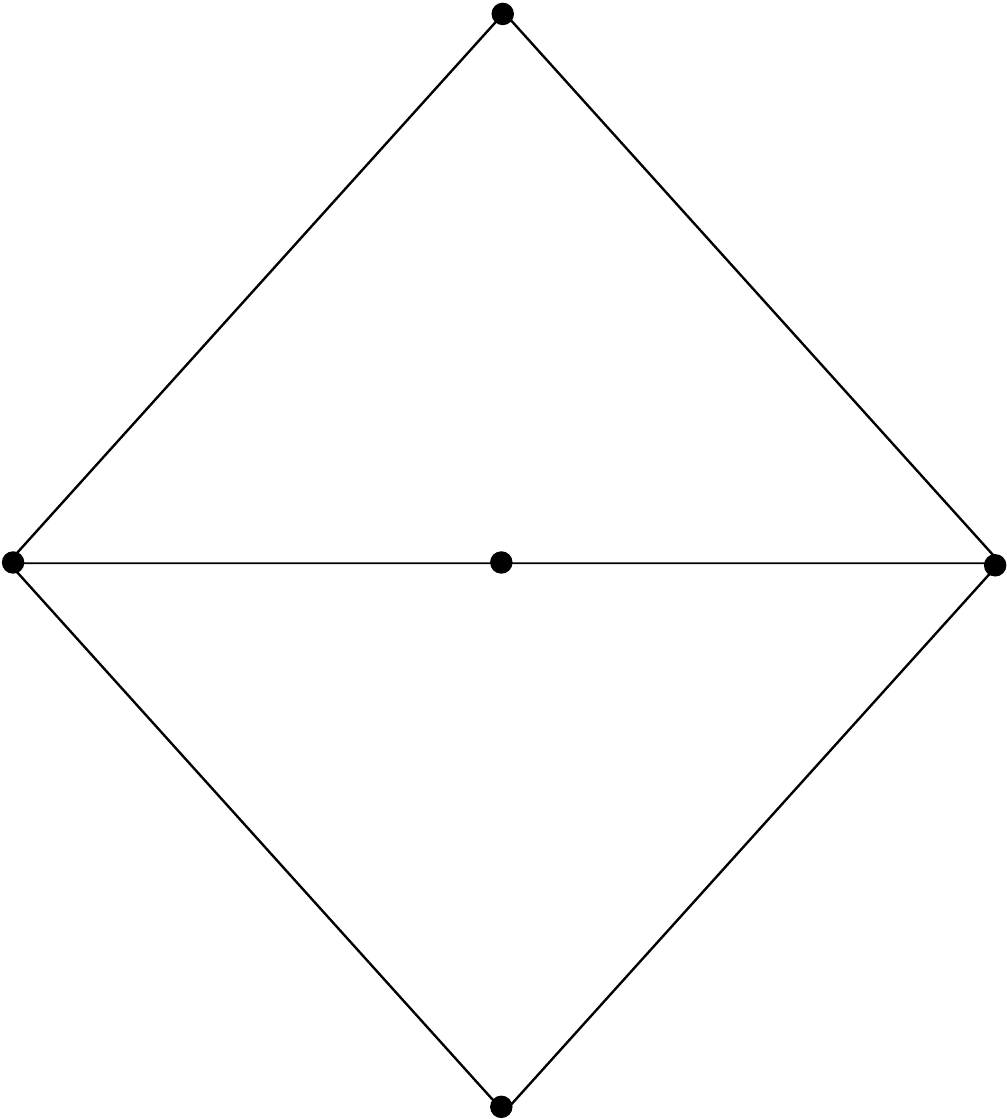}}
  \end{center}
\caption{Singularities of index $\geq 1/2$ lead to degenerate quad meshes. {\bf (left)} A boundary singularity of index $1/2$ leads to an infinite number of separatrices converging to a single point. {\bf (right)} An internal singularity of index $1/2$ leads to a pair of the degenerate ``doublet'' quads.}
\label{fig:degenerate-singularities}
\end{figure}

The following Lemma specifies the behavior of the cross field within each sector around either a boundary or an interior singularity.

\begin{lemma}\label{lem:sector_index}
Let $N$ be even and $d < N/2$. Consider a component of the partition described in \cref{lem:int_singularity_partition,lem:bnd_singularity_partition}.
A singularity, $c$, in the corner of this component, when viewed as a boundary singularity, has index $ \frac{1}{2} - \frac{1}{N}$.
\end{lemma}

\begin{proof}
If the singularity is an internal singularity, then \cref{eq:interior-separatrices} gives that on each of the sectors
\[ \interior(c) = \frac{2\pi}{N-d}. \]
From \cref{eq:linearization}, we compute
\begin{align*}
-\Delta \arg(c) &= d\left[\frac{2\pi(k+1)}{N-d} + \frac{\theta_0}{N-d}\right] + \theta_0 - \left(d\left[\frac{2\pi k}{N-d} + \frac{\theta_0}{N-d}\right] + \theta_0\right) = d\frac{2\pi}{N-d}.
\end{align*}
Using \cref{eq:boundary_index}, the index of the corner is
\[
\Index(c) = \frac{\pi - \frac{2\pi}{N-d} - \frac{d}{N}\frac{-2\pi}{N-d}}{2\pi} = \frac{\frac{N}{2}-1}{N}.
\]

If the singularity is a boundary singularity, then from eq.~\cref{eq:bnd_separatrices},
\[ \interior(c) = \frac{\phi_c}{\frac{N}{2}-d}. \]
Using \cref{eq:bnd_linearization}, we compute
\[
\Delta \arg(c) = \frac{N(\pi - \phi_c) - 2\pi d}{\phi_c}\left(\frac{\phi_c (k+1)}{\frac{N}{2} - d} - \frac{\phi_c k}{\frac{N}{2} - d}\right) = \frac{N(\pi - \phi_c) - 2\pi d}{\frac{N}{2} - d}.
\]
Thus \cref{eq:boundary_index} gives that the index of the corner is
\[
\Index(c) = \frac{\pi - \frac{\phi_c}{\frac{N}{2}-d} - \frac{1}{N}\frac{N(\pi - \phi_c) - 2\pi d}{\frac{N}{2} - d}}{2\pi} = \frac{\frac{N}{2} - 1}{N}.
\]
\end{proof}

\Cref{lem:sector_index} says that the index of the boundary singularity around the corner of any sector is $(\frac{N}{2} - 1)/N$. When $N = 4$, this simply means that each corner of a sector looks like a right angle with respect to the cross field. Stated another way, the cross field in each sector allows for a local $(u,v)$ parameterization; see \cref{fig:LocalAnalysis}.

The next section shows how we can use the local structure of singularities to determine the structure of a quad layout on $D$.

\subsection{Partitioning into Four-Sided Regions}
The skeleton of a quad mesh (\cref{sec:quad-meshing-def}) gives the basic structure of the mesh in the sense that it partitions the domain into the coarsest possible quad layout for the given choice of irregular nodes and connectivity between them; see \cref{fig:overview} (bottom), \cref{fig:LocalAnalysis} (bottom left), and \cref{fig:ExampleMeshes} (right). Any mesh with this structure, including the original mesh, can be seen as simply a refinement of this quad layout.

It is well known that a cross field can be generated on a quad mesh by locally aligning the crosses with quad edges (see \cite{beaufort_computing_2017,wang_frame_2016}). Beaufort et al. \cite{beaufort_computing_2017} show that a cross field created in such a way will have singularities exactly at the irregular nodes. Further, the separatrices of the cross field will be exactly the curves traced out by the skeleton of the quad mesh. Presumably, one could create the skeleton of a quad layout by reversing this process, i.e., simply tracing separatrices of the cross field (this is the approach taken in \cite{kowalski_pde_2013}). This would allow for meshing the domain by conformally mapping a regular grid into each region of the quad layout. Unfortunately, it is not always so simple. \Cref{fig:LimitCycle} shows a geometry that contains a limit cycle. The yellow separatrix begins at the corner, but continues indefinitely as it approaches the limit cycle. The following theorem shows that this is the only type of failure case that can occur.  The second part of the proof follows \cite{myles_robust_2014}.

\begin{figure}
\begin{center}
\includegraphics[width=.45\linewidth, trim=300 0 300 0, clip]{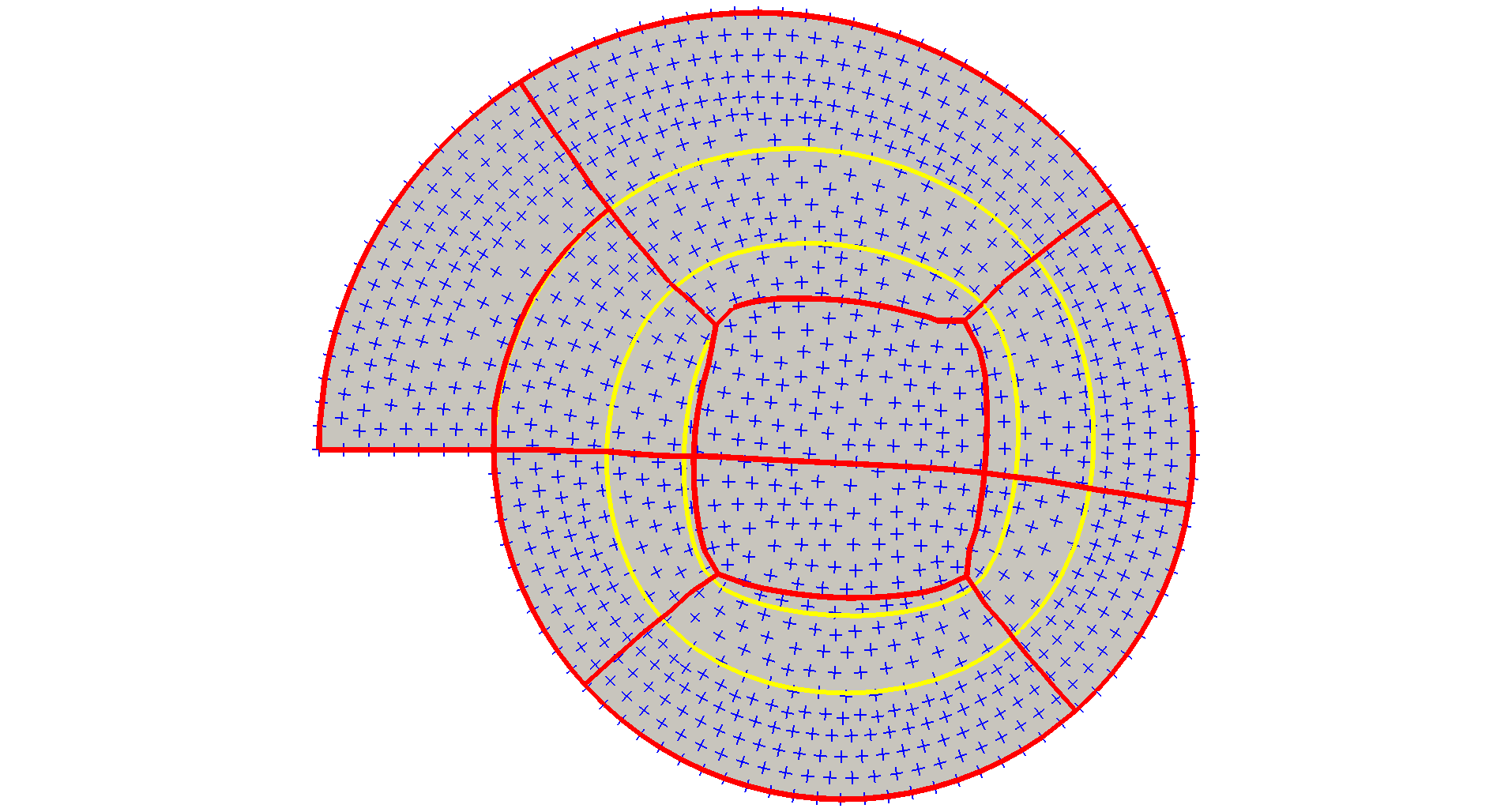}
\includegraphics[width=.45\linewidth, trim=300 0 300 0, clip]{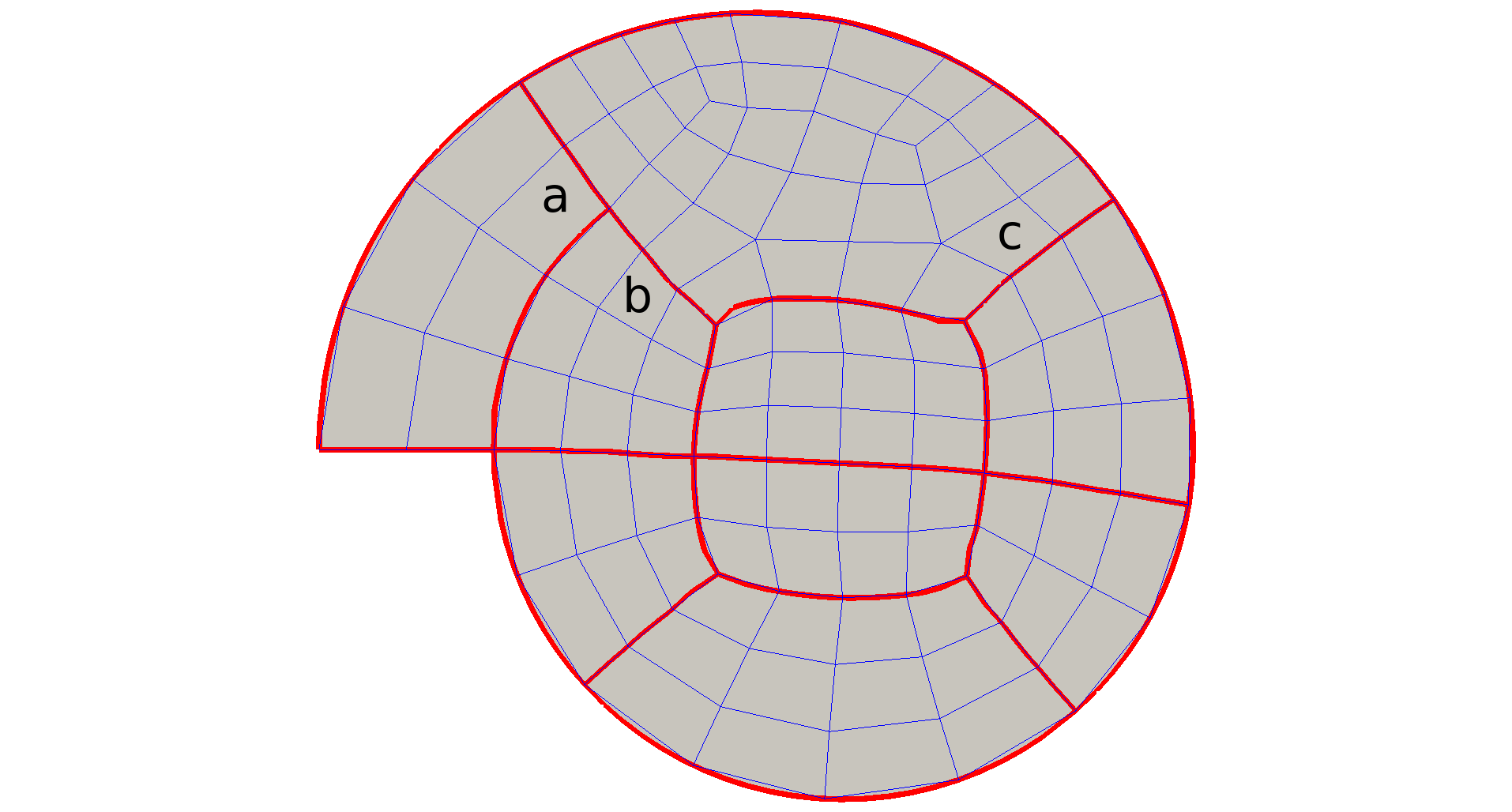}
\end{center}
\caption{A cross field on a domain that exhibits a limit cycle. {\bf (left)} The separatrix traced in yellow converges to a limit cycle. The partition shown in red inserts a T-junction at the first place where the yellow separatrix intersects another one. {\bf (right)} The four-sided regions without a T-junction can be meshed by conformal mapping. The region with the T-junction requires the insertion of more irregular nodes in order to conform with the mesh on its right and left sides.}
\label{fig:LimitCycle}
\end{figure}

\begin{theorem}\label{thm:quad-partitioning}
  Let $f$ be a boundary-aligned canonical harmonic cross field on $D$ whose singularities have index $\leq 1/4$. If no separatrix of $f$ converges to a limit cycle, then the separatrices of $f$, along with $\partial D$ partition $D$ into a quad layout.
\end{theorem}

\begin{proof}
  By \cref{lem:int_singularity_partition}, a finite number of separatrices meet at each singularity. Since no separatrix converges to a limit cycle on $D$ (and consequently, none on $\mathcal{R}$) the Poincar\'e-Bendixson theorem for manifolds \cite{schwartz_generalization_1963} guarantees that each separatrix must either end at another singularity, or exit the domain orthogonal to the boundary. The set of separatrices along with $\partial D$ then partition the domain into bounded regions that contain no singularities. If a curve of the boundary of any region meets another curve, the corner where they meet must have index $+1/4$ because they are either separatrices intersecting each other or the boundary orthogonally, or they meet at singularities and by \cref{lem:sector_index} have index $+1/4$. Since there are no internal singularities, the total index must come from corner singularities, and since the index for each corner is positive, the sum must be non-negative. By the Poincar\'e-Hopf theorem for cross fields \cite{ray_n-symmetry_2008}, the total index of a given region must equal the Euler characteristic of that region. The genus of each region is zero because the domain is defined in two dimensions. Thus there are only two possibilities; either there is one boundary and the Euler characteristic is one, in which case we have four corners each of index +1/4, a quad element. Otherwise, there are two boundaries and the Euler characteristic is zero. In this case the total index is zero, so there are no singularities, i.e., an annulus.
\end{proof}

The main takeaway is that when no separatrix of the cross field converges to a limit cycle, the topology of the cross field is sufficient to uniquely determine a quad layout, which can then easily be meshed by mapping a regular grid into each region. The irregular nodes of any such mesh mirror exactly the singularity structure of the cross field through the relationship established in \cref{lem:int_singularity_partition,lem:bnd_singularity_partition}. Further, since the index of each singularity is $\leq 1/4$, each singularity will have at least three separatrices, and so none of the quad elements of the mesh will be degenerate; see \cref{fig:degenerate-singularities}.

The failure case occurs when one of the separatrices converges to a limit cycle. In a case like this, we propose that such separatrices can easily be handled by allowing T-junctions on the quad layout; see \cref{fig:LimitCycle}. This complicates the meshing problem slightly; a domain can no longer be meshed by simply mapping a regular grid into each region of the quad layout. Instead, regions adjacent to T-junctions will require additional singularities to resolve the differing number of quads needed on opposite sides of the region. This however has been addressed in \cite{_cubit_2017,takayama_robust_2013,takayama_pattern-based_2014}, and so partitioning the domain into a quad layout with T-junctions is sufficient to produce a quad mesh where the number, location, and valence of its irregular nodes are determined only by the singularities and T-junctions in the partition, and the target density of the quad mesh. See \cref{sec:meshing} for more details.

It is not immediately clear that a separatrix can always be cut off by another separatrix to form a T-junction. The question is even more delicate on 2-manifolds with genus greater than zero. Though not actually used in the robust implementation of their algorithm, in the proof of their main result, Myles et al. \cite{myles_robust_2014} rely on including sections of a sufficiently fine background grid in their partition in order to guarantee that separatrices will terminate. The following two results show that in two dimensions, separatrix tracing alone is sufficient to guarantee a partition into four-sided regions.

\begin{theorem}\label{thm:must_intersect}
Let $f$ be a smooth boundary-aligned cross field on $D$ with a finite singularity set. Every separatrix that converges to a limit cycle intersects at least one other separatrix.
\end{theorem}
\begin{proof}
  Let $s \colon [0,\infty) \rightarrow D$ be a separatrix beginning at a singularity $s(0)$ with $s(t)$ converging to the limit cycle, $\gamma_\infty$, as $t \rightarrow \infty$. Since $s$ converges to $\gamma_\infty$, there exists a $t^*$ such that a streamline intersecting $s$ at $s(t^*)$ must also intersect $\gamma_\infty$. Let $s_{\perp,t^*}^+$ be the segment of this streamline beginning at $s(t^*)$ and continuing in the direction towards $\gamma_\infty$. Consider the the family of streamline segments $s_{\perp,t}^+$ for $t \in [0,\infty)$ beginning at $s(t)$, and continuing in the direction from $s$ consistent with $s_{\perp,t^*}^+$. Let this family of curves be parameterized so that each curve, $s_{\perp,t}^+(r)$, has unit speed and starts on $s$ at $s(t)$ when $r = 0$. Finally, let $s_{\perp,t}$ be the corresponding streamlines.

  If $s_{\perp,t}^+$ intersects $\gamma_\infty$, for every $t$, then $s_{\perp,0}^+$ must also be a separatrix (since it starts at a singularity), and since it intersects $\gamma_\infty$ it must also intersect $s$. If $s_{\perp,t}^+$ does not intersect $\gamma_\infty$, for all $t$, then there is a greatest lower bound, $\tau$, such that for $t > \tau$, $s_{\perp,t}^+$ must intersect $\gamma_\infty$.

  We claim that $s_{\perp,\tau}^+$ cannot intersect $\gamma_\infty$. If it did, then by the stability of ordinary differential equations with respect to initial data, \cite[Proposition 2.76]{chicone_ordinary_2006}, there would be a neighborhood around $s(\tau)$ within which any streamline would also intersect $\gamma_\infty$, contradicting the fact that $\tau$ is a greatest lower bound.

  We claim that $s_{\perp,\tau}^+$ connects to a singularity, making $s_{\perp,\tau}$ a separatrix, and argue by contradiction. If $s_{\perp,\tau}^+$ does not connect to a singularity, then by the Poincar\'e-Bendixson theorem, either $s_{\perp,\tau}$ is a periodic orbit (case 1), or $s_{\perp,\tau}^+$ will exit the boundary (case 2) or approach a limit cycle (case 3).

  \underline{Case 1:} If $s_{\perp,\tau}$ is a periodic orbit, then since both $s_{\perp,\tau}$ and $\gamma_\infty$ are closed sets, there is a finite distance, $\delta>0$, between $s_{\perp,\tau}$ and $\gamma_\infty$. Let $V = \{x \in D \colon d(x,s_{\perp,\tau}) < \delta/2\}$. By \cite[Proposition 2.76]{chicone_ordinary_2006}, there exists an $\varepsilon > 0$ such that $s_{\perp,t}^+$ must remain in $V$ for $t \in (\tau-\varepsilon,\tau+\varepsilon)$. Then $s_{\perp,\tau+\varepsilon/2}^+$ must remain in $V$, and on the other hand, must intersect $\gamma_\infty$ since $\tau + \varepsilon/2 > \tau$. This is a contradiction since $V \cap \gamma_\infty$ is empty.

  \underline{Case 2:} Suppose that $s_{\perp,\tau}^+$ exits the boundary. Since $s_{\perp,\tau}^+$ and $\gamma_\infty$ are a compact sets, there exists a $\delta>0$ such that the set $V = \{x \in D \colon  d(x,s_{\perp,\tau}^+) < \delta/2\}$ has an empty intersection with $\gamma_\infty$. Again by \cite[Proposition 2.76]{chicone_ordinary_2006} there is an $\varepsilon>0$ such that $s_{\perp,t}^+$ must remain in $V$ for all $r \geq 0$ when $t \in (\tau-\varepsilon,\tau+\varepsilon)$. But $s_{\perp,t+\varepsilon/2}^+$ intersects $\gamma_\infty$, which again is a contradiction since $V \cap \gamma_\infty$ is empty.

  \underline{Case 3:} Suppose $s_{\perp,\tau}^+$ approaches a limit cycle $\ell_\infty$. Then there exists an $r^*$ and an $\varepsilon_1>0$ such that if $t \in (\tau-\varepsilon_1,\tau+\varepsilon_1)$, then $s_{\perp,t}^+(r)$ approaches $\ell_\infty$ asymptotically as $r$ increases from $r^*$ to $\infty$. Consider the segment of $s_{\perp,\tau}^+(r)$ for $r \in [0,r^*]$. This segment is a closed set, so again there is a $\delta$ such that the set $V = \{x \in D \colon d(x,s_{\perp,\tau}^+(r)) < \delta/2 \text{ for } r \in [0,r^*] \}$  has an empty intersection with $\gamma_\infty$. Again, by \cite[Proposition 2.76]{chicone_ordinary_2006}, there is a $\varepsilon_2>0$ such that $s_{\perp,t}^+$ must remain in $V$ for all $r \in [0,r^*]$ when $t \in (\tau-\varepsilon_2,\tau+\varepsilon_2)$. Let $\varepsilon = \min(\varepsilon_1,\varepsilon_2)$. Then for $t \in (\tau-\varepsilon,\tau+\varepsilon)$, $s_{\perp,t}+(r)  \in V$ for $r < r^*$, and approaches $\ell_\infty$ asymptotically afterwards. Thus $s_{\perp,t+\varepsilon/2}^+$ cannot intersect $\gamma_\infty$, which again is a contradiction.
\end{proof}

\begin{figure}
  \begin{center}
    \includegraphics[width=.65\linewidth, trim=0 0 0 0, clip]{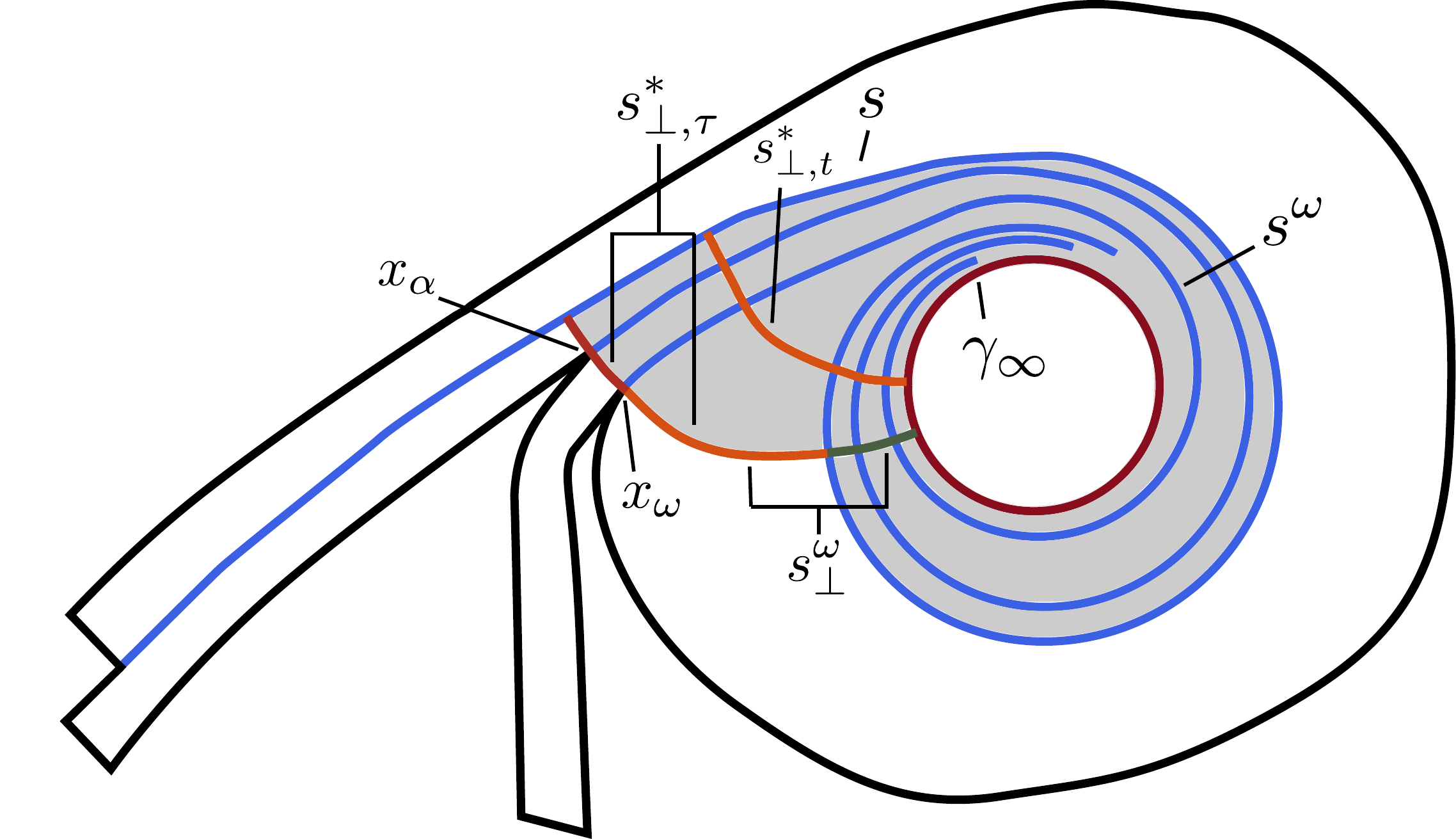}
  \end{center}
\caption{An illustration of case B in the proof  of \cref{cor:orthogonal-separatrix}. The area filled in with grey is the set A. }\label{fig:orthogonal-separatrix}
\end{figure}

\begin{corollary}\label{cor:orthogonal-separatrix}
  Let $f$ be a boundary-aligned canonical harmonic map for some singularity configuration where each singularity has index $\leq 1/4$. Let $\gamma_\infty$ be a limit cycle with a separatrix, $s$, converging to it. There is a separatrix, $s'$, (not necessarily the same as $s$) that converges to $\gamma_\infty$ and intersects a separatrix, $s'_\perp$, that begins at the same singularity as $s'$ and intersects $\gamma_\infty$.
\end{corollary}

\begin{proof}
Let $s_{\perp,t}^+$ be  defined as in the proof of \cref{thm:must_intersect}. Either $s_{\perp,t}^+$ intersects $\gamma_\infty$ for all $t$ (case A), or not (case B). In case A, by the proof of \cref{thm:must_intersect}, $s$ intersects $s_{\perp,0}$, which is a separatrix beginning at the same singularity and intersecting $\gamma_\infty$, so letting $s' = s$ and $s'_\perp = s_{\perp,0}$, the proof is complete. In case B, as in the proof of \cref{thm:must_intersect}, there is a greatest lower bound, $\tau$, such that for $t > \tau$, $s_{\perp,t}^+$ intersects $\gamma_\infty$ and $s_{\perp,\tau}^+$ connects to a singularity, $x_\alpha$. For $t > \tau$, let $s_{\perp,t}^*$ be the segment of $s_{\perp,t}^+$ between $s$ and $\gamma_\infty$. The segment $s_{\perp,t}^*$ does not contain a singularity for any $t > \tau$, otherwise $\tau$ would not be a greatest lower bound. Let $A = \cup_{t > \tau}{s_{\perp,t}^*}$. $A$ contains no singularities and any streamline in it running orthogonal to the segments $s_{\perp,t}^*$ must converge to $\gamma_\infty$.

Let $s_{\perp,\tau}^*$ be  the segment of $\boundary{A}$ that starts at $s(\tau)$, runs orthogonal to $s$, and ends at the first intersection of $\boundary{A}$ and $s$. This curve contains $x_\alpha$ and may contain other singularities. However, $s_{\perp,\tau}^*$ is aligned with the cross field,   so it is a piecewise smooth finite union of separatrices and segments of separatrices. Each singularity on $s_{\perp,\tau}^*$ has a separatrix entering $A$ and running orthogonal to every $s_{\perp,t}^*$, thus each of these separatrices must converge to $\gamma_\infty$. Consider $x_\omega$, the singularity that is farthest from $s(\tau)$ along $s_{\perp,\tau}^*$. Let $s^\omega$ be the separatrix beginning at $x_\omega$ converging to $\gamma_\infty$. The separatrix, denoted $s^\omega_\perp$, beginning at $x_\omega$ and continuing along $s_{\perp,\tau}^*$ intersects $s$, and therefore must intersect $\gamma_\infty$ since $s_{\perp,t}^*$ intersects $\gamma_\infty$ for all $t > \tau$.  Then setting $s' = s^\omega$ and $s'_\perp = s^\omega_\perp$ satisfies the statement of the theorem.
\end{proof}

\Cref{thm:must_intersect,cor:orthogonal-separatrix}  allow us to deal with the issue of limit cycles in a predictable manner. When a limit cycle occurs in the cross field, the partition can be obtained by first tracing out all of the separatrices that don't converge to a limit cycle, and then tracing out each separatrix that does converge to a limit cycle until it reaches another separatrix, placing a T-junction at that point. The exact process for partitioning a domain into a quad layout with T-junctions is specified in \cref{alg:partitioning}. We denote the set of separatrices that converge to any limit cycle by $\mathcal{P}$.

\begin{algorithm}[t!]
\caption{Partitioning $D$ into a quad layout with T-junctions.} \label{alg:partitioning}

\vspace{.2cm}

\begin{algorithmic}[t]
\STATE{\bfseries Input:} A domain $D$ satisfying \cref{as:dom}, and a boundary-aligned canonical harmonic cross field $f$ with singularities of index $\leq 1/4$.

\vspace{.2cm}

\STATE{\bfseries Output:} A set $\mathcal{B}$ containing limit cycles and separatrices that define a quad layout with T junctions.

\vspace{.2cm}

\STATE Let $\mathcal{S}$ be the set of separatrices that do not converge to a limit cycle. Let $\mathcal{P}$ be the set of separatrices that do. Let $\mathcal{L}$ be the set of limit cycles.

\vspace{.2cm}

\STATE Initialize the set $\mathcal{B} = \mathcal{S}$.

\vspace{.2cm}

\FOR {$l \in \mathcal{L}$}
  \IF {no element of $\mathcal{B}$ intersects $l$}
  \STATE (i) Add $l$ to $\mathcal{B}$.
  \STATE (ii) By \cref{cor:orthogonal-separatrix}, there is an element of $\mathcal{P}$ that intersects $l$. Let $\rho'$ be the portion of that separatrix beginning at the singularity and ending in a T-junction with $l$.
  \STATE (iii) Add $\rho'$ to $\mathcal{B}$.
  \STATE (iv) remove $\rho$ from $\mathcal{P}$.
  \ENDIF

\ENDFOR

\FOR {$\rho \in \mathcal{P}$}
\STATE Let $\rho'$ be the curve segment of $\rho$ beginning at the singularity and continuing until it intersects an element of $\mathcal{B}$. Add $\rho'$ to $\mathcal{B}$.
\ENDFOR

\end{algorithmic}
\end{algorithm}

\begin{theorem}\label{thm:partitioning_termination}
Given a domain D satisfying \cref{as:dom}, and a boundary-aligned canonical harmonic cross field f with singularities of index $\leq 1/4$, \Cref{alg:partitioning} is well-defined, terminates in finite time, and partitions $D$ into a quad layout with exactly $|\mathcal{P}|$ T-junctions.
\end{theorem}

\begin{proof}
The first `for' loop is well-defined by \cref{cor:orthogonal-separatrix}.

In the second `for' loop, each separatrix that converges to a limit cycle is guaranteed to intersect another separatrix by \cref{thm:must_intersect}. If it does not intersect a separatrix in $\mathcal{S}$, then it must intersect one that runs orthogonal to the limit cycle itself which is added to $\mathcal{B}$ in the first for loop.

Clearly \cref{alg:partitioning} partitions the domain into regions without singularities. As in \cref{thm:quad-partitioning}, each corner of the partition must have an index of $1/4$ because it either meets at a singularity, intersects another separatrix, or exits the boundary orthogonally. Following the argument from the second part of the proof of \cref{thm:quad-partitioning}, this is enough to guarantee that each region is a quad or an annulus. There are exactly $|\mathcal{P}|$ T-junctions because they are created exactly when we trace a separatrix in $\mathcal{P}$ until it reaches another separatrix that is already in the set $\mathcal{B}$.
\end{proof}

\section{Computational Methods} \label{sec:computational_methods}
  In this section we describe computational methods to practically apply the results of \cref{sec:ginzburg-landau,sec:topology}. We introduce a new method for cross field design, discuss its implementation on a discrete mesh, and compare it to other cross field design methods. We also implement \cref{alg:partitioning} in a discrete setting and use the resulting partitions to obtain quad meshes on various geometries.

\subsection{Cross Field Design and the MBO Diffusion Generated Method}\label{sec:MBO}
Another consequence of the connection between the cross field design problem and Ginzburg-Landau theory is an efficient computational method for minimizing the Ginzburg-Landau functional \cref{eq:GL_functional} based on a generalization of the Merriman-Bence-Osher (MBO) diffusion generated method.
The MBO diffusion generated method was originally introduced in \cite{merriman_computational_1992,merriman_diffusion_1993,merriman_motion_1994} in the context of mean curvature flow and extended to  \cref{eq:GL_functional} in  \cite{ruuth2001} to study the evolution of vortex filaments in three dimensions. The details of the algorithm are given in  \cref{alg:MBO}. We view this method as an energy splitting method for the Ginzburg-Landau functional \cref{eq:GL_functional}. The first term of \cref{eq:GL_functional} is the Dirichlet energy of the representation field, whose gradient flow is diffusion. Thus in the first step of  \cref{alg:MBO}, we diffuse for a short time $\tau>0$. The second term of \cref{eq:GL_functional} penalizes complex numbers that do not lie on the unit circle, $\mathbb{T}$, and the gradient flow for this term is
$$
u_t = \varepsilon^{-2} (|u|^2 - 1) u.
$$
The field evolves pointwise in the radial direction; i.e., the representation vectors lengthen or shorten radially toward $\mathbb{T}$. Since this term is multiplied by a factor $\varepsilon^2$, and $\varepsilon$ is small, we can rescale time as $\tilde t = t/\varepsilon^2$, giving $u_{\tilde t} =  (|u|^2 - 1) u$. As $\varepsilon \rightarrow 0$, $\tilde{\tau} = \tau/\varepsilon^2 \rightarrow \infty$ and the solution tends towards the closest point on $\mathbb{T}$. Thus, we approximate the second step of the energy splitting method with pointwise renormalization as in the second step of \cref{alg:MBO}.

The MBO method (\cref{alg:MBO}) provides an efficient way to approximate local minimizers for the Ginzburg-Landau energy \cref{eq:GL_functional}. By the discussion in \cref{sec:renormalized}, local minimizers will have isolated singularities of degree $\pm 1$. Such a field is a boundary-aligned canonical harmonic map for some singularity configuration, so the results in \cref{sec:topology} apply, and \cref{thm:quad-partitioning,thm:partitioning_termination} guarantee that its separatrices will partition the domain into four-sided regions.

\begin{remark}
While the original MBO diffusion generated method is known to converge to mean curvature flow, Laux and Yip \cite{laux_analysis_2018} have recently proven the codimension two result that the method converges to the gradient flow of \cref{eq:GL_functional}.
\end{remark}

\begin{algorithm}[t!]
\caption{A diffusion generated method algorithm for approximating minimizers of the Ginzburg-Landau energy \cref{eq:GL_functional}. } \label{alg:MBO}

\vspace{.2cm}

\begin{algorithmic}[t]
\STATE{\bfseries Input:} Let $D$ be a domain satisfying \cref{as:dom},  $\tau > 0$, and $\delta > 0$.
Fix boundary conditions $g$ on $\partial U$ as in \cref{def:boundary-aligned}.
Initialize a representation map of the cross field, $u_0\colon D \to \mathbb C$, with $u_0(x) = g(x)$ for every $x \in \partial D$. Let $k=0$.

\vspace{.2cm}

\WHILE {$k >0 $ and $\| u_k - u_{k-1} \| > \delta $,}
\STATE (i) Solve the diffusion equation,
\begin{subequations}
\begin{align}
\partial_t v(t,x) &= \Delta v(t,x) && x\in D \\
v(t,x) &= g(x) && x \in \partial D \\
\qquad v(0,x) &= u_k(x) && x \in D,
\end{align}
\end{subequations}
until time $\tau$. Denote the solution by $\tilde{u}_{k+1} = e^{\tau \Delta}u_k = v(\tau)$.

\bigskip

\STATE  (ii) Set $u_{k+1} = \frac{\tilde{u}_{k+1}}{ | \tilde{u}_{k+1} | } $ and $k = k+1$.

\ENDWHILE
\end{algorithmic}
\end{algorithm}

\subsection{Discrete Computation and Meshing} \label{sec:discrete}

We now turn our attention to an example implementation of our method. Given a domain $D$, we discretize it with a triangle mesh $M$ and use a $P1$ Lagrange basis. This assigns crosses to nodes of the mesh and singularities to the faces. For a canonical harmonic cross field with isolated singularities of degree $\pm 1$, a $P1$ basis is sufficient to represent the topology of the field without aliasing problems \cite{vaxman_directional_2016}. Further, the results on singularities of Kowalski et al. \cite{kowalski_pde_2013} apply to our discretization and thus the properties of singularities of canonical harmonic fields proven in \cref{sec:topology} still hold in the discrete setting.

We can approximate a canonical harmonic cross field with simple isolated singularities either by fixing the singularity locations and degrees and using \cref{eq:u_0} or via the MBO method. The former consists of a single solve of the real scalar Laplace equation (see \cref{thm:fixedField}).
To implement the MBO method, we initialize the field either with a canonical harmonic map using \cref{eq:u_0} or with a harmonic map obtained by solving the complex Laplace equation on $D$. In our experimentation we have also initialized with randomly assigned fields. We have not found an example domain where the initialization affects the final solution. We implement the diffusion steps of the MBO method using the backward Euler time discretization.
The sparse symmetric positive definite matrix appearing in this discretization is factored efficiently once and then reused at each iteration of the MBO method. This is an advantage over most other cross field methods, which require a full system solve at each iteration.

To choose the parameter $\tau$ in \cref{alg:MBO}, we use a time scale related to the characteristic length of the diffusion equation. Roughly speaking, the solution will evolve at length scales less than $\sqrt{\tau}$. Taking the characteristic length to be $\lambda_1^{-\frac 1 2}$, where $\lambda_1$ is the principle eigenvalue of the Dirichlet-Laplacian, we take $\tau = \lambda_1^{-1}$.

\Cref{fig:MBO_repVec} illustrates the time evolution of the MBO method until the solution becomes stationary on a half disk. We randomly initialize the field with a uniformly distributed unit vector field. The computation is performed on a mesh with 29,856 nodes with $\tau = \lambda_1^{-1}/10$. We choose a smaller value of $\tau$ here in order to slow down the evolution of the system in the first few MBO iterations. After 1600 iterations the field satisfies $\| u_{k} - u_{k-1} \|_{\ell^2} \leq  2n \times 10^{-15}$ where $n$ is the number of interior nodes of the mesh, however visually the solutions at each iteration are indistinguishable after $\approx$ 300 iterations.

\subsubsection{Comparison to Other Approaches} \label{sec:comparison}
In this section, we provide a comparison of the MBO method (\cref{alg:MBO}) to minimizing \cref{eq:GL_functional} directly using the L-BFGS method and to instant meshing \cite{jakob_instant_2015}.

For the L-BFGS method, we use the \verb+minlbfgsoptimize+ command in the numerical library \cite{alglib_2018} with default parameters.
To initialize \cref{alg:MBO} and the L-BFGS methods, we use the solution to the Laplace equation with boundary-aligned boundary conditions,
\begin{align*}
& \Delta u = 0  && \Omega \\
& u = g && \partial \Omega.
\end{align*}
For a convergence criterion for \cref{alg:MBO}, and L-BFGS, we use the condition
$\| u_{k} - u_{k-1} \|_{\ell^2} \leq  2n \times 10^{-4}$,
where $n$ is the number of nodes in the triangle mesh. Here we have chosen a relatively large convergence criterion, which is appropriate for meshing.

For instant meshes, we use the implementation at the interactive geometry lab website\footnote{\url{http://igl.ethz.ch/projects/instant-meshes/}} with the intrinsic energy and boundary alignment options selected  \cite{jakob_instant_2015}.
We ran instant meshes on a single thread so that the times are comparable with the other methods.

The number of iterations and wall-clock time for each method on several discretizations of various models are tabulated in \cref{t:MethComp}. The instant meshes method and \cref{alg:MBO} are roughly the same speed and an order of magnitude faster than the L-BFGS method.

A comparison of the singularity placements for a representative simulation on the gear model is given in \cref{fig:comparison}.
The singularity placements for \cref{alg:MBO} and L-BFGS are nearly the same.
Since instant meshes performs local iterations, the solution has not converged and this causes poor singularity placement for boundary-aligned fields.

\begin{table}
\centering
{\small
\begin{tabular}{l l ll c ll c ll}
& & \multicolumn{2}{c}{{\bf MBO}} & \ &\multicolumn{2}{c}{{\bf L-BFGS}} & \ &\multicolumn{2}{c}{{\bf instant meshes}}  \\
\cline{3-4} \cline{6-7} \cline{9-10}
Model       & nodes & iters  & time && iters & time && iters & time  \\
\hline
\bf{half disk}   & 269   & 47    & 0.0017 && 10 & 0.0425 && 66 & 0.0092    \\
\hline
{\bf disk}  & 264   & 223  & 0.0072 && 17 & 0.0280   && 66 & 0.0116 \\
            & 5967   & 16  & 0.0370 && 46 & 1.0724   && 96 & 0.0411 \\
\hline
{\bf gear}  & 6375   & 34  & 0.0289 && 8 & 0.8484   && 102 & 0.057 \\
            & 108787   & 14  & 0.6843 && 29 & 27.938   && 252 & 1.853 \\
\hline
{\bf nautilus}  & 964   & 95  & 0.0156 && 17 & 0.0996   && 78 & 0.017 \\
\end{tabular} }
\caption{A comparison of three methods (MBO, L-BFGS, and instant meshes) on four different models for different discretizations. For each iterative method, we report the number of iterations and the time (in seconds) for the computation. Details are given in \cref{sec:comparison}.}
\label{t:MethComp}
\end{table}

\begin{figure}
  \begin{center}
    \includegraphics[width=1\linewidth, trim={0 5cm 0 10cm}, clip]{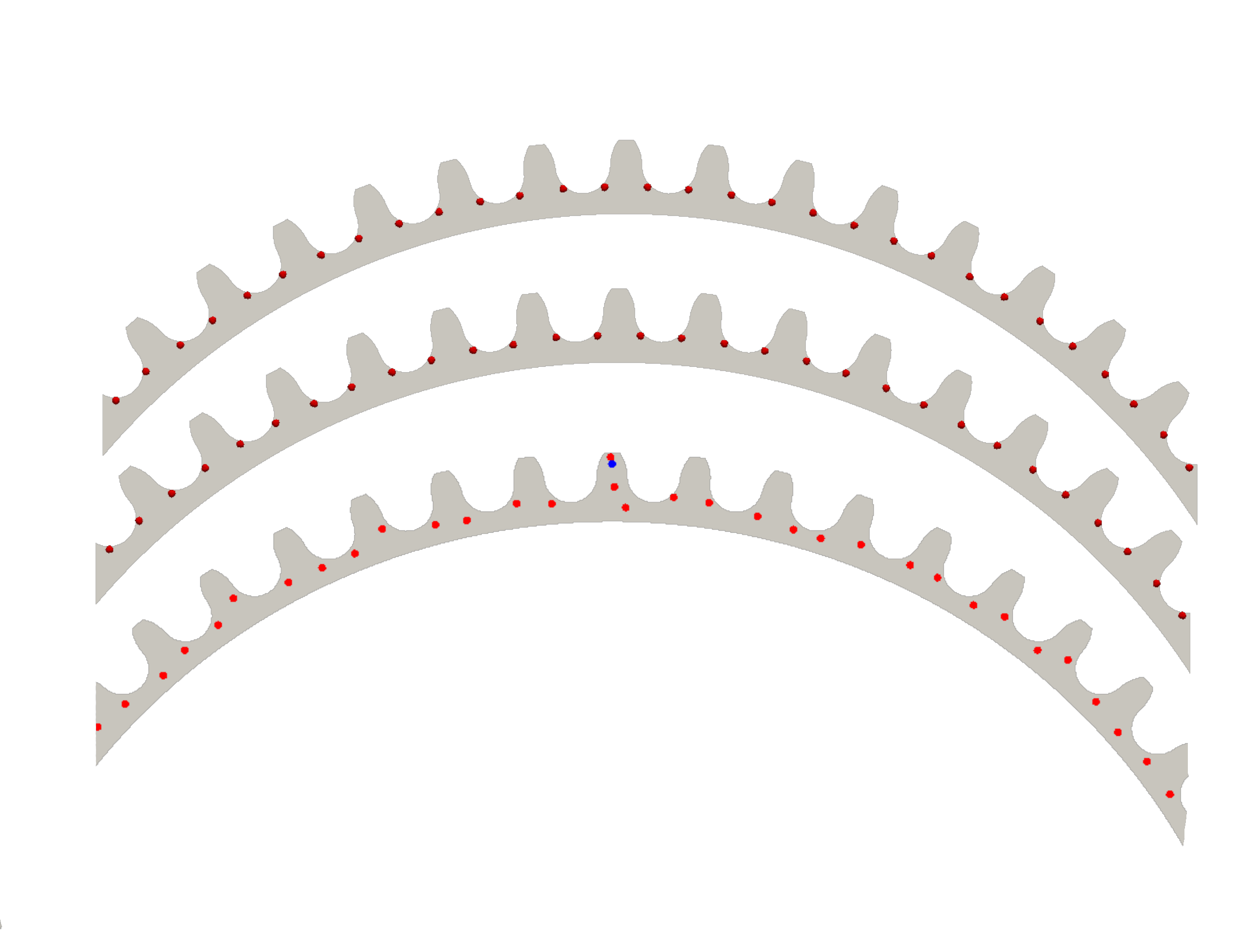}
  \end{center}
  \caption{The singularities of the gear model generated using \cref{alg:MBO} {\bf (top)}, the L-BFGS method {\bf (center)}, and instant meshes {\bf (bottom)}. See \cref{sec:comparison} for details.}
\label{fig:comparison}
\end{figure}

\begin{figure}
\begin{center}
\includegraphics[width=.48\linewidth]{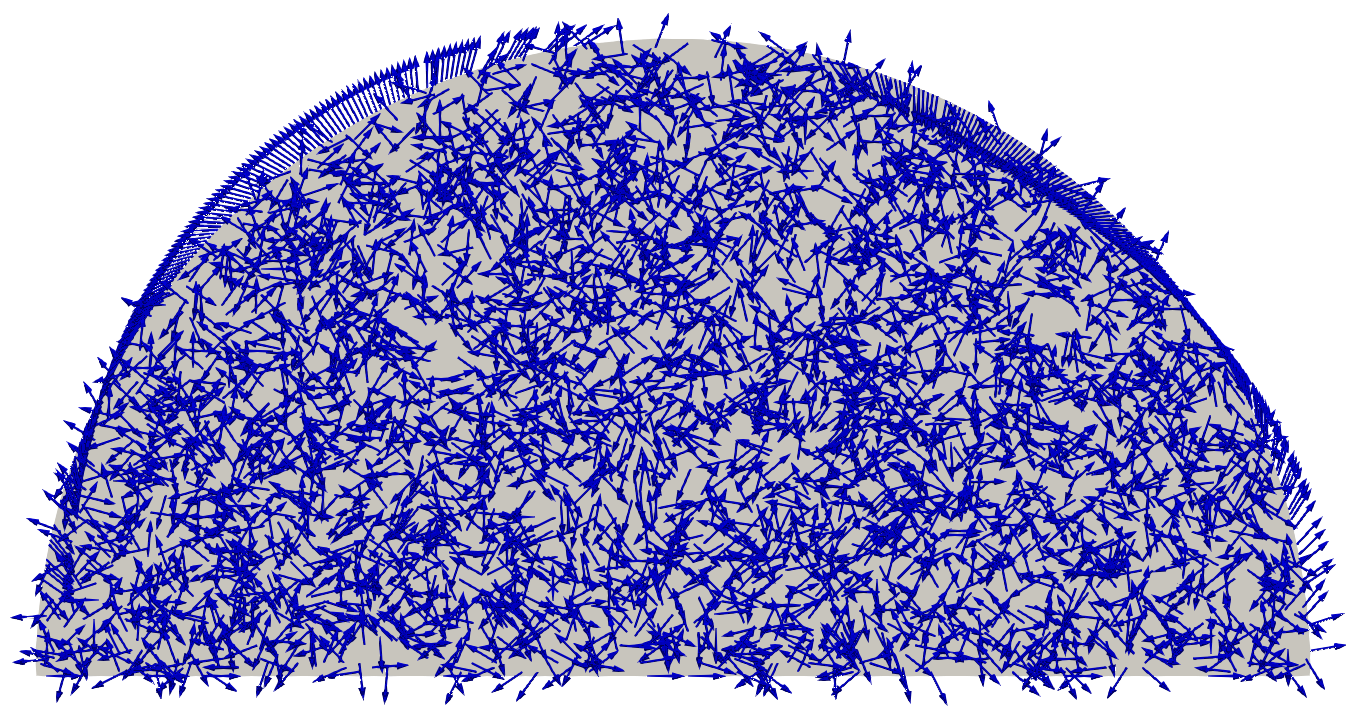}
\includegraphics[width=.48\linewidth]{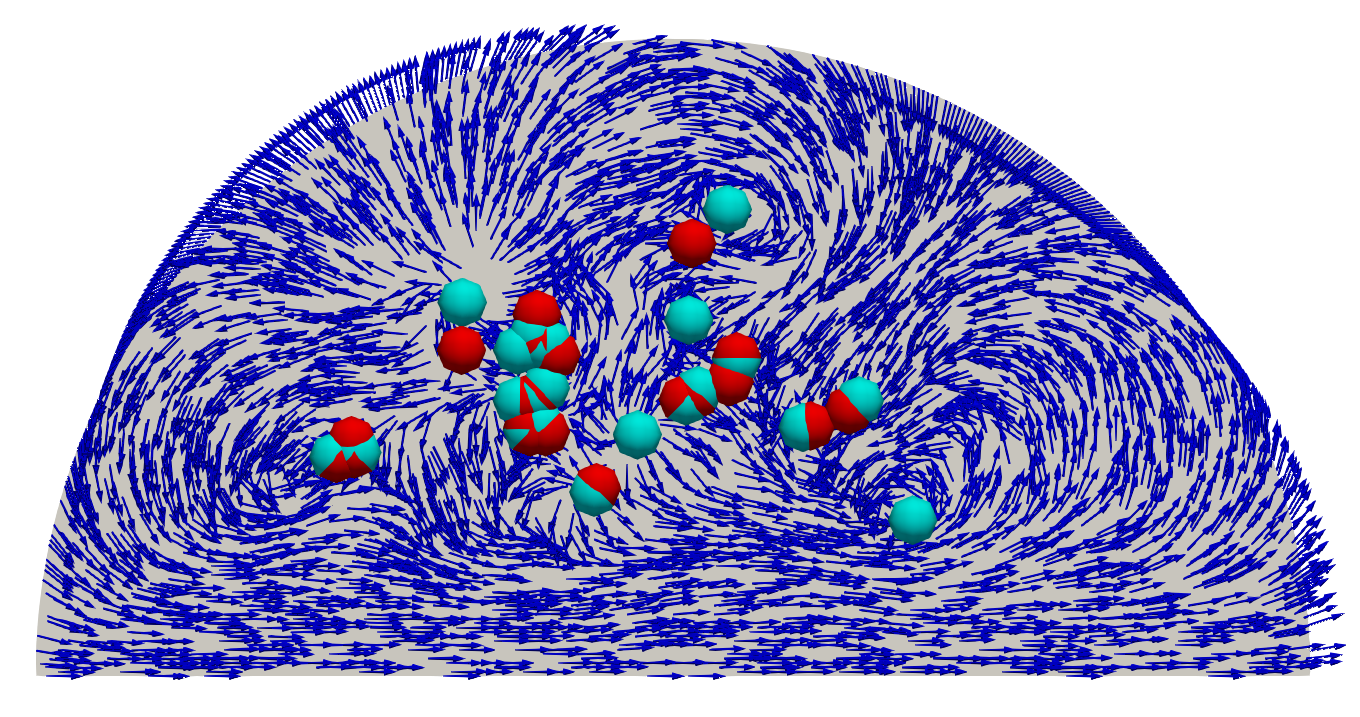} \\
\includegraphics[width=.48\linewidth]{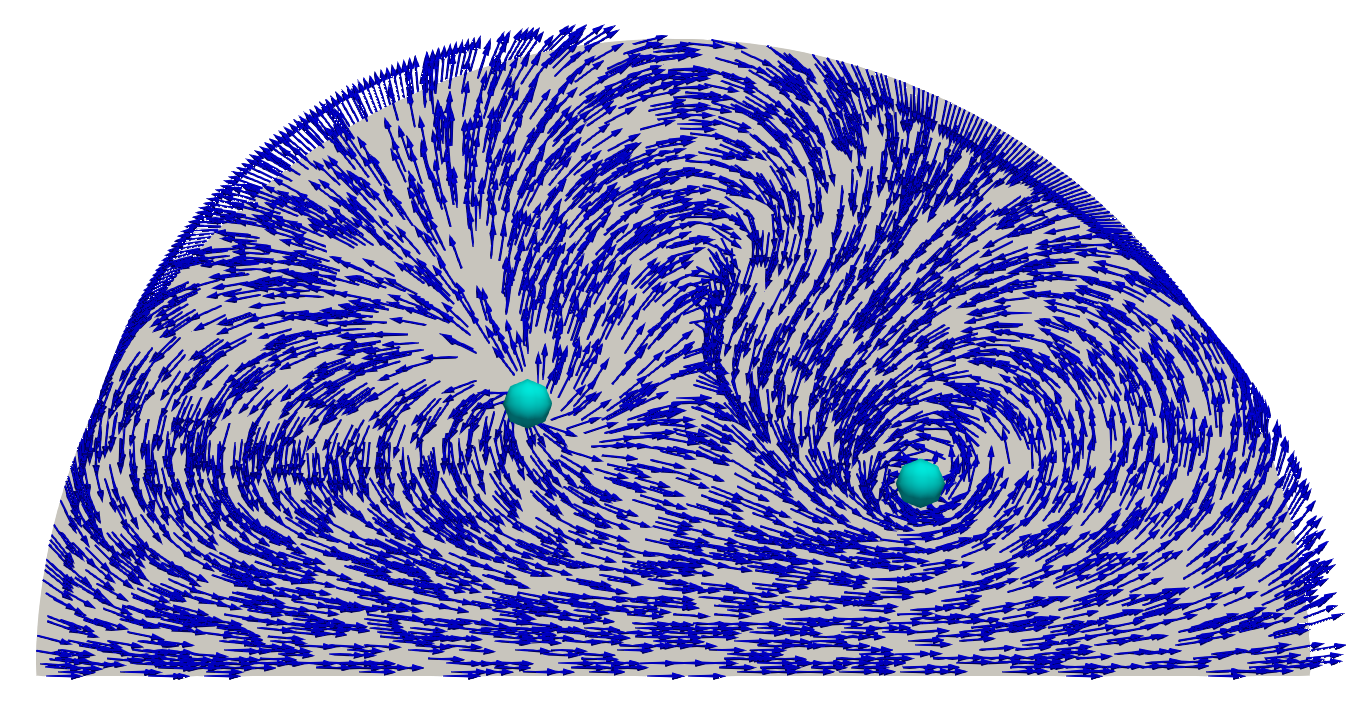}
\includegraphics[width=.48\linewidth]{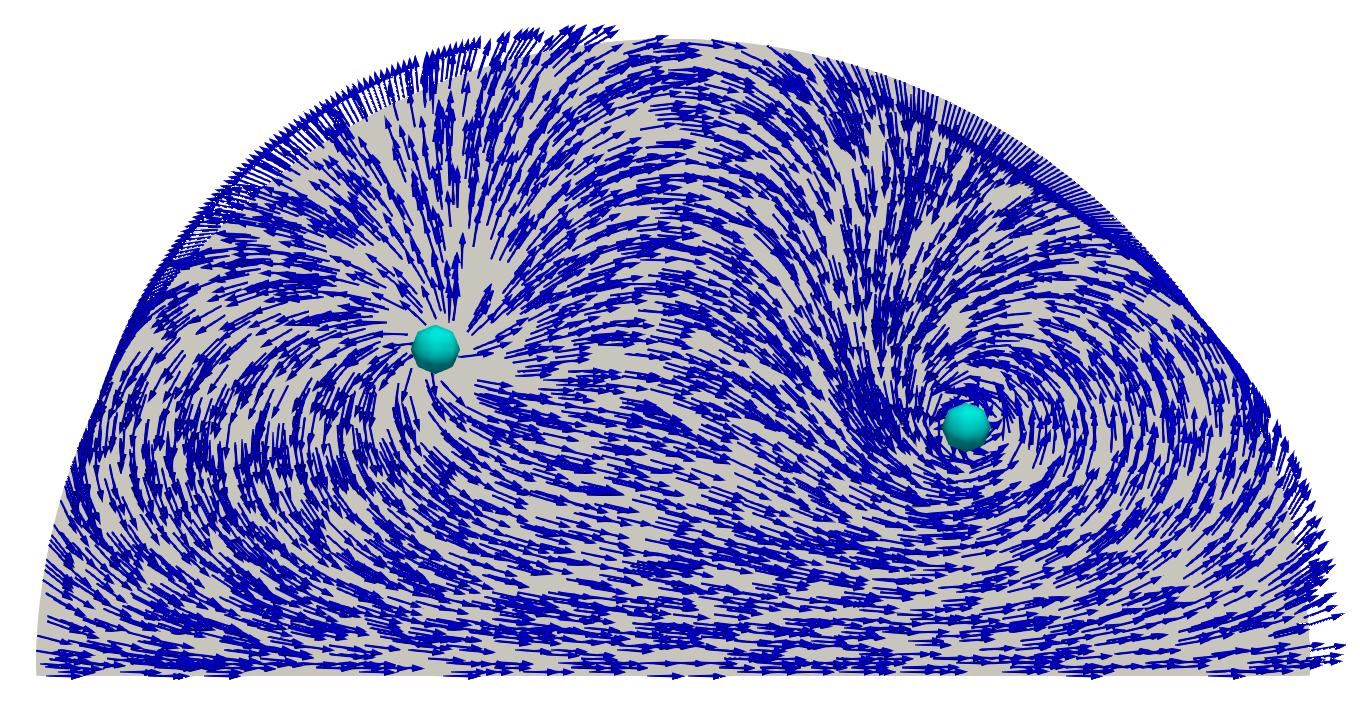} \\
\includegraphics[width=.48\linewidth]{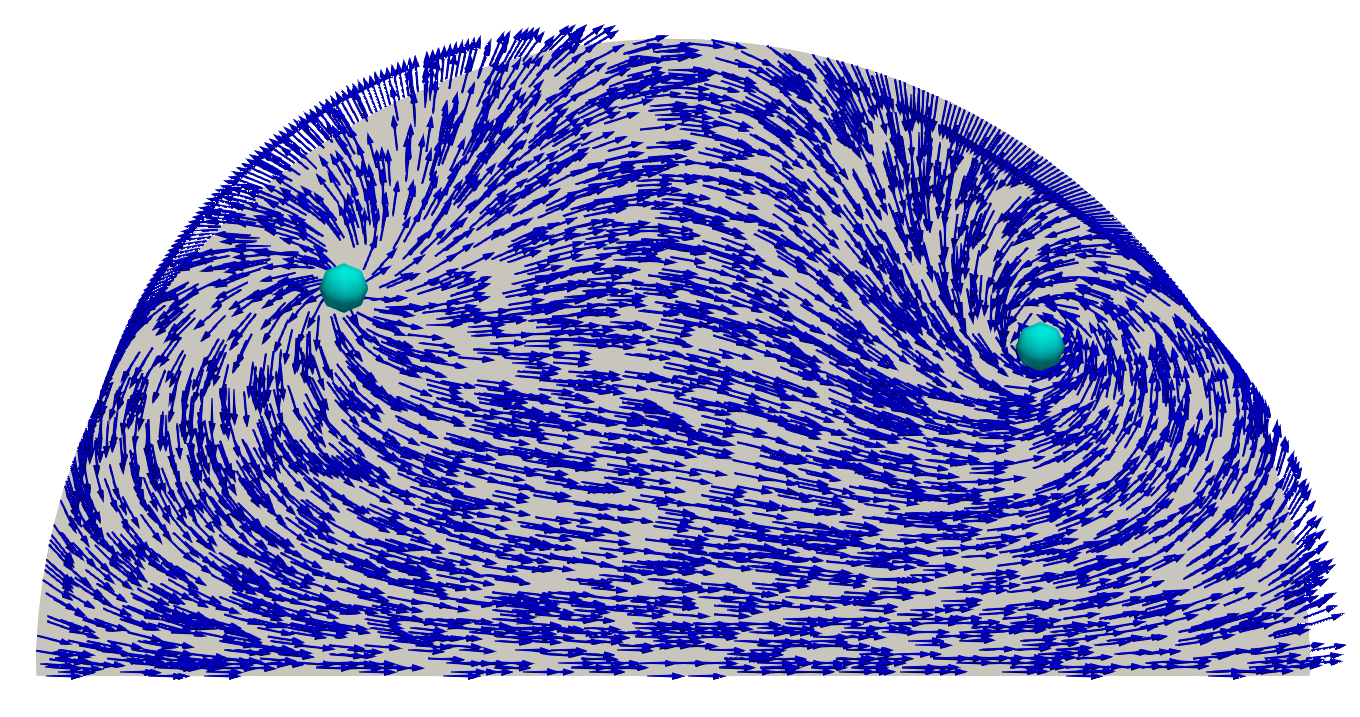}
\includegraphics[width=.48\linewidth]{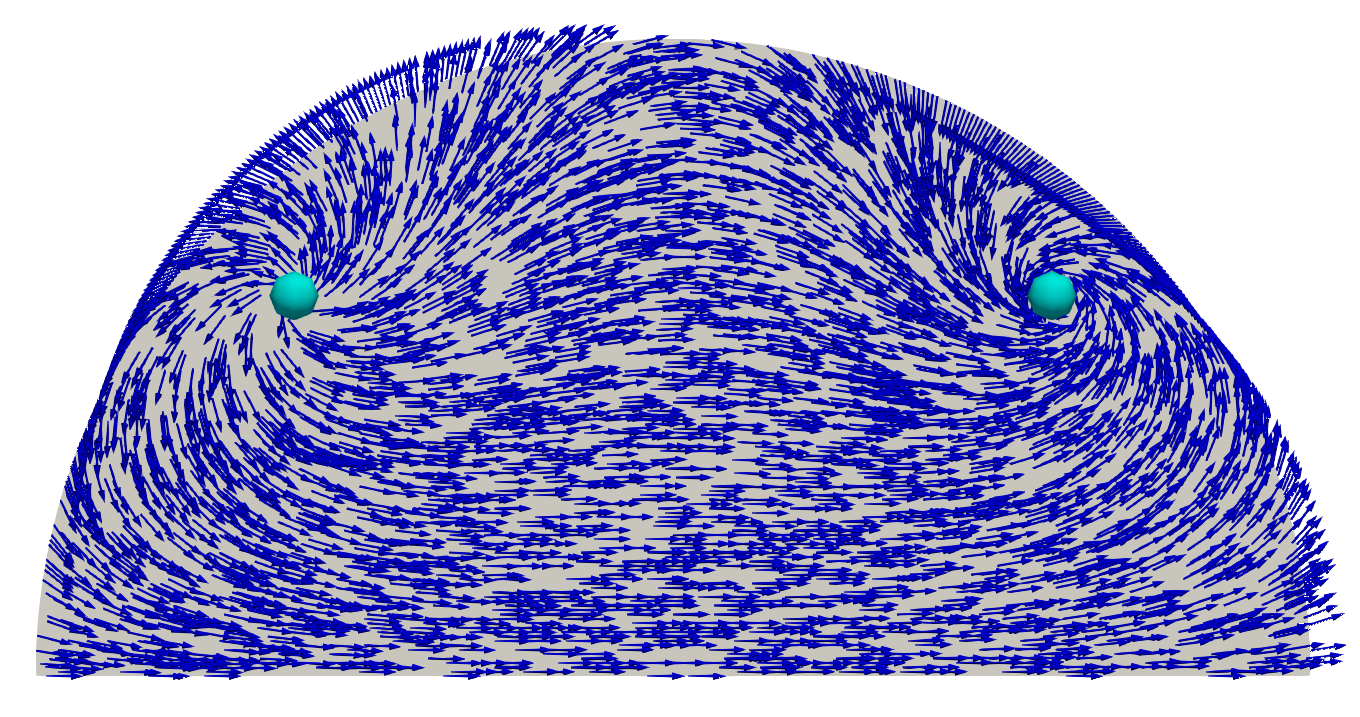} \\
\end{center}
\caption{A random initialization and iterations 1, 3, 10, 30, and 300 of \cref{alg:MBO} for a half disk. At each iteration the representation field is drawn. Singularities of index -1 are shown in red while singularities of index +1 are shown in cyan for each MBO iteration (not shown for the random initialization). As time evolves the singularities in the field move as to reduce the Ginzburg-Landau energy \cref{eq:GL_functional}. Details are included in \cref{sec:discrete}.}
\label{fig:MBO_repVec}
\end{figure}

\subsubsection{Creating a Quad Mesh from a Representation Field}\label{sec:meshing}
The discrete representation map produced by our implementation of the MBO method can be transformed into a discrete cross field by taking the fourth root of the representation map at any point. Singularities occur at the zeros of the representation map. The directions that separatrices exit a singularity are then computed by solving for locations on the boundary of a singular triangle where one of the cross vectors is pointing at the singularity location. Separatrices are then computed using a fourth order Runge-Kutta method and traced out as specified in \cref{alg:partitioning}. There is no need to compute a Riemann surface for the domain for streamline tracing because locally the appropriate vector field is completely determined. Separatrices are snapped to singularities if the distance to the singularity falls below a fixed tolerance. This method is identical to the one described in \cite{kowalski_pde_2013} except that we offer a way in which to treat separatrices converging to limit cycles. We apply it here primarily for the purposes of illustration, as neither the Runge-Kutta tracing method nor the separatrix snapping approach are robust for a variety of geometries. For more robust streamline based domain partitioning algorithms we refer the reader to \cite{campen_quantized_2015,myles_robust_2014}, and leave a robust implementation of \cref{alg:partitioning} as a subject of future research \cite{viertel_toward_2017}.

Following this approach, we obtain a partition of the domain into a quad layout with exactly one T-junction for every separatrix that approaches a limit cycle. For quad layouts that contain no T-junctions, a regular grid can be conformally mapped into each quad region to obtain the mesh. We do this using the CUBIT software \cite{_cubit_2017}.

When a T-junctions appears in a quad layout, a mesh can no longer be obtained by simply mapping a grid into each region. The problem is illustrated in \cref{fig:LimitCycle} (right). To mesh a region, the opposite sides of that region must have the same number of quads. If we were to map a regular grid into each subsequent region, we would have to satisfy the conditions $a + b = c$, $c = b$, $a > 0$, $b > 0$ and $c > 0$ where $a$ and $b$ are the number of quads required on the sides adjacent to the T-junction, and $c$ is the number of quads required on the opposite side; see \cref{fig:LimitCycle}. Clearly these conditions are incompatible, so any mesh in such a region will require additional singularities. In our examples, we mesh these regions using the paving algorithm \cite{blacker_paving:_1991}, however a more deterministic method such as \cite{takayama_pattern-based_2014} could be used.

These singularities do not necessarily need to be placed in the region adjacent to the T-junction as shown in \cref{fig:LimitCycle}, but can be distributed throughout the regions by assigning the number of quads to appear on each curve. This is a combinatorial problem similar to the user specified interval assignment problem considered in \cite{mitchell_high_2000}.

\subsection{Example Mesh Generation}
\Cref{fig:ExampleMeshes} shows several example meshes using this method. The first example is a surface from a CAD designed mechanical part \cite{_cubit_2017}. This domain has Brouwer degree two and the cross field contains two singularities of index $+1/4$. Three separatrices meet at each of these singularities, and thus the corresponding quad mesh has 3-valent nodes at the singularity locations.

The second example is a block U. The field was initialized with a canonical harmonic cross field with a boundary condition of Brouwer degree $-2$ and a singularity configuration placing two singularities of index $-1/4$ in the bottom corners of the U. The corresponding quad mesh has two 5-valent nodes.

The next two examples are multiply connected domains. These examples could be handled within the framework of this paper by cutting the domain or handled directly by applying the results in \cite{rubinstein_homotopy_1996}, however in practice we apply the MBO method directly to the multiply connected domain. The first multiply connected domain has Brouwer degree $-1$, which we count by subtracting the Brouwer degree of the interior boundaries from the exterior one. One singularity of index $-1/4$ appears near the curve that contributes to the negative Brouwer degree. The cross field in the last example contains two singularities with index $-1/4$, and also contains periodic orbits that intersect themselves.

\begin{figure}[t!]
\begin{center}
  \includegraphics[width=.32\linewidth, trim=400 0 400 0, clip]{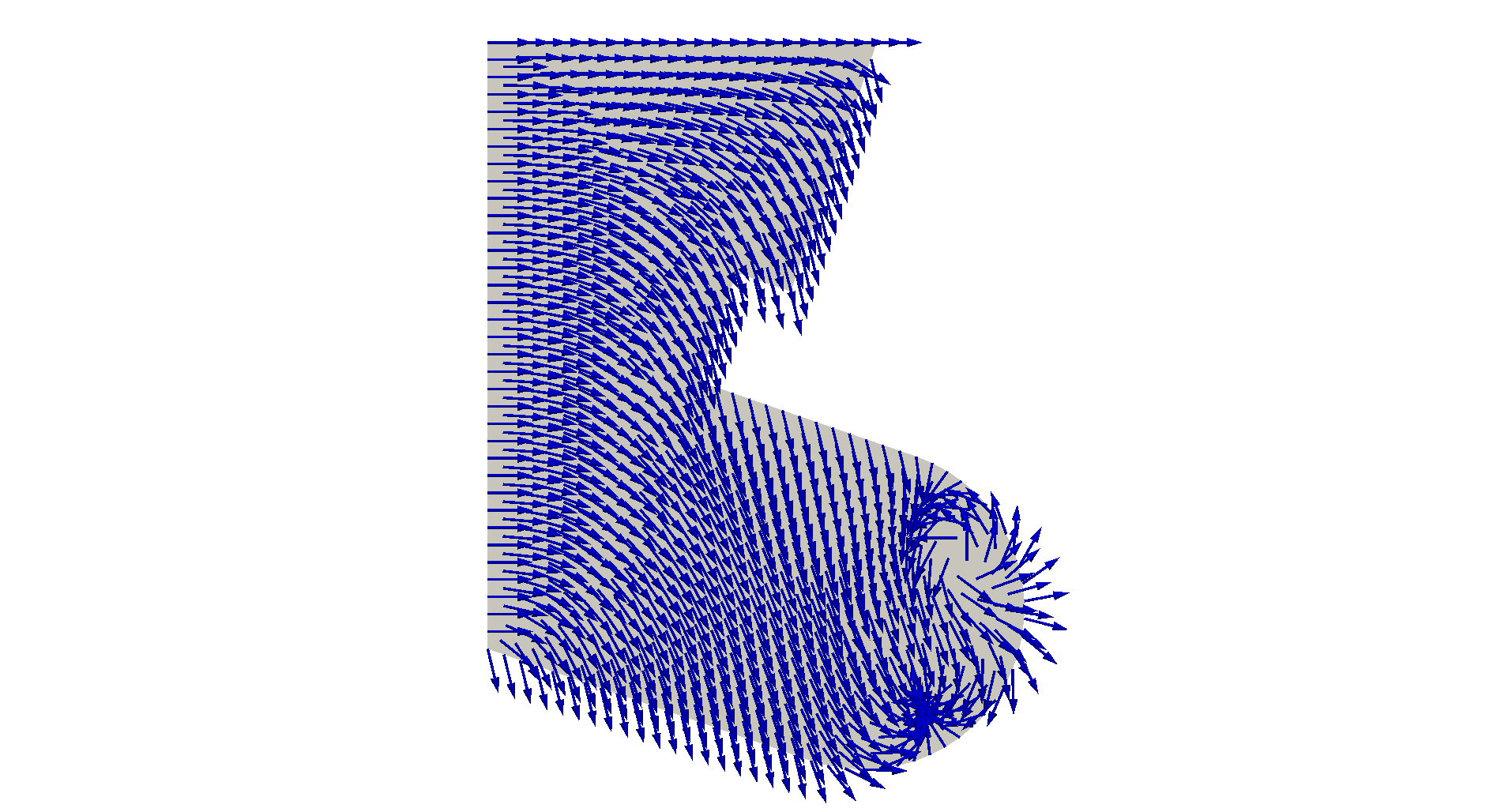}
  \includegraphics[width=.32\linewidth, trim=400 0 400 0, clip]{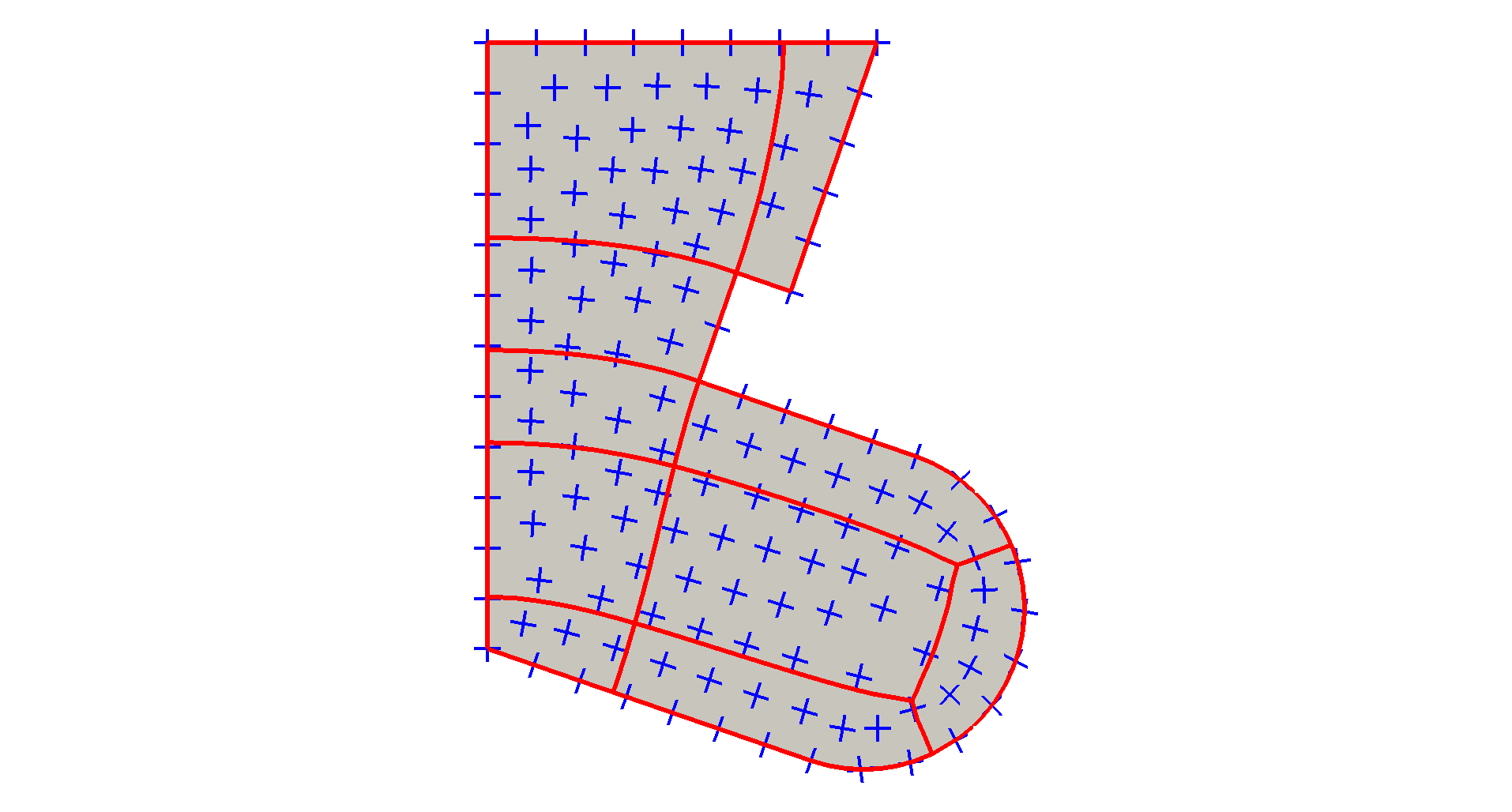}
  \includegraphics[width=.32\linewidth, trim=400 0 400 0, clip]{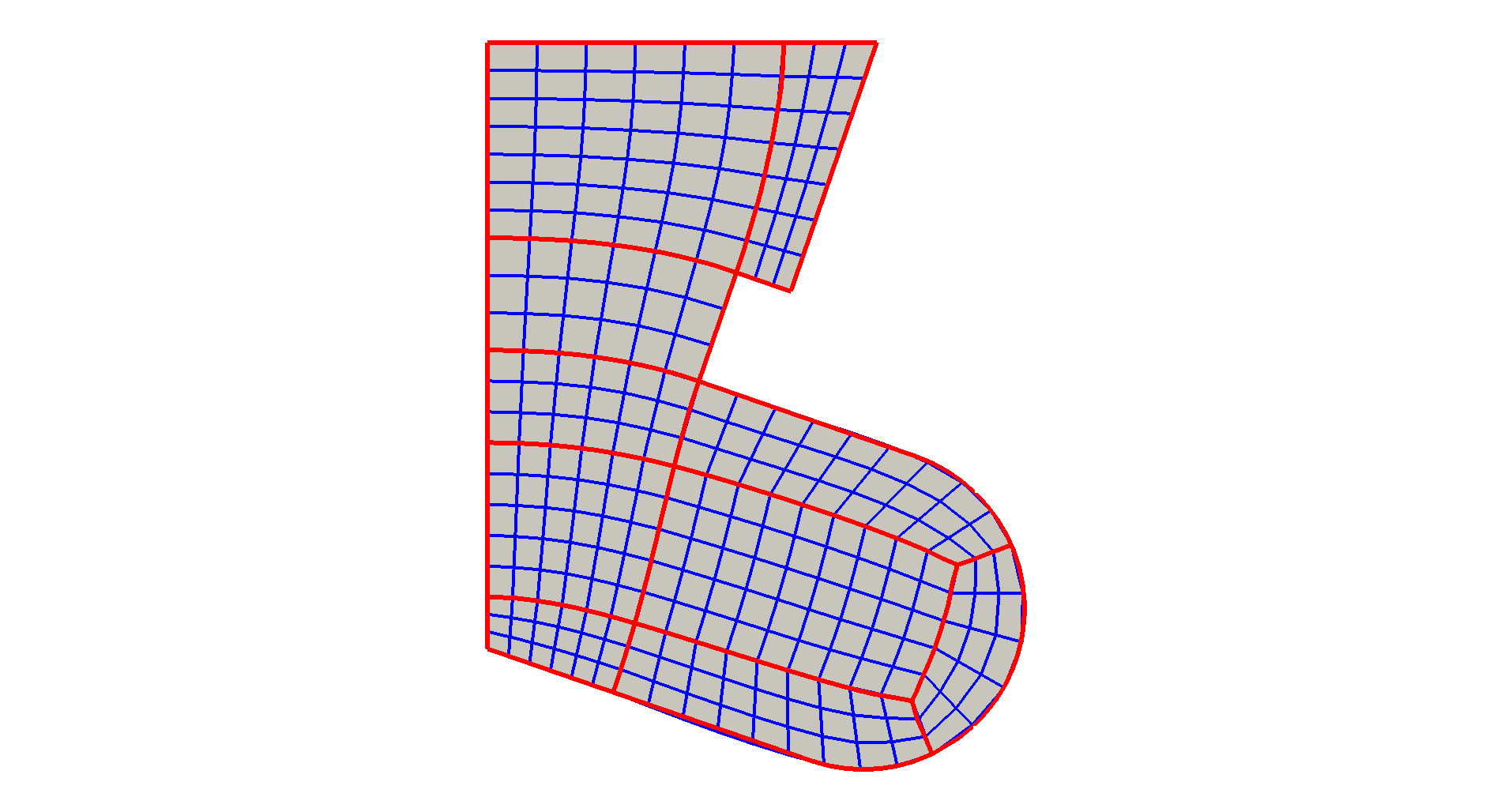} \\
\vspace{.8cm}
  \includegraphics[width=.32\linewidth, trim=270 0 270 0, clip]{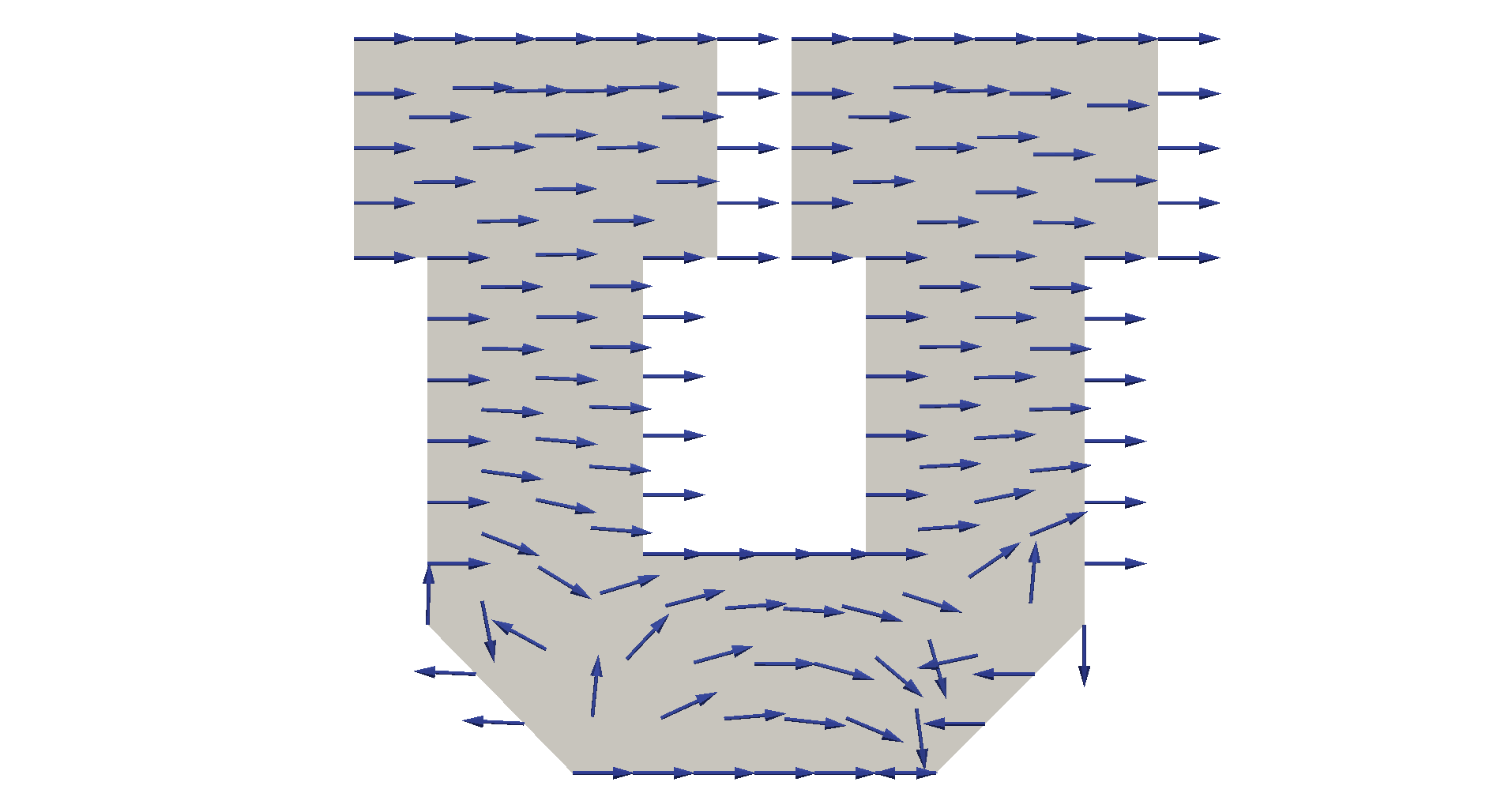}
  \includegraphics[width=.32\linewidth, trim=270 0 270 0, clip]{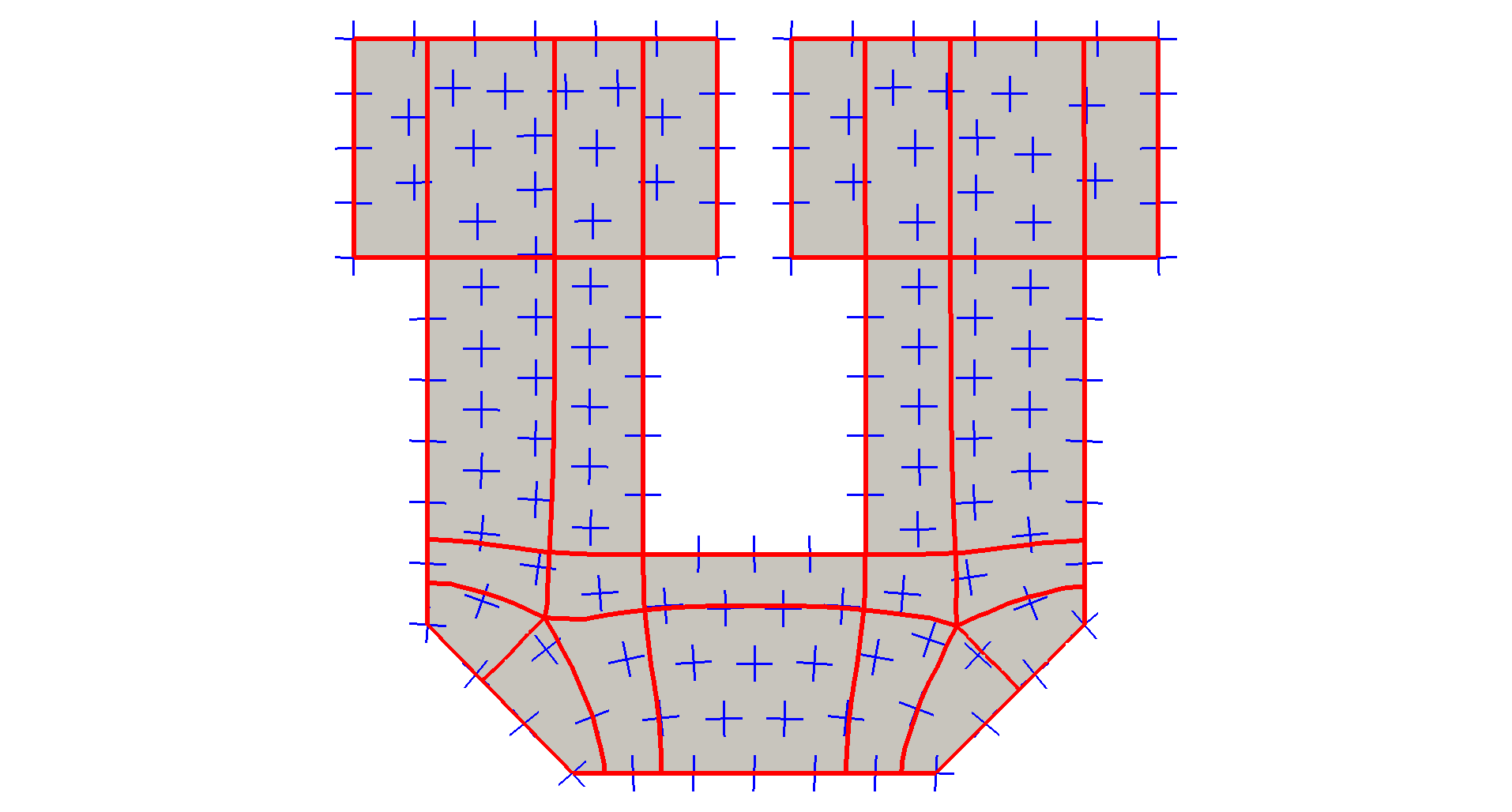}
  \includegraphics[width=.32\linewidth, trim=270 0 270 0, clip]{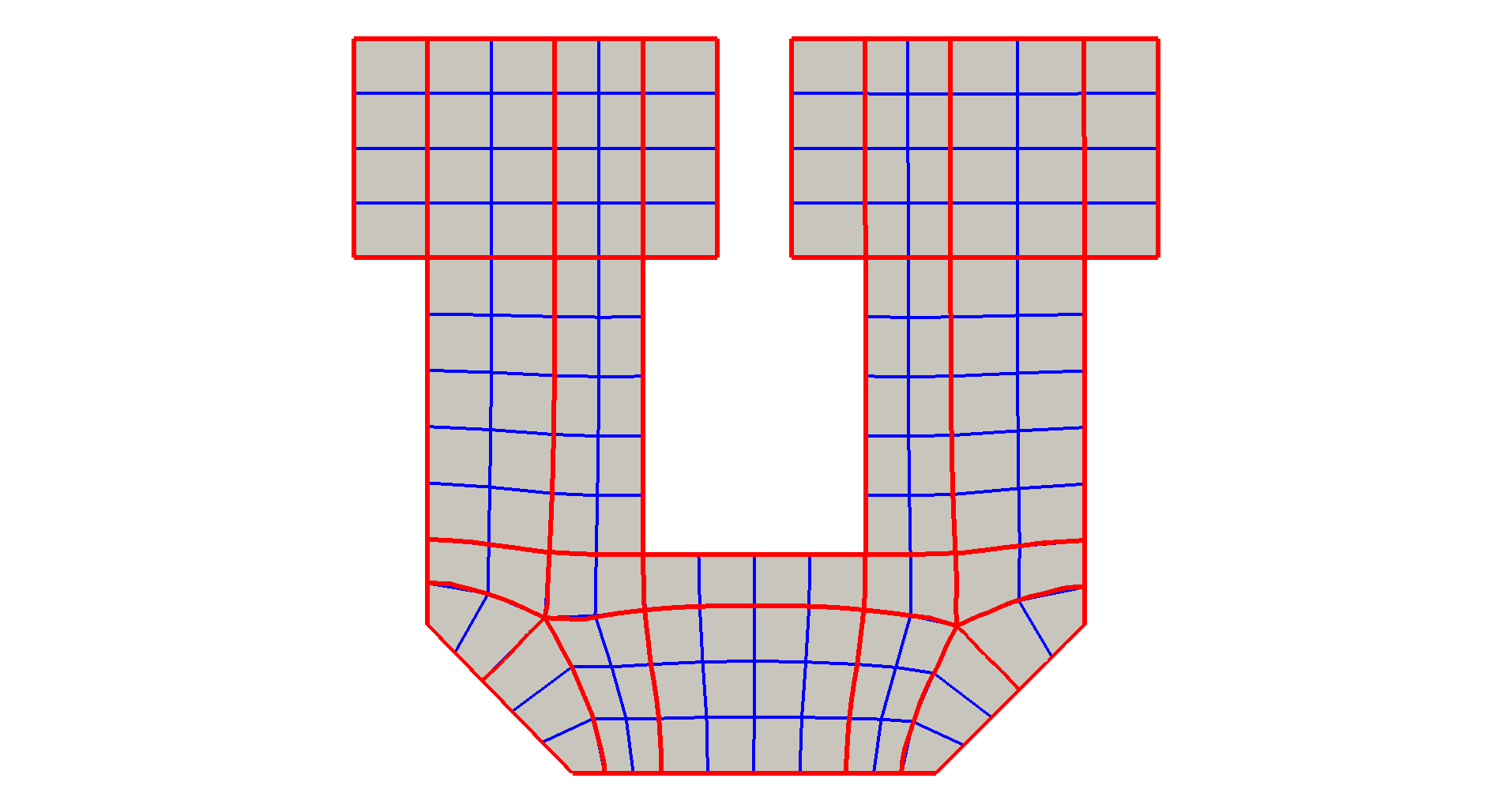} \\
\vspace{.6cm}
  \includegraphics[width=.32\linewidth, trim=250 0 250 0, clip]{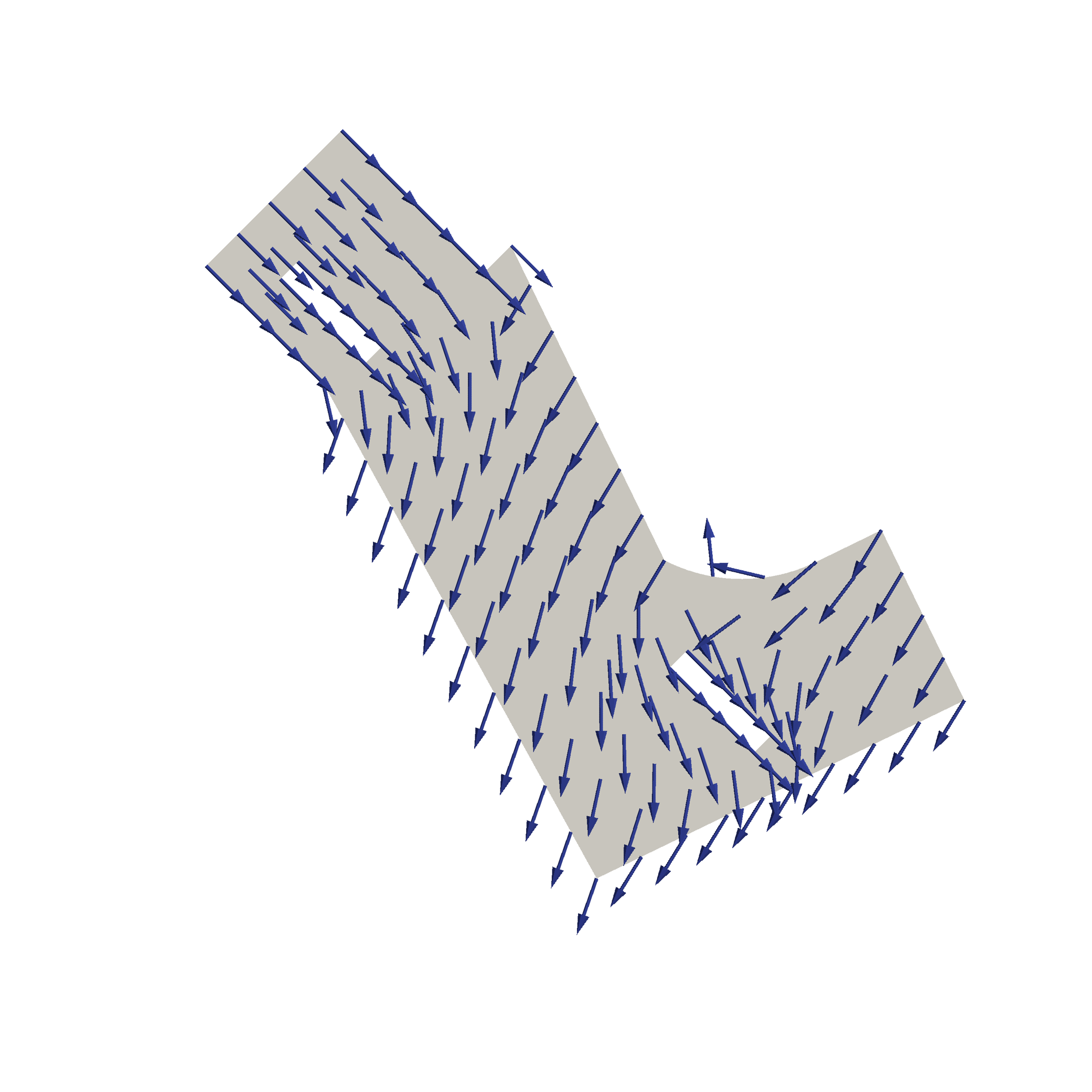}
  \includegraphics[width=.32\linewidth, trim=250 0 250 0, clip]{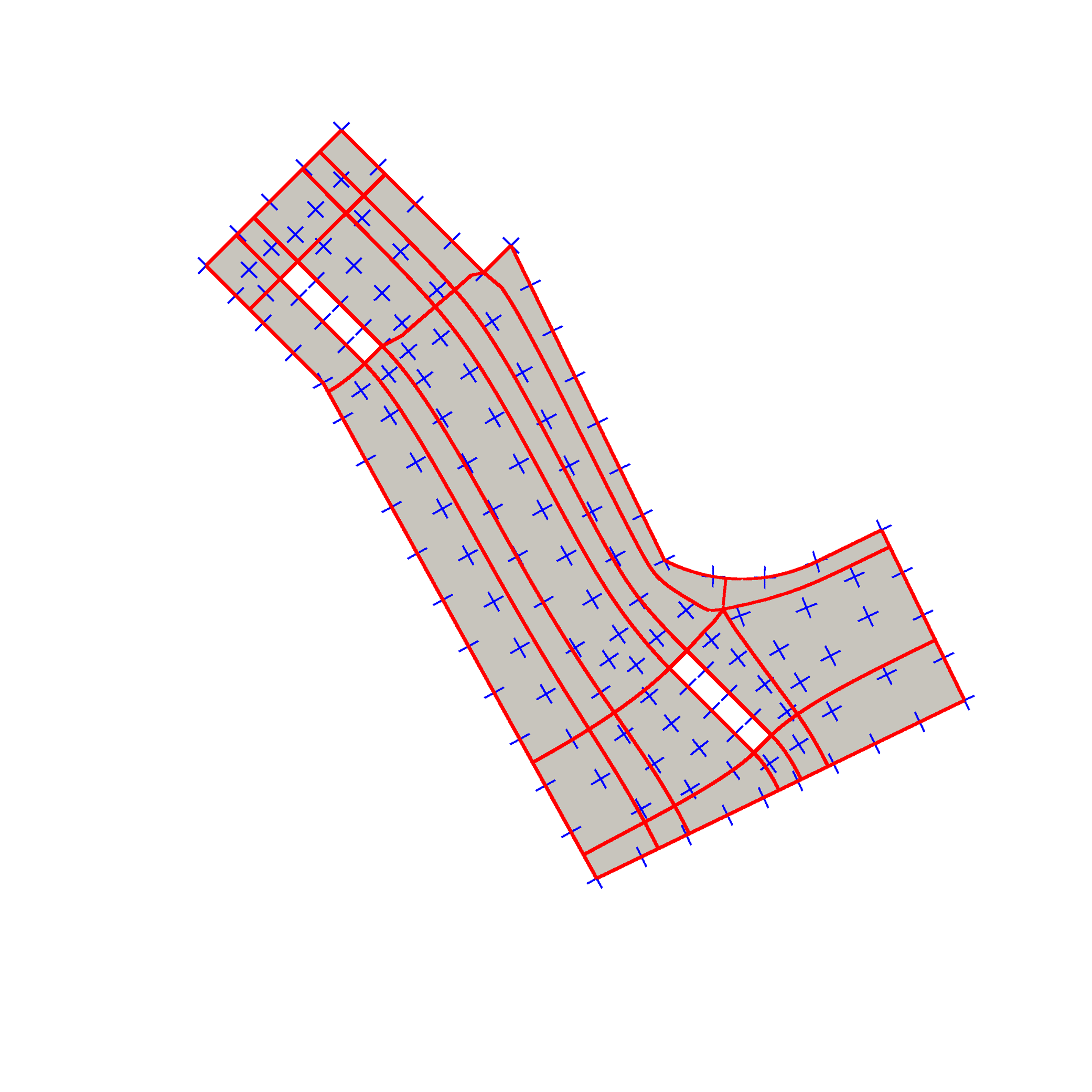}
  \includegraphics[width=.32\linewidth, trim=250 0 250 0, clip]{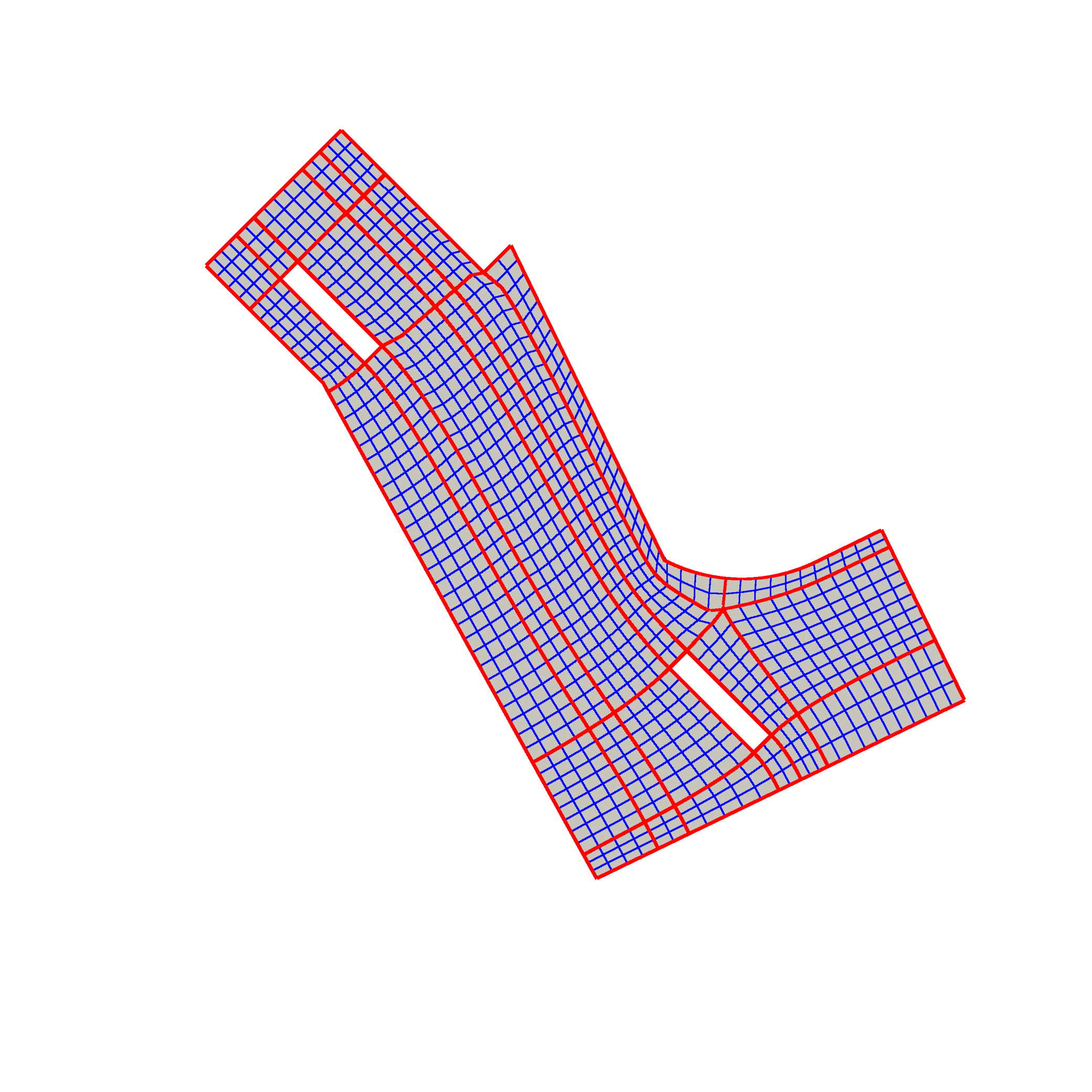} \\
\vspace{.1cm}
  \includegraphics[width=.32\linewidth, trim=500 0 500 0, clip]{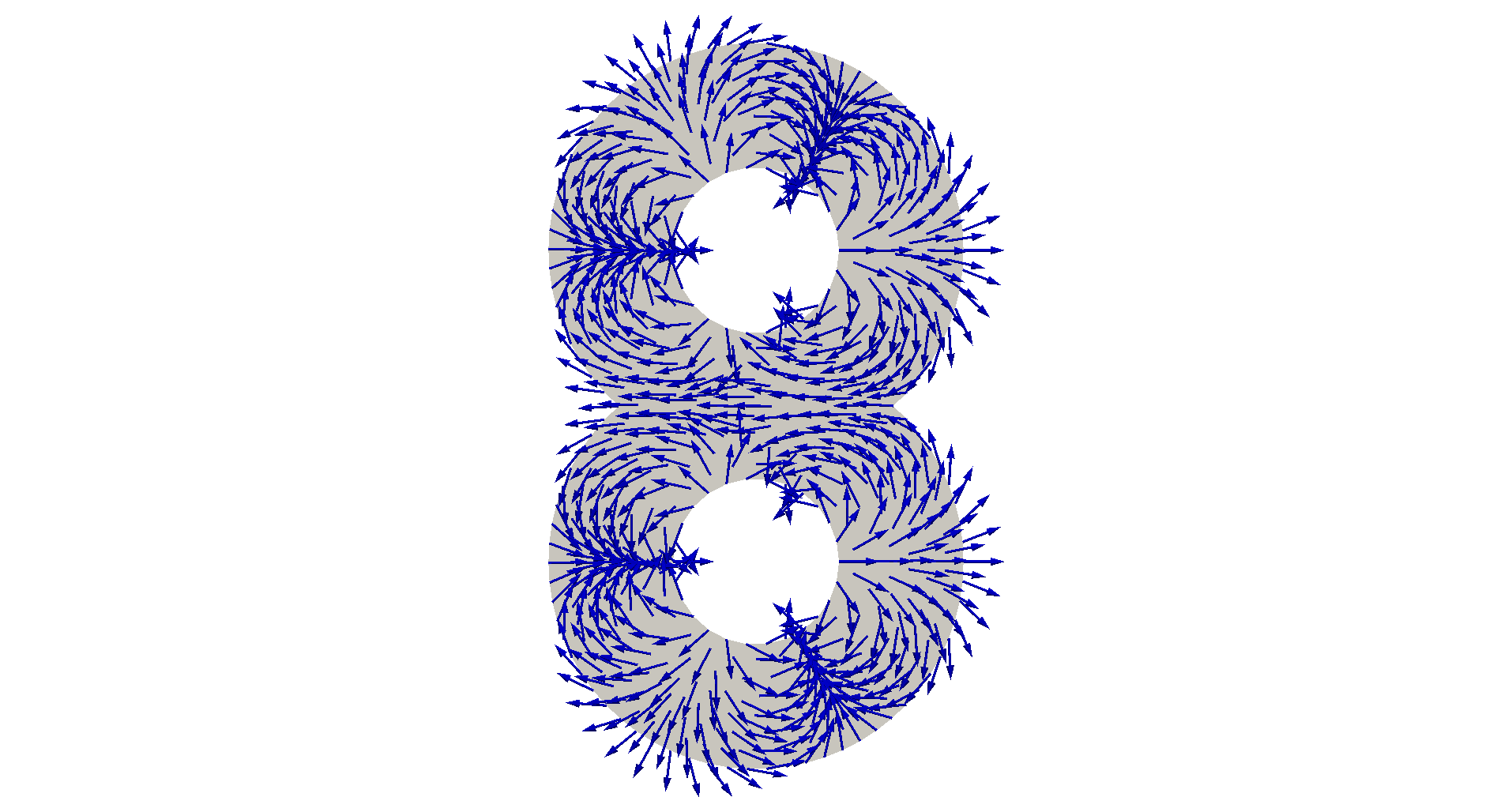}
  \includegraphics[width=.32\linewidth, trim=500 0 500 0, clip]{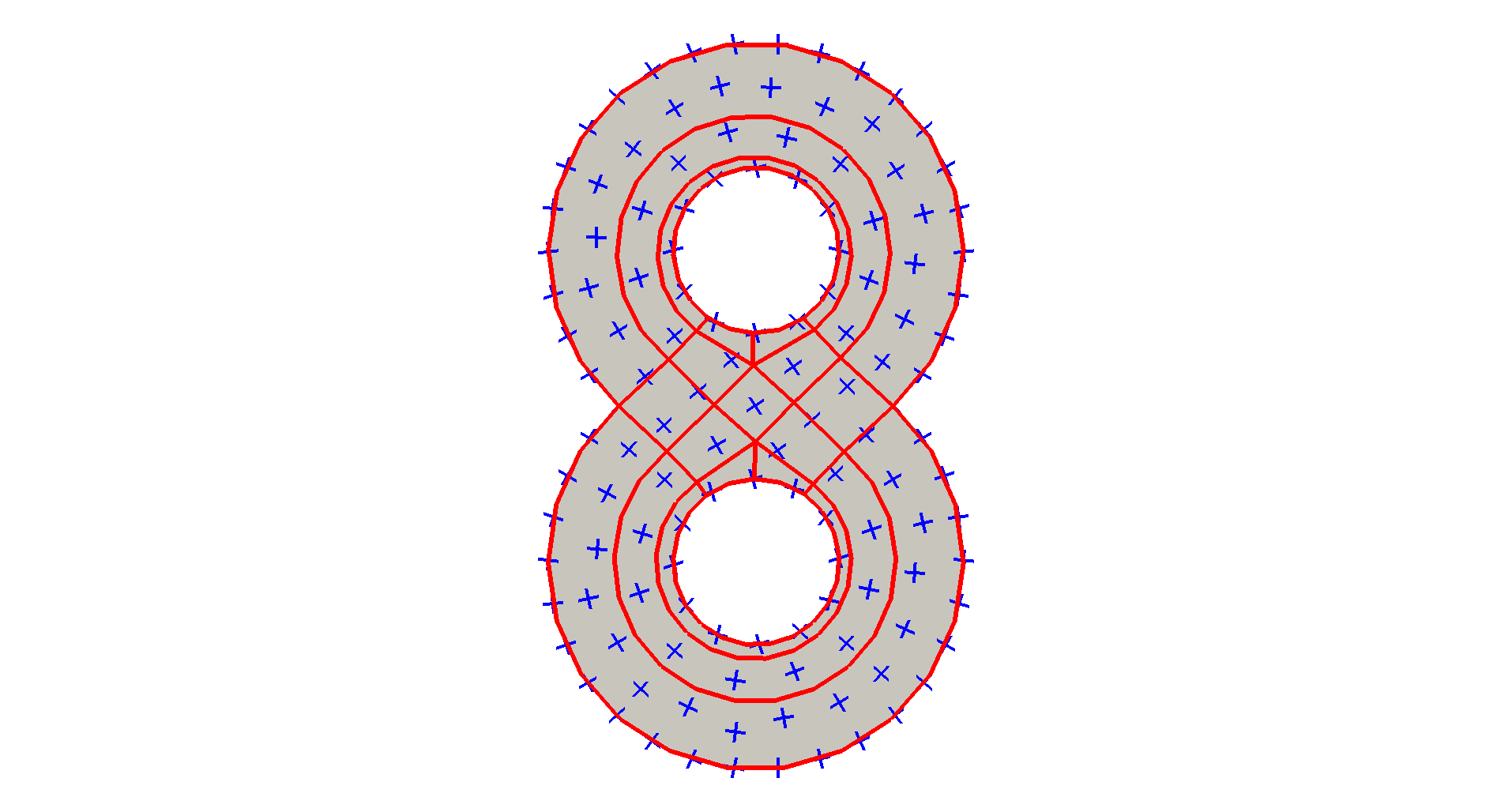}
  \includegraphics[width=.32\linewidth, trim=500 0 500 0, clip]{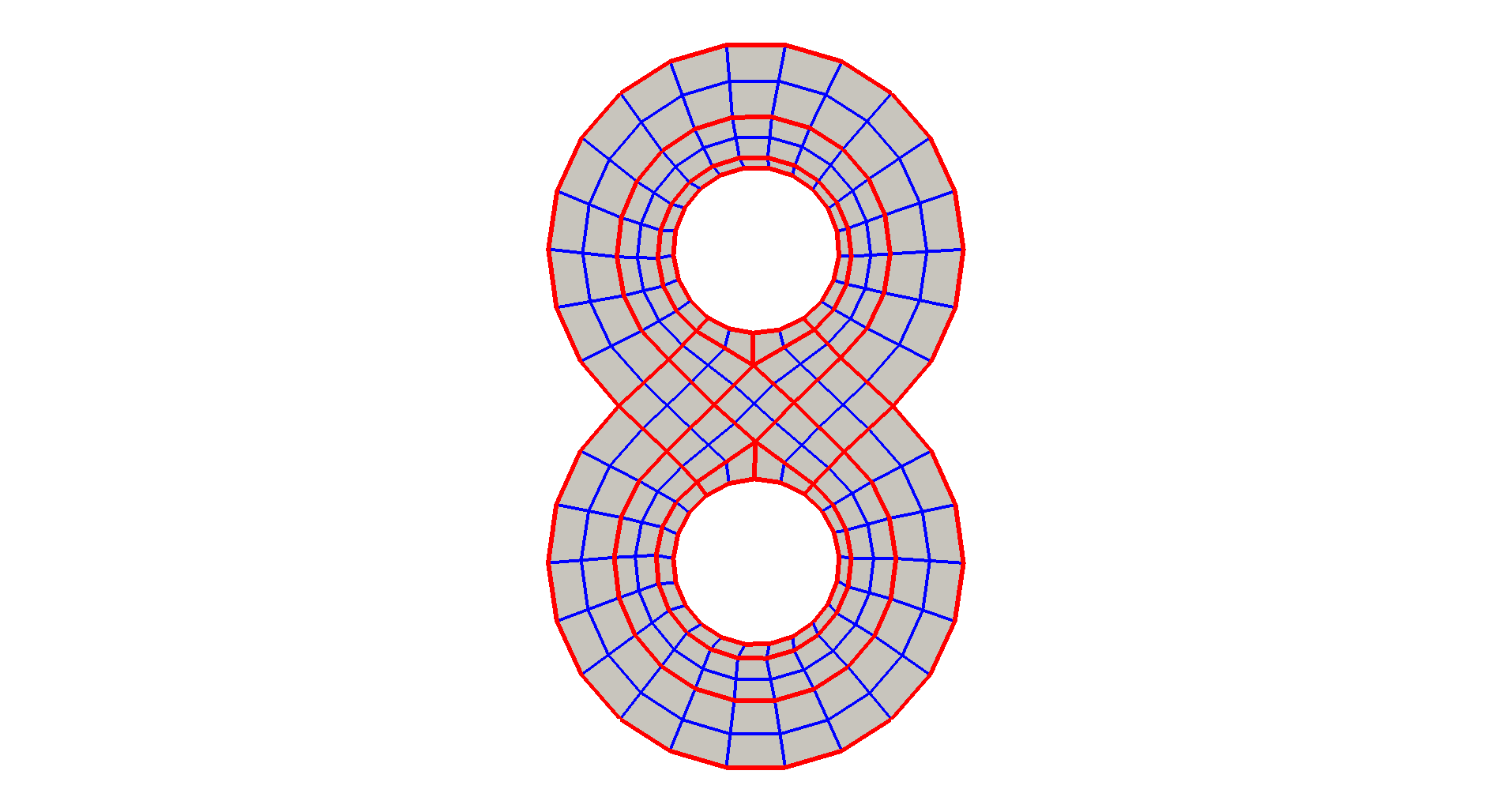}
\end{center}
\caption{ For several different geometries (rows), we  plot the
{\bf (left)} representation field obtained via the MBO method (\cref{alg:MBO}),
{\bf (center)} the cross field and quad layout obtained from the separatrices of the cross field, and
{\bf (right)} quad mesh with skeleton drawn in red.}
\label{fig:ExampleMeshes}
\end{figure}

\section{Discussion and Future Directions}\label{sec:discussion}
In this paper, we have made the observation that cross field design for two-dimensional quad meshing is related to the well-known Ginzburg-Landau problem from mathematical physics. Using this identification, we prove that this procedure can be used to produce a cross field whose separatrices partition the domain into four sided regions.  This identification also allows for an extension of the Merriman-Bence-Osher (MBO) threshold dynamics method to be used to find representation fields that approximately minimize the  Ginzburg-Landau energy. The methods are demonstrated with a number of computational examples.

Some limitations exist when using the energy \cref{eq:energy}. Since the problem is non-convex, for most domains, we cannot guarantee that we will reach a global minimum. Ironically, a global minimum may not always result in the best mesh. For example, the cross field in \cref{fig:mushroom} (left) is the global minimizer of the Ginzburg-Landau energy for this domain because it is the canonical harmonic cross field with no singularities. The one shown in \cref{fig:mushroom} (center) is the cross field obtained through the MBO method. Isotropy of mesh elements is a desirable property in many meshing applications, and in this case the local minimizer found by the MBO method produces a more isotropic mesh than the global minimizer because of the additional singularities. This suggests that this definition of cross field energy may assign too much weight to singularities. The infinite energy at singularities is also problematic because the discrete measure of field energy depends on the discretization. As the mesh is refined near singularities, the energy is increased.

Despite these limitations, we have shown that approximate solutions to problem \cref{eq:energy} has many desirable properties for meshing. For example, the MBO method produces cross fields with isolated singularities of index $\pm 1/4$ (\cref{sec:renormalized}) that locally exhibit the same structure as irregular mesh nodes \cref{sec:behavior}. Further, the separatrices of these cross fields are guaranteed to partition a domain into a quad layout, possibly with T-junctions.

There are a number of future directions for this work. One direction is a more careful numerical analysis of the finite element problem arising in the discretization of \cref{eq:GL_functional}. Along the same line, a more careful comparison should be made with the quad meshing methods in \cite{jakob_instant_2015,jiang_frame_2015,knoppel_globally_2013,kowalski_pde_2013,ray_periodic_2006}. Also, a comparison between the location of singularities of fields generated via this method and those generated via the method in \cite{knoppel_globally_2013}, utilizing a different cross field energy, is necessary.

In \cref{sec:definitions} we discuss a method to smooth the corners of a piecewise smooth domain so as to be able to apply the Ginzburg-Landau theory. A more elegant solution would be to extend the Ginzburg-Landau theory to handle domains with piecewise smooth boundary.

In this paper, the index of a boundary singularity is determined by the corner smoothing operation in \cref{sec:definitions}. While this is a natural choice, it does not allow us to mesh geometries with sharp corners such as that in \cref{fig:degenerate-singularities} (left), since a boundary singularity of index $1/2$ would be assigned. While quad elements are more isotropic when the index values are chosen close to $\pi - \interior(c)/2\pi$, in reality, there is ambiguity in the index assignment of any corner. For example, the bottom four corners on the block U in \cref{fig:ExampleMeshes} were each assigned an index of $1/4$, but an index of zero would have been just as reasonable, and would change the resulting cross field and mesh. It may be preferable to let the user have control of the index assignment for each corner to allow for greater flexibility for the mesh and enable meshing of geometries with sharp corners. This is often as simple as modifying the assignment of the boundary cross on a singular corner.

Proving general symmetry results for solutions of the Ginzburg-Landau energy is a difficult problem. In \cite{LiebLoss1994}, a the symmetric solution on a disc is analyzed and is shown to be stable under perturbations. See also \cite{Brezis1999} for further discussion and open problems. However, we observe in numerous numerical examples that the symmetries of a domain are inherited by the configuration of the Ginzburg-Landau vortices. Analytical results in this direction \emph{for the specific boundary conditions considered here} could be used to guarantee symmetries in the resulting mesh.

For simplicity, in this paper we have restricted to \emph{simple planar Euclidean domains} and not \emph{surfaces} with boundary. However, due to the topological nature of our results and the Ginzburg-Landau theory, we expect that many of these results can be extended to surfaces \cite{sternberg_dynamics_2013}. Further, the problem of designing a cross field aligned to a prescribed set of orthogonal directions, such as principle curvature directions, is often desired on curved surfaces. This is typically achieved by adding a least-squares fitting term to the energy. While we are not aware of any results from the classical Ginzburg-Landau theory that treat this extension of the energy, the MBO method can be adapted to this setting; see \cite{Osting_Wang_2018}.

In three dimensions, the analogous approach is to minimize the Dirichlet energy over the set of $H^1$ functions taking values in $SO(3) / O$ where $O$ is the octahedral symmetry group \cite{huang_boundary_2011,ray_practical_2016,solomon_boundary_2017}. Unfortunately, this approach is more complicated as the field topology is no longer sufficient to determine the underlying structure of a hex mesh \cite{viertel_analysis_2016}. Finding an efficient representation for elements of the set $SO(3) / O$ and studying the singularity structure for generalized solutions of this problem requires further attention.

\section*{Acknowledgments}
We would like to thank Dan Spirn and Matthew Staten for helpful discussions and comments, and Franck Ledoux for providing access to the GMDS meshing library. We would also like to thank the referees for their valuable comments.


\bibliographystyle{siamplain}
\bibliography{frame-fields}
\end{document}